\documentclass[reprint,twocolumn,amsmath,amssymb,superscriptaddress]{revtex4-2}

\usepackage{graphicx, graphpap}
\usepackage[dvipsnames]{xcolor}
\usepackage[caption=false,subrefformat=parens,labelformat=parens]{subfig}
\usepackage{amsmath}
\usepackage{overpic}
\usepackage{physics}
\usepackage{enumitem}
\usepackage{comment}
\usepackage{amssymb}
\usepackage{amsthm}
\usepackage{pstricks}
\usepackage{float}
\usepackage{multirow}
\usepackage{hyperref}
\usepackage[T1]{fontenc}
\usepackage{framed} 
\usepackage{bm}
\usepackage{bbm}
\definecolor{shadecolor}{gray}{0.9}
\definecolor{ECCgreen}{RGB}{0, 153, 0}
\definecolor{EditPurple}{RGB}{160, 32, 240}

\renewcommand{\Tr}[1]{\mathrm{Tr}\left[#1\right]}

\usepackage{placeins}

\newtheorem{theorem}{Theorem}
\newtheorem{corollary}[theorem]{Corollary}
\newtheorem{lemma}[theorem]{Lemma}
\newtheorem{proposition}[theorem]{Proposition}
\newtheorem{definition}[theorem]{Definition}

\begin{document}
	
\title{Finite-Size Security for Discrete-Modulated Continuous-Variable Quantum Key Distribution Protocols}
\author{Florian Kanitschar}
\email{florian.kanitschar@outlook.com}
\affiliation{Institute for Quantum Computing and Department of Physics and Astronomy,
University of Waterloo, Waterloo, Ontario, Canada N2L 3G1}
 \affiliation{Technische Universität Wien, Faculty of Mathematics and Geoinformation, Wiedner Hauptstraße 8, 1040 Vienna, Austria}
 \author{Ian George}
  \affiliation{Institute for Quantum Computing and Department of Physics and Astronomy,
University of Waterloo, Waterloo, Ontario, Canada N2L 3G1}
	\affiliation{Department of Electrical \& Computer Engineering, University of Illinois at Urbana-Champaign, Urbana, Illinois 61801, USA}

\author{Jie Lin}
  \affiliation{Institute for Quantum Computing and Department of Physics and Astronomy,
University of Waterloo, Waterloo, Ontario, Canada N2L 3G1}
\affiliation{Department of Electrical\&Computer Engineering, University of Toronto, Toronto, Ontario M5S 3G4, Canada}

\author{Twesh Upadhyaya}
\affiliation{Institute for Quantum Computing and Department of Physics and Astronomy,
University of Waterloo, Waterloo, Ontario, Canada N2L 3G1}

	\author{Norbert Lütkenhaus}
  \affiliation{Institute for Quantum Computing and Department of Physics and Astronomy,
University of Waterloo, Waterloo, Ontario, Canada N2L 3G1}

\date{\today}
	
\begin{abstract}
Discrete-Modulated (DM) Continuous-Variable Quantum Key Distribution (CV-QKD) protocols are promising candidates for commercial implementations of quantum communication networks due to their experimental simplicity. While tight security analyses in the asymptotic limit exist, proofs in the finite-size regime are still subject to active research. We present a composable finite-size security proof against independently and identically distributed collective attacks for a general DM CV-QKD protocol. We introduce a new energy testing theorem to bound the effective dimension of Bob's system and rigorously prove security within Renner's $\epsilon$-security framework and address the issue of acceptance sets in protocols and their security proof. We want to highlight, that our method also allows for nonunique acceptance statistics, which is necessary in practise. Finally, we extend and apply a numerical security proof technique to calculate tight lower bounds on the secure key rate. To demonstrate our method, we apply it to a quadrature phase-shift keying protocol, both for untrusted, ideal and trusted nonideal detectors. The results show that our security proof method yields secure finite-size key rates under experimentally viable conditions up to at least $72$~km transmission distance.
\end{abstract}
	
\maketitle


\section{Introduction \label{sec:Intro}} 
Quantum key distribution (QKD) \cite{Bennett_Brassard_1984, Ekert_1991} enables two remote parties to establish an information-theoretically secure key, even in the presence of an eavesdropper, which is known to be impossible by classical means. The generated key can then be used in cryptographic routines like the one-time pad. Comprehensive reviews about QKD can be found in \cite{Scarani_2009,Diamanti_2015,Pirandola_2020}. Depending on the used detection technology, we distinguish between discrete-variable (DV) protocols like the famous BB84 \cite{Bennett_Brassard_1984} and protocols with continuous-variables (CVs) \cite{Ralph_1999}. While the first class relies on rather expensive components like single-photon detectors, the latter ones make use of state-of-the-art communication infrastructure and employ much cheaper photodiodes to perform homodyne or heterodyne measurements. In contrast to CV QKD being easier to implement compared to DV QKD, proofs of security for CV-QKD are often more difficult to establish as the physical systems are described by infinite dimensional Hilbert spaces. Based on the modulation type, CV-QKD can be further subdivided into protocols with Gaussian modulation (GM) \cite{Cerf_2001, Grosshans_2002, Grangier_2002,Silberhorn_2002} and discrete modulation (DM) \cite{Heid_2006, Zhao_2009, Sych_2010}. While Gaussian-modulated  protocols have been examined extensively \cite{Navascues_2006, Garcia_2006, Leverrier_2010, Diamanti_2015}, for a practically useful security analysis, one has to take the influence of finite constellations into account \cite{Jouguet_2012}. Furthermore, from a technical perspective, GM-protocols put high requirements on the classical error correction routine and on the modulation device.

Discrete modulation schemes for continuous-variable quantum key distribution (CV-QKD) enjoy implementation simplicity and compatibility with the existing telecommunication infrastructures. These features make them attractive to be deployed in future quantum-secured networks. While early security proofs for DM CV-QKD protocols were restricted to idealised cases~\cite{Heid_2006, Sych_2010} and have been lagging behind proofs for Gaussian-modulated protocols, significant progress has been made in the asymptotic regime recently \cite{Ghorai_2019,Lin_2019,Denys_2021,Upadhyaya_2021}. Although these analyses serve as an important first step toward a full security proof against general attacks in the finite-size regime, there remain challenging gaps to fill in order to complete the proof. A recent work \cite{Matsuura_2021} provides a finite-key analysis of the binary modulation protocol. This security proof uses the phase error rate approach that is commonly used in discrete-variable QKD security proofs, which seems to be challenging to extend beyond binary modulation. Unfortunately, due to the limitation of the binary modulation scheme, the key rate obtained is rather limited even for short distances and large block sizes \cite{Rigas_2006, Haeseler_2010}. One expects that much better performance can be obtained for higher constellation modulation schemes. Of particular interest is the quadrature phase-shift keying (QPSK) scheme. 
Very recently, a security proof against collective independently and identically distributed (i.i.d.) attacks for a discrete-modulated CV-QKD protocol was published \cite{Lupo_2022}. However, the secure finite-size key rates there converge against the asymptotic key rates in \cite{Denys_2021}, which - in contrast to Refs. \cite{Lin_2019, Upadhyaya_2021} are known to be loose for quaternary modulation.  

In this work, we present a finite-size security analysis for discrete-modulated CV-QKD protocols under the assumption of i.i.d. collective attacks. Although this does not represent the most general type of attacks, it is believed that key rates against collective i.i.d. attacks can be related to key rates against general attacks \cite{Renner_Cirac_2009, Christandl_2009, Leverrier_2017}, hence are optimal up to de-Finetti reduction terms. However, as DM CV-QKD protocols are described in infinite dimensional Hilbert spaces and lack the universal rotation symmetry of CV protocols with Gaussian modulation, these techniques cannot be applied directly. We emphasize that our proof method is very general and does apply to general discrete modulation patterns. For illustration purposes, we demonstrate our proof method for a four-state quadrature phase-shift keying protocol and calculate secure key rates using the security proof framework of Refs. \cite{Coles_2016, Winick_2018}. 

While there already exists an extension of this numerical security proof framework to the finite-size regime \cite{George_2020} for finite-dimensional spaces, we extend and generalise this to infinite dimensional Hilbert spaces, as required to treat CV-QKD protocols. In our work, we focus on heterodyne detection and examine only reverse reconciliation, which is known to perform better than direct reconciliation for long transmission distances. We want to emphasise that our proof method is not restricted to these cases and can be adapted to include homodyne measurements as well as direct reconciliation. Our approach does not assume a priori a finite maximum photon-number but employs a rigorous treatment of infinite dimensions. While the work in Ref. \cite{Lupo_2022} exploits the finite detection range of realistic detectors but assumes perfect detection efficiency, our approach also takes nonunit detection efficiencies into account and allows trusted detection. Even though a direct comparison of the obtained key rates is difficult, we observe that our finite-size key rates converge to the asymptotic key rates given in Ref. \cite{Upadhyaya_2021}, while the finite-size key rates in Ref. \cite{Lupo_2022}, based on a Gaussian extremality argument, converge to the asymptotic key rates in Ref. \cite{Denys_2021}, which for quaternary modulation are known to be loose and clearly lower than the key rates in Ref. \cite{Upadhyaya_2021}. This leads to clearly higher key rates and significantly higher maximum transmission distances for our proof.

This paper is structured as follows. In Section \ref{sec:Protocol}, we describe the general DM CV-QKD protocol. In Section~\ref{sec:Background}, we introduce the notation for our paper (Section~\ref{sec:Notation}), discuss briefly Renner's $\epsilon$-security framework (Section~\ref{sec:epsSecurityFramework}) and the dimension reduction method (Section~\ref{sec:DimRedMethod}). In Section~\ref{sec:SecProof} we first outline the idea of our security proof (Section~\ref{sec:outline}), and then state our energy testing theorem (Section~\ref{sec:EnergyTest}) as well as our  acceptance test theorem (Section \ref{sec:parameter_estimation}). Finally, we present our security proof in Section~\ref{sec:FinSizeSecPrf}.
In Section \ref{sec:NumPrfMethod}, we summarise the numerical method we are going to use to calculate a lower bound on our key rate expression from the previous section and state the minimisation problem we have to solve. Furthermore, we include a brief explanation of the trusted, nonideal detector model. We present numerical key rates in Section \ref{sec:Results}, both for untrusted, ideal and trusted nonideal detectors. For ease of comparison to previous work, we present most of our findings in the setting of a `unique acceptance set' as previous works often do. However, as acceptance sets define on which observations the protocol does not abort, they are important to evaluate the expected secure key rates of protocols (see Section \ref{sec:epsSecurityFramework}). Thus, in Section \ref{sec:nonUA}), we also provide plots of the key rate for a nonunique acceptance set. Finally, in Section \ref{sec:conclusion}, we summarise our results and give an outlook.


\section{Protocol description\label{sec:Protocol}}
In what follows, we describe the discrete-modulated CV-QKD protocol we consider in the present work, where $N_{\mathrm{St}} \in \mathbb{N}$ denotes the number of distinct signal states used in the protocol and Greek letters put in bra-ket notation refer to coherent states. We present the prepare-and-measure version of the protocol. Note that thanks to the source-replacement scheme \cite{Curty_2004, Ferenzci_2012} this is equivalent to the entanglement-based version of the protocol and we are free to switch between both versions in case this eases the security analysis.
\begin{itemize}
    \item[1] \textbf{\textit{State preparation---}} Alice prepares one out of $N_{St}$ possible coherent states $|\alpha\rangle$ with $\alpha \in \{ \alpha_0, ..., \alpha_{N_{St}-1}\}$ in her lab with equal probability and sends it to Bob using the quantum channel. Alice associates every state with a symbol and keeps track of what she sent in a private register.
    
    \item[2] \textbf{\textit{Measurement---}} Bob receives the signal and performs a heterodyne measurement to determine the quadratures of the received signal. This can be described by a positive operator-valued measure (POVM), for example, $\{E_{\gamma}= \frac{1}{\pi} |\gamma\rangle\langle \gamma|~:~ \gamma \in \mathbb{C}\}$. After applying this POVM, Bob holds a complex number $y_k \in \mathbb{C}$ that is stored in his private register.
 \end{itemize}   
 
    Steps 1 and 2 are repeated $N$ times.
    
 \begin{itemize}   
    \item[3] \textbf{\textit{Energy test---}} After completing the state preparation and measurement phases, Bob performs an energy test on $k_{\mathrm{T}}<\!\!<N$ rounds by using the measurement results related to these rounds. If for most of the tested signals, the heterodyne detection gave small measurement results (see Eq. \ref{eq:EnergyTest}) , the test passes. This means that most of the weight of the transmitted signals lies within a finite-dimensional Hilbert space, except with some small probability $\epsilon_{\mathrm{ET}}$. Otherwise, Alice and Bob abort the protocol. For details about the energy test, we refer to Section \ref{sec:EnergyTest}.
    

    \item[4] \textbf{\textit{Acceptance test ---}} If the energy test was successful, Bob discloses the data from the rounds he used for the energy test via the classical channel. This information is used by Alice and Bob to determine statistical estimators for their observables. If they lie within the acceptance set, Alice and Bob proceed, otherwise, they abort the protocol. 

    \item[5] \textbf{\textit{Key map---}} Bob performs a reverse reconciliation key map on the remaining $n:=N-k_T$ rounds to determine the raw key string $\tilde{z}$. For this purpose, Bob's measurement outcomes are discretised to an element in the set $\{0, ..., N_{\mathrm{St}}-1, \perp \}$, where symbols mapped to $\perp$ are discarded. By choosing a key map that discards results in certain regions of the phase space, Bob can perform postselection as described in \cite{Lin_2019}. 
    
    \item[6] \textbf{\textit{Error correction---}} Alice and Bob publicly communicate over the classical channel to reconcile their raw keys $\tilde{x}$ and $\tilde{z}$. After the error correction phase, Alice and Bob share a common string except with a small probability $\epsilon_{\mathrm{EC}}$.
    
    \item[7] \textbf{\textit{Privacy amplification---}} Finally, they apply a two-universal hash function to their common string. Except with small probability $\epsilon_{\mathrm{PA}}$, in the end, Alice and Bob hold a secret key.
\end{itemize}

We note that step 4 is often called \textbf{\textit{parameter estimation}}. However, we want to emphasise that in the finite-size regime we can never estimate any properties of the `real' density matrix, but only determine some statistical quantities based on our observations. First, we define a so-called acceptance-set, which can be imagined as a list of accepted observations. Based on our measurement results, we partition the set of all density matrices into two disjoint sets. The first one contains density matrices that lead to accepted statistics with probability less than $\epsilon_{\mathrm{AT}}$, i.e., the protocol aborts with high probability for those states. The second set is the complement of the first one and in what follows, we can restrict our security considerations to states lying in the latter set, called the `relevant set'. Based on this construction, we restrict our analysis to states that are $\epsilon$-secure with $\epsilon < \epsilon_{\mathrm{AT}}$. For a more detailed discussion of the idea of acceptance sets, we refer the reader to \cite[Section II.B]{George_2020}, where this notion is discussed for discrete-variable QKD.

 While we present our security proof approach for an arbitrary number $N_{\mathrm{St}}$ of signal states, we demonstrate our numerical results for a quadrature phase-shift keying protocol with $N_{St} = 4$, where all four states are arranged equidistant on a circle with radius $|\alpha|$, $\alpha_k \in \left\{ |\alpha|, i |\alpha|, -|\alpha|, -i |\alpha| \right\}$, where $i$ denotes the complex unit. In this case, the key map in step 5 of the protocol description looks as follows

\begin{equation}
            \tilde{z}_k = \left\{
\begin{array}{ll}
0 \text{ if } -\frac{\pi}{4} \leq \arg(y_k) <  \frac{\pi}{4} &\land~ |y_k| \geq \Delta_r, \\
1 \text{ if } \frac{\pi}{4} \leq \arg(y_k) <  \frac{3\pi}{4} &\land~  |y_k| \geq \Delta_r, \\
2 \text{ if } \frac{3\pi}{4} \leq \arg(y_k) <  \frac{5\pi}{4} &\land~  |y_k| \geq \Delta_r, \\
3 \text{ if } \frac{5\pi}{4} \leq \arg(y_k) <  \frac{7\pi}{4} &\land~  |y_k| \geq \Delta_r, \\
\perp &  \textrm{ otherwise,} \\
\end{array}
\right.
        \end{equation}

 where $\Delta_r \geq 0$ is the radial postselection parameter and $\arg(z)$ denotes the polar angle between the vector representing $z$ and the positive $q$ axis.

\section{Background}\label{sec:Background}
In this section, we set the stage for our security analysis by giving the necessary background. We summarise the notation used (Section \ref{sec:Notation}), briefly discuss attack types and $\epsilon$-security (Section \ref{sec:epsSecurityFramework}) and summarise the proof method of dimension reduction (Section \ref{sec:DimRedMethod}).

\subsection{Notation}\label{sec:Notation}
We start by clarifying the mathematical terminology and notation.
\subsubsection{Miscellaneous notation}
In the present work, by $\mathcal{H}$ we denote a separable Hilbert space, where we do not make any assumptions about the dimension. In particular, $\mathcal{H}$ can be infinite dimensional. If we want to explicitly refer to a finite-dimensional Hilbert space, we add a superscript $\mathcal{H}^n$, where $n$ refers to the highest number state that is still part of the Hilbert space. Since number states start with the vacuum state $|0\rangle$, $\mathcal{H}^n$ contains a maximum of $n+1$ linearly independent vectors; hence, the dimension of $\mathcal{H}^n$ is $n+1$. By $\mathcal{B}(\mathcal{H})$ we mean bounded operators on $\mathcal{H}$, while  $\mathcal{T}(\mathcal{H}) := \{ X \in \mathcal{B}(\mathcal{H})~:~ ||X||_{1} < \infty \} \subseteq \mathcal{B}(\mathcal{H})$ denotes the set of trace-class operators, where $||\cdot||_{1}$ is the Schatten-1 norm, $||A||_p := \sqrt[p]{\Tr{|A|^p}} = \left( \sum s_k(A) \right)^{\frac{1}{p}}$. Note that $s_k(A)$ denotes the $k$th singular value of $A$ (i.e., the $k^{\mathrm{th}}$ eigenvalue of $|A|:=\sqrt{A^{\dagger}A}$). If we add the subscript $1$, $\mathcal{T}_1(\mathcal{H})$, we refer to trace-class operators with norm $\leq 1$, while adding the superscript $+$, $\mathcal{T}^+(\mathcal{H})$, we denote the set of positive trace-class operators. By $\mathrm{Pos}(\mathcal{H}) := \mathrm{cone}\left(\mathcal{T}^+(\mathcal{H})\right)$ we denote the positive cone. The set of density operators on $\mathcal{H}$ is given by $\mathcal{D}(\mathcal{H}) := \left\{ X \in \mathrm{Pos}(\mathcal{H})~:~ ||X||_1 = 1 \right\}$ and by adding the sub-script $\leq$, we refer to the class of subnormalised density operators on $\mathcal{H}$,  $\mathcal{D}_{\leq}(\mathcal{H}) := \left\{ X \in \mathrm{Pos}(\mathcal{H})~:~ ||X||_1 \leq 1 \right\}$. Finally, by $\mathcal{S}_1(\mathcal{H})$ we denote the set of pure states on $\mathcal{H}$.

We use natural units in the whole manuscript, hence the quadrature operators read $\hat{q} := \frac{1}{\sqrt{2}}(\hat{a}^{\dagger}+ \hat{a})$ and $\hat{p} := \frac{i}{\sqrt{2}}(\hat{a}^{\dagger}- \hat{a}),$
where $\hat{a}$ and $\hat{a}^{\dagger}$ are the bosonic ladder operators defined by their action on number states $\hat{a}^{\dagger} |n\rangle = \sqrt{n+1} |n+1\rangle$ and $\hat{a} |n\rangle = \sqrt{n} |n-1\rangle$. Then, the commutation relation between the quadratures $q$- and $p$ reads $[\hat{q}, \hat{p}] = \mathbbm{1} i$. Another important operator will be the displacement operator $\hat{D}(\beta) := \exp{\beta \hat{a}^{\dagger} - \beta^* \hat{a} }$. We denote displaced quantities by writing the displacement into the subscript. For example, displaced number states (with displacement $\beta$) will be denoted by $|n_{\beta}\rangle := \hat{D}(\beta) \ket{n}$.

\subsubsection{Distance measures}
The trace distance and purified distance are two common distance measures used in this work to quantify the distance between two quantum states. The trace distance is given by $\Delta(\rho, \sigma) := \frac{1}{2} ||\rho-\sigma||_1$ while the purified distance is defined as $\mathcal{P}(\rho,\sigma) := \sqrt{1 - F_{*}(\rho,\sigma)}$, where 
\begin{equation}
 F_{*}(\rho,\sigma) := \sup_{\mathcal{H}':~ \mathcal{H}'\supseteq \mathcal{H}} ~\sup_{\stackrel{\bar{\rho}, \bar{\sigma} \in \mathcal{D}(\mathcal{H}')}{\Pi \bar{\rho}\Pi = \rho, ~ \Pi \bar{\sigma}\Pi = \sigma }} F(\bar{\rho}, \bar{\sigma})   
\end{equation}
is the generalised fidelity. Here, $\Pi$ is the projector onto $\mathcal{H}$ and $F(\rho, \sigma) := \left( \Tr{ \sqrt{\sqrt{\rho} \sigma \sqrt{\rho}}} \right)^2$ is the traditional fidelity. 

The purified distance and the trace distance are related via the Fuchs-van de Graaf inequalities \cite{Fuchs_1999}
\begin{equation}\label{eq:FvdG}
    \Delta(\rho,\sigma) \leq \mathcal{P}(\rho, \sigma) \leq \sqrt{2 \Delta(\rho,\sigma)}.
\end{equation}

\subsubsection{Smooth min-entropy}\label{sec:SmoothMinEntropy}
Besides the von Neumann entropy, the (smooth) min-entropy is an important information measure in QKD security analyses and is used to quantify the uncertainty of an observer on a quantum state. Therefore, in the present subsection, we briefly define and introduce this quantity. For  separable Hilbert spaces $\mathcal{H}_A, \mathcal{H}_B$ as well as $\rho_{AB} \in \mathcal{D}(\mathcal{H}_A \otimes \mathcal{H}_B)$, $\sigma_B \in \mathcal{D}(\mathcal{H}_B)$ we define the min-entropy of $\rho_{AB}$ relative to $\sigma_B$ by 
\begin{align*}
    H_{\mathrm{min}}(\rho_{AB} || \sigma_B)&:= -\log_2 \inf\left\{ \lambda \in \mathbb{R}:~ \lambda \mathbbm{1}_A\otimes \sigma_B \geq \rho_{AB} \right\}.
\end{align*}
The min-entropy of $\rho_{AB}$ given $\mathcal{H}_B$ is then
\begin{align*}
H_{\mathrm{min}}(A | B)_{\rho} &:= \sup_{\sigma_B \in \mathcal{D}(\mathcal{H}_B)} H_{\mathrm{min}}(\rho_{AB} || \sigma_B).
\end{align*}
Based on the nonsmoothed version, we introduce the smooth min-entropy of $\rho_{AB}$ relative to $\sigma_B$
\begin{align*}
H_{\mathrm{min}}^{\epsilon}(\rho_{AB} | \sigma_B) &:= \sup_{\tilde{\rho} \in \mathcal{B}^{\epsilon}(\rho)} H_{\mathrm{min}}(\tilde{\rho}_{AB} || \sigma_B),
\end{align*}
with, $\mathcal{B}^{\epsilon}(\rho)$ denoting the $\epsilon$-ball around $\rho$. Depending on the distance measure used for smoothing, the $\epsilon$-ball reads

\begin{align*}
    \mathcal{B}_{\mathrm{TD}}^{\epsilon}(\rho)  := &\left\{ \tilde{\rho} \in \mathrm{Pos}(\mathcal{H}_A \otimes \mathcal{H}_B): \Tr{\rho} \geq \Tr{\tilde{\rho}} \right. \\ 
    &~\left. \land ||\rho - \tilde{\rho}||_1 \leq \Tr{\rho} \epsilon\right\} \\
    \mathcal{B}_{\mathrm{PD}}^{\epsilon}(\rho) := &\left\{ \tilde{\rho} \in \mathcal{D}_{\leq}(\mathcal{H}_A \otimes \mathcal{H}_B) :~ \mathcal{P}(\rho, \tilde{\rho}) \leq \epsilon \right\}. 
\end{align*}

Finally, the smooth min-entropy of $\rho_{AB}$ given $\mathcal{H}_B$ reads
\begin{align*}
H_{\mathrm{min}}^{\epsilon}( \rho_{AB}| B) := \sup_{\sigma_B \in \mathcal{D}(\mathcal{H}_B)} H_{\mathrm{min}}^{\epsilon}(\rho_{AB} | \sigma_B).
\end{align*}

In the remaining text, we are going to indicate the used smoothing ball in the subscript, so $H_{\mathrm{min(TD)}}^{\epsilon}$ for trace distance smoothing and $H_{\mathrm{min(PD)}}^{\epsilon}$ for purified distance smoothing.

\subsection{Composable security and the $\epsilon$-security framework}\label{sec:epsSecurityFramework}
In this section, we summarize the idea of composable security, Renner's $\epsilon$-security framework, and $\epsilon-$completeness \cite{Renner_2004, Renner_2005}. Usually, we analyse the security of cryptographic tasks that will be combined with other cryptographic routines to form a large cryptographic protocol. Therefore, we demand so-called composable security of cryptographic routines, which means that the security of a combination of those routines can be given solely relying on the security of its subprotocols. The definition of composable security compares an ideal secure protocol with the real protocol and asks if an adversary is able to distinguish between both protocols when given access to the outputs of both protocols but not Alice's and Bob's private data. Formally, for QKD this means that the adversary is given two quantum states, 
$\rho_{\mathrm{ideal}} := \frac{1}{|\mathcal{S}|} \sum_{s \in \mathcal{S}} |s\rangle\langle s| \otimes |s\rangle\langle s| \otimes \rho_E$ and $\rho_{\mathrm{real}} := \rho_{S_A S_B E}$, where the first one is the output of the ideal protocol and the second one the output of the real protocol. Here, $\mathcal{S}$ is the set of possible keys, $S_A$ and $S_B$ are Alice's and Bob's keys, respectively, and $\rho_E$ denotes Eve's state. 

Since we cannot expect any protocol to be perfectly secure, we aim to limit the adversary's advantage when distinguishing between the ideal and the real protocol by some small number $\epsilon > 0$. The formal security condition then reads
\begin{align*}
        \frac{1}{2} \left|\left| \rho_{S_A S_B E} - \left(\frac{1}{|\mathcal{S}|} \sum_{s \in \mathcal{S}} |s\rangle\langle s| \otimes |s\rangle\langle s|\right) \otimes \rho_E \right|\right|_1 \leq \epsilon.
    \end{align*}
  
So, the adversary's advantage when distinguishing between the ideal and the real protocol is smaller or equal to $\frac{1}{2}+ \epsilon$. Taking a closer look at this difference, we observe that we formalise how much the realistic state differs from a situation where Alice and Bob share exactly the same key and Eve is fully decoupled from their system. By applying a triangle inequality in the security definition, these conditions can be considered separately as $\epsilon_{\mathrm{cor}}$-correctness and $\epsilon_{\mathrm{sec}}$-secrecy (see, for example, \cite[Theorem 4.1]{Portman_2014}). The $\epsilon_{\mathrm{cor}}$-correctness condition, $\mathrm{Pr}\left[ s_A \neq s_B \right] \leq \epsilon_{\mathrm{cor}}$, describes the situation where the protocol does not abort and Alice and Bob do not share the same key, chosen according to the distribution defined by $\rho_{S_{A} S_{B}}$. The $\epsilon_{\mathrm{sec}}$-secrecy condition can be written as $(1-p_{\mathrm{abort}}) \Delta\!\left(\rho_{S_A E}, \frac{1}{|\mathcal{S}|} \sum_{s \in \mathcal{S}} |s\rangle\langle s|  \otimes \rho_E  \right) \leq \epsilon_{\mathrm{sec}}$ and captures the situation, where the protocol does not abort and the shared key is not private, i.e., known to Eve. A more detailed discussion of composability and $\epsilon$-security can be found in Ref. \cite{Portman_2014}.

\paragraph*{Completeness} Lastly, we remark that $\epsilon-$security alone does not imply that a protocol is practical. This is easy to see. Consider a protocol that aborts unless it observes a specific set of statistics $q^{\star} \in \mathbb{R}^{m}$ for some $m \in \mathbb{N}$, which we later refer to as `unique acceptance.' Then, in general, one would expect even if one were sampling from the distribution $q^{\star}$, the probability of observing $q^{\star}$ would be small for a finite number of samples. Therefore, the probability of aborting the protocol will be high. It would follow that even if one could generate a great deal of key conditioned on nonaborting, the protocol is not very useful because it might almost always abort. The definition of completeness captures this notion.
\begin{definition}\label{def:completeness}
    A QKD protocol is $\nu^{\mathrm{c}}_{\mathrm{QKD}}-$complete if 
    $$ \Pr[\mathsf{Abort}|\mathsf{Honest}] \leq \nu_{\mathrm{QKD}}^{c} \, $$
    where $\mathsf{Honest}$ means the honest implementation of the protocol, which is defined by the expected behaviour of the devices and the communication channel. That is, it is $\nu^{\mathrm{c}}_{\mathrm{QKD}}-$complete only if when Eve `does nothing', the protocol accepts except with probability $\nu^{c}_{\mathrm{QKD}}$.
\end{definition}

\subsection{Dimension reduction method}\label{sec:DimRedMethod}
Proving the security of CV-QKD protocols involves dealing with optimisation problems over infinite dimensional Hilbert spaces. However, numerical methods for key rate calculation can only be applied to finite-dimensional problems. Assuming an artificial heuristically argued cutoff is not rigorous enough for a finite-size security analysis. The dimension reduction method~\cite{Upadhyaya_2021} connects an infinite dimensional convex optimisation problem to a finite-dimensional problem. In more detail, under some reasonable requirements for the objective function, the dimension reduction method tightly lower-bounds the infinite dimensional convex optimisation problem by a finite-dimensional convex optimisation problem and some penalty term. In what follows, we state the main theorem (\cite[Theorem 1]{Upadhyaya_2021}) where we used the improved correction term from Refs. \cite{Upadhyaya_Thesis_2021, Upadhyaya_2022}. We refer the reader to the original paper for further details.

\begin{theorem}[\textbf{Dimension Reduction}\label{thm:dimReduction}]
Let $\mathcal{H}$ be a separable Hilbert space and $\Pi$ the projection onto some finite-dimensional subspace $\mathcal{H}_{\mathrm{fin}}$ of $\mathcal{H}$ as well as $\Pi^{\perp}$ the projection onto $\left(\mathcal{H}_{\mathrm{fin}}\right)^{\perp}$. Let $\rho_{\infty} \in \mathcal{D}_{\leq}(\mathcal{H})$ and $\rho_{\mathrm{fin}}~\in~\mathcal{D}_{\leq}(\mathcal{H}_{\mathrm{fin}})$.\\
If $f: \mathcal{D}_{\leq}(\mathcal{H}) \rightarrow \mathbb{R}$ is uniformly close to decreasing under projection, that is
\begin{align*}
    F(\sigma, \Pi \sigma \Pi) \geq \Tr{\sigma} - w ~\Rightarrow~ f(\Pi\sigma\Pi) - f(\sigma) \leq \Delta(w),
\end{align*}
and $w \leq \Tr{\rho \Pi^{\perp}}$, then
\begin{align*}
    f(\rho_{fin}) - \Delta(w) \leq f(\rho_{\infty}),
\end{align*}
where 
\begin{equation}\label{eq:weightCorrection}
 \Delta(w) := \sqrt{w} \log_2(|Z|) + (1+\sqrt{w}) h\left( \frac{\sqrt{w}}{1+\sqrt{w}} \right).   
\end{equation} Here, $|Z|$ denotes the dimension of the key map and $h(\cdot)$ is the binary entropy.
\end{theorem}

Note that the weight $w$ depends on the dimension of the chosen finite-dimensional Hilbert space, so the correction term $\Delta(w)$ depends on the chosen subspace $\mathcal{H}^{\mathrm{fin}}$. Consequently, we aim to choose a subspace such that the weight can be expected as small as possible. Based on a model for fibre-based implementations of QKD protocols, it was shown in \cite{Upadhyaya_2021} that it is advantageous to project onto a subspace spanned by displaced Fock states $|n_{\gamma} \rangle = \hat{D}(\gamma) |n\rangle$. Then, for the $i$th state the projection acting on Bob's Hilbert space reads $\Pi := \sum_{n=0}^{n_c} |n_{\beta_i}\rangle\langle n_{\beta_i} |$, where $\{\beta_i\}_{i=0}^{N_{\mathrm{St}}-1}$ is a list of complex numbers, chosen as $\sqrt{\eta} \alpha_i$.


\section{Security proof approach}\label{sec:SecProof}
In contrast to discrete-variable QKD and Gaussian-modulated CV-QKD, the security of discrete-modulated CV-QKD protocols has so far mainly been analysed in the asymptotic limit. Many useful symmetry properties, simplifications and tricks for protocols with Gaussian modulation that help to handle infinite dimensions there do not apply to discrete-modulated protocols, so we cannot expect security proofs to have a similar structure. Instead we to apply the numerical security proof framework introduced in Refs. \cite{Coles_2016, Winick_2018} to obtain lower bounds on the secure key rate. Before we can do so, we need to find an expression for a lower bound on the secure key rate in the finite-size regime and argue the security of the underlying protocol.

In contrast to the asymptotic case, in finite-size analyses, the expectation values of our observables are not known with certainty. Hence, we need to define an acceptance set and consider in our security analysis only states that are more than $\epsilon$-likely to produce a compatible observation. Therefore, we need to perform a statistical test. Unfortunately, most of the standard (well-scaled) concentration inequalities only apply for bounded observables, while, for example, the photon-number operator is unbounded for infinite dimensional Hilbert spaces. This is a serious issue, since the standard dimension reduction method, which one might want to use to reduce the dimension of the problem, cannot be applied directly as we need to know the (finite) expectation values of our (unbounded) observables to even formulate the finite-dimensional lower bound of the original optimisation. Besides that, we expect additional correction terms that are suppressed in the limit of infinitely many rounds but may become relevant for a finite number of signals.

Finally, from the perspective of a security proof, we note that many statements in Renner's thesis \cite{Renner_2005} assume finite-dimensional Hilbert spaces; therefore we need to carefully analyse which statements in the $\epsilon$-security framework we want to use can be extended to infinite dimensional Hilbert spaces. Having listed the difficulties of a DM CV-QKD security proof, we provide a high-level outline of our proof in the following section.

\subsection{High-level outline of the security proof}\label{sec:outline}
Before we discuss the intricacies of our security proof, let us present the big picture of our approach. In our proof, we consider i.i.d. collective attacks. This means that Eve prepares a fresh ancilla state to interact with each round of the protocol in an identical manner and then stores them in a quantum memory. Once Alice and Bob have finally executed their protocol, she measures her quantum memory, encompassing all the ancillae, collectively. In particular, this means that there are no correlations between different rounds, enabling us to treat each round equally. 

Since Alice's quantum signals went through the quantum channel, which is under Eve's control, we do not know a priori if there is a maximum photon-number in the states Bob receives.  Moreover, since the worst-case scenario occurs when Eve possesses a purification of Bob's states, her purifying system is also infinite dimensional. Consequently, we require a security proof that encompasses infinite dimensional systems.  

Within Renner's finite-size framework \cite{Renner_2005}, the leftover hashing lemma tells us that if Alice and Bob apply a randomly chosen hash function from the family of two-universal hash functions, the output is secure as long as it is smaller than Eve's uncertainty about Alice's and Bob's initial key strings. However, Renner's initial work assumes finite-dimensional Hilbert spaces, so we cannot apply his results directly. To resolve this, we use the leftover hashing lemma against infinite dimensional side information (\cite[Proposition 21]{Berta_2016}) to derive our entropic condition on the key length (Lemma \ref{lemma:leftoverHashing} in Appendix \ref{apdx:TechnicalLemmas}). It remains to take the effect of classical communication during the error correction phase into account. Thanks to Lemma \ref{lemma:removingClassRegister} in Appendix \ref{apdx:TechnicalLemmas} we can separate Eve's information leakage from information reconciliation and from other sources, and convert the effect of the information reconciliation term into a leakage term, even if one one of the conditioning systems (Eve's purifying system) is still infinite dimensional. We then use various properties of the smooth min-entropy to simplify the expression, giving an upper bound on the secure key rate. Following the methodology of Furrer et al. \cite{Furrer_2011}, we establish the asymptotic equipartition property (AEP) from Renner's thesis~\cite{Renner_2005} and extended it to infinite dimensional quantum side information (Corollary \ref{cor:AEP} in Appendix \ref{apdx:GeneralisationAEP}) \footnote{We generalize this form of the AEP rather than using the extension of the fully quantum AEP \cite{Furrer_2011} as it applies for all block lengths and is simpler to apply for numerical calculations.}.

We aim to apply a generalized version of the numerical security proof framework introduced in Refs. \cite{Coles_2016, Winick_2018}. Hence, we have to represent the relevant occurring quantum systems on a computer and solve optimization problems. We cannot represent infinite dimensional states or spaces on a computer. In particular, there is a maximum  practical dimension that can be represented numerically, which means that the numerical dimension of the problem cannot grow with the block size. To make our security proof rigorous, we do not want to simply assume a cutoff dimension. Thus, we design a method that guarantees that the analysed quantum states have high weight in a low-dimensional (thus, in particular, finite-dimensional) subspace. In more detail, within the framework of our acceptance analysis, we develop an energy test (Theorem~\ref{thm:energy_test} below) that rigorously bounds the effective dimension. If the test passes, except with some small probability $\epsilon_{\mathrm{ET}}$, most of the weight of the states sent lies within the chosen cutoff space $\mathcal{H}^{n_c}$. The remaining errors due to cutting off at some finite-dimension are handled by the dimension reduction method \cite{Upadhyaya_2021} (Theorem~\ref{thm:dimReduction}), which allows us to translate the infinite dimensional optimization problem into a finite-dimensional semidefinite program.

It remains to discuss how the acceptance set is defined. The acceptance analysis guarantees that the state generated by Eve's attack either results in a secure key via the specified protocol or generates statistics such that the protocol aborts except with small probability. Now, recall that the security proof has to be done in infinite dimensions and that the dimension reduction method relates a well-defined infinite dimensional optimization problem with a finite-dimensional one. Thus, the acceptance set has to be defined on the infinite dimensional states. Unfortunately, the convergence of the sample mean to the true mean of unbounded random variables is only limited by Chebyshev's inequality, which gives slow convergence, and hence low key rates. Hence, using unbounded observables would be impractical. To enforce our observables to be bounded, we introduce a ``soft detection limit,'' i.e., we coarse grain the measurement results, which allows us to bound our modified observables. We then can use Hoeffding's inequality to perform a statistical test, the acceptance test (Theorem \ref{Thm:ParameterEstimation} below), and obtain bounds on the expectations. Mathematically, we distinguish between two scenarios for both tests. Either the test fails, meaning that with high probability the observed statistical quantity does not correspond to a state in our acceptance set, or the test passes. Hence, after performing both the energy test and the acceptance test, we know that the actual state is $\epsilon_{\mathrm{ET}}+\epsilon_{\mathrm{AT}}$ close to the set we consider in our security analysis. 

Finally, we obtain a semidefinite program that we solve with an extension of the numerical framework presented in \cite{Coles_2016, Winick_2018}.

\subsection{Bounding observables}\label{sec:BoundingObservables}
As we we argued in the previous section, it is crucial for the security proof that the observables are bounded. To achieve this, inspired by real detectors, we modify our detector model such that detectors have a finite detection range, i.e., possible measurement outcomes are confined in a finite region $\mathcal{M}$, for example, $q,p \in \mathcal{M}=[-M,M]$ of the phase space. We note that this parameter $M$ does not have to be exactly the physical limit of the real detector (e.g., the value corresponding to the maximal output of the analog-to-digital converter (ADC) ) as we simply introduce a `soft detection limit' that only has to be smaller than the physical detection limit. This method takes results $q$ and $p$ with values larger (smaller) than $M$ ($-M$) and simply sets them to $M$ ($-M$). We want to highlight that this means that we do not need to model the exact physical process happening when strong laser pulses enter the detector, as long as we set $M$ small enough. Effectively, we introduce an additional postselection region for measurement results with absolute value larger than $M$, which is already included in our postprocessing framework (see Ref. \cite{Lin_2019}). For the time being, it suffices to know that this allows us to bound every observable $\hat{X}$ by some $x(M) < \infty$ and we postpone the detailed derivation for the observables occurring in the protocol we used to illustrate our security proof to later.

\subsection{Energy Test}\label{sec:EnergyTest}
One of the first steps in our protocol is to perform an energy test. The goal of performing an energy test is to make a probabilistic statement about the maximum energy of a set of states by testing a subset of the total number of signals. Before we come to our version, we briefly discuss issues with existing energy tests \cite{Renner_Cirac_2009, Leverrier_2013, Furrer_2014} that prevented us from applying one of those.

 The energy test presented in Ref. \cite{Renner_Cirac_2009} makes use of the permutation invariance of the individual rounds in many QKD protocols. There, the authors performed testing on some subset of the signals and states that, except with some small probability, most of the remaining rounds live in finite-dimensional Hilbert spaces. However, since there remain some possibly infinite dimensional rounds, we cannot apply this energy test. In contrast, the energy test in Ref. \cite{Leverrier_2013} examines a small subset of all rounds, resulting in a statistical statement about the dimension of all remaining rounds and does not leave back any possibly infinite dimensional systems. However, this test requires a very strong phase-space rotation symmetry that our protocol does not satisfy. The approach in Ref. \cite{Furrer_2014} adds a beam splitter to the experimental setup and therefore performs testing on some small fraction of every signal. However, as this comes with additional components such as a beam splitter and a second heterodyne measurement setup, it is experimentally less favourable. Thus, we developed our own energy test that does not require additional hardware and does not assume any particular phase-space symmetry.

As outlined in the protocol description, after transmitting $N$ rounds of signals, Alice and Bob perform an energy test on $k_{T}<\!\!<N$ modes, i.e. they perform a heterodyne measurement to determine the quadratures of the chosen rounds. As we show in Appendix \ref{apdx:proof_ET}, this can be used for the following statement.

\begin{theorem}[\textbf{Noise robust energy test}\label{thm:energy_test}]
Consider signal states of the form $\rho^{\otimes N}$, and let $k_T \in \mathbb{N}$, $k_T <\!\!< N$, be the number of signals sacrificed for testing and $l_T \in \mathbb{N}$ be the number of rounds that may not satisfy the testing condition. Denote by $(Y_1, ..., Y_{k_T})$ the absolute values of the results of the test measurement. Pick a weight $w\in [0,1]$, a photon cutoff number $n_{c}$ and a testing parameter $\beta_{\mathrm{test}}$ satisfying $M \geq \beta_{\mathrm{test}} > 0$, where $M>0$ is the finite detection range of the heterodyne detectors.  Define $r:= \frac{\Gamma(n_c+1,0)}{\Gamma(n_c+1,\beta_{\mathrm{test}})}$, where, $\Gamma(n,a)$ is the upper incomplete gamma function, as well as  $Q_y := \begin{pmatrix}
    1-y \\ y
\end{pmatrix}$
 and $P_{j} := \begin{pmatrix}
    1-\frac{j}{k_T}\\
    \frac{j}{k_T}
\end{pmatrix}$. Finally, let $\Pi^{\perp}$ be the projector onto the complement of the photon cutoff space $\mathcal{H}^{n_c}$.\\
Then, as long as $\frac{l_T}{k_T} < \frac{w}{r}$ for all $\rho$ such that $\Tr{ \Pi^{\perp} \rho} \geq w$,
\begin{equation}\label{eq:EnergyTest}
\begin{aligned}
    \mathrm{Pr}&\left[\left|\left\{ Y_j:~ Y_j < \beta_{\mathrm{test}}\right\} \right| \leq l_T \right]  \\ ~~~~& \leq (l_T+1) \cdot 2^{-k_T D\left(P_{l_T} || Q_{\frac{w}{r}}\right)}  =: \epsilon_{\mathrm{ET}},
\end{aligned}
\end{equation}
where $D(\cdot||\cdot)$ is the Kullback-Leibler divergence.
\end{theorem}
\begin{proof}
See Appendix \ref{apdx:proof_ET}.
\end{proof}

In other words, the energy test tells us that for all $\rho$ that satisfy $\Tr{\Pi^{\perp}\rho} \geq w$ the energy test will fail except with probability $\epsilon_{\mathrm{ET}}$.

Note that the theorem only tells us something in the case in which the energy test passes. If the energy test fails, we abort the whole protocol and therefore it is (trivially) secure. Furthermore, as Alice's lab is assumed to be inaccessible to Eve, the test needs to be performed only by Bob.

\subsection{Acceptance test}\label{sec:parameter_estimation}
After passing the energy test, working in a finite-dimensional Hilbert space allows us to specify the relevant set for our observables. This is the set we restrict our security analysis to (see our discussion in Section \ref{sec:Protocol}), based on statistical bounds for the observed values of our observables. This statistical test replaces the parameter estimation step in asymptotic security analyses. In particular, for any given set of observed statistics, the protocol must either abort or accept. To be secure, the acceptance set is a set of states such that any state not in the set could only have generated any of the accepted statistics with probability less than $\epsilon_{\mathrm{AT}}$. The following theorem establishes such a set of states.

\begin{theorem}[\textbf{Acceptance Test}]\label{Thm:ParameterEstimation}

Let $\Theta$ be the set of Bob's observables. Let $\mathbf{r} \in \mathbb{R}^{|\Theta|}$ and $\mathbf{t} \in \mathbb{R}^{|\Theta|}_{\geq 0}$, where $|\Theta|$ denotes the cardinality of $\Theta$. Define the set of accepted statistics as
\begin{equation}\label{eq:accepted-observations} \mathcal{O} := \{ \mathbf{v} \in \mathbb{R}^{\Theta} : \forall X \in \Theta, |v_{X} - r_{X}| \leq t_{X} \} \ ,
\end{equation}
and the corresponding acceptance set as
\begin{equation}\label{eq:ATset}
\begin{aligned}
    &\mathcal{S}^{\mathrm{AT}}:=  \left\{ \rho \in \mathcal{D}(\mathcal{H}_A \otimes \mathcal{H}_B^{n_c}):\right. \\
    & \hspace{18mm} \forall X \in \Theta, |\Tr{\rho X} - r_{X}| \leq \mu_{X} + t_{X} \},
\end{aligned}
\end{equation}
where $r_{X}$ is the $X$th element of the vector $\mathbf{r}$ and likewise for $t_{X}$. For every $X \in \Theta$, let \begin{equation*}
    \mu_X := \sqrt{\frac{2x^2}{m_{X}} \ln\left( \frac{2}{\epsilon_{\mathrm{AT}}} \right)}
\end{equation*}
or, if $X$ is a positive semidefinite operator,
\begin{equation*}
    \mu_X := \sqrt{\frac{x^2}{2m_{X}} \ln\left( \frac{2}{\epsilon_{\mathrm{AT}}} \right)},
\end{equation*}
where $x := \|X\|_{\infty}$ and $m_{X}$ is the number of tests for the observable $X$. If $\rho \not \in \mathcal{S}^{\mathrm{AT}}$, then the probability of accepting the statistics generated by the i.i.d. measurements of $\rho^{\otimes n}$ is bounded above by $\epsilon_{\mathrm{AT}}$. That is, the complement of $\mathcal{S}^{\mathrm{AT}}$ are all $\epsilon_{\mathrm{AT}}$-filtered.
\end{theorem}
\begin{proof}
First, using Hölder's inequality, for the observable $X$, we obtain
\begin{align*} 
    ||X \rho ||_1 \leq ||X||_{\infty} ||\rho||_1 = ||X||_{\infty} =: x,
\end{align*}
therefore, $\mathbb{E}(X) = \Tr{\rho X} \leq x$. This implies that our measurement results with respect to the observable $X$ lie within the interval $[-x,x]$ (or $[0,x]$ in case $X$ is positive semidefinite). Hence, we can apply Hoeffding's inequality~\cite{Hoeffding_1963} which states that 
\begin{equation}\label{eq:HoeffdingBound}
    \mathrm{Pr}\left[\left| \bar{X} - \mathbbm{E}[X]\right| \geq \mu_X\right] \leq 2 e^{-\frac{2 m_{X} \mu_X^2}{(2x)^2}} =: \epsilon_{\mathrm{AT}}^{X} , 
\end{equation}
where $\overline{X}$ is the average of the observations, i.e. the empirical mean. For positive semidefinite $X$, we replace $2x$ in the denominator of the exponent by $x$. Then, we obtain the $\mu_X$ given in the theorem statement from basic algebra.

Next we show that if $|\Theta|=1$ with a unique element $\widehat{X}$ then $\mathcal{S}^{\mathrm{AT}}$ only has $\epsilon_{\mathrm{AT}}-$filtered states in its complement. For this case, we denote the set  $\mathcal{S}^{\mathrm{AT}}_{\widehat{X}}$. Let $v_{\widehat{X}} \in \mathbb{R}$ be the empirical mean of this unique observable, e.g. $v_{\widehat{X}} := \overline{\widehat{X}}$ for $\widehat{X} \in \Theta$. Then we have by Hoeffding's inequality that except with probability $\epsilon_{\mathrm{AT}}^{\widehat{X}}$, $\left|v_{\widehat{X}} - \Tr{\rho \widehat{X}}\right| < \mu_{\widehat{X}}$, where $\rho$ is the state from which we are i.i.d. sampling. Now we show every state not in $\mathcal{S}^{\mathrm{AT}}_{\widehat{X}}$ is $\epsilon_{\mathrm{AT}}-$filtered. Let $\sigma \not \in \mathcal{S}^{\mathrm{AT}}_{\widehat{X}}$. Then,
\begin{align*}
    & \Pr\left[\mathsf{Accept}\mathrm{~AT}|\sigma\right] \\ 
    =& \Pr\left[\left|\Tr{\sigma \widehat{X}} - r_{\widehat{X}}\right| > \mu_{\widehat{X}} + t_{\widehat{X}} 
   \, \land \, \left|v_{\widehat{X}} - r_{\widehat{X}}\right| \leq t_{\widehat{X}}  \right] \ ,
\end{align*}
which follows from the definition of $\mathcal{S}^{\mathrm{AT}}_{\widehat{X}}$ and the definition of the accepted statistics \eqref{eq:accepted-observations}. Now note the implication
\begin{align*}
    & \left|\Tr{\sigma \widehat{X}} - r_{\widehat{X}}\right| > \mu_{\widehat{X}} + t_{\widehat{X}} \,\, \land \,\, |v_{\widehat{X}} - r_{\widehat{X}}| \leq t_{\widehat{X}} \\
    & \hspace{4cm} \Rightarrow  \left|\Tr{\sigma \widehat{X}} - v_{\widehat{X}}\right| > \mu_{\widehat{X}} \ ,
\end{align*}
which follows from the triangle inequality:
\begin{align*}
    \left|\Tr{\sigma \widehat{X}} - r_{\widehat{X}}\right|
    =& \left|\Tr{\sigma \widehat{X}} - v_{\widehat{X}} + v_{\widehat{X}} - r_{\widehat{X}}\right| \\
    \leq& \left|\Tr{\sigma \widehat{X}} - v_{\widehat{X}}| + |v_{\widehat{X}} - r_{\widehat{X}}\right| \ .
\end{align*}
Therefore, combining these points,
\begin{align*}
     \Pr\left[\mathsf{Accept}\mathrm{~AT}|\sigma\right]
     \leq \Pr\left[ \left|\Tr{\sigma \widehat{X}} - v_{\widehat{X}}\right| > \mu_{\widehat{X}}\right] = \epsilon^{\widehat{X}}_{\mathrm{AT}} \ .
\end{align*}
Thus, we have shown in the one-parameter case, the set $\mathcal{S}_{\widehat{X}}^{\mathrm{AT}}$ only has $\epsilon^{\widehat{X}}_{\mathrm{AT}}-$filtered states in its complement. 

All that is left to do is to lift from the one-parameter case to the many-parameter case. We want to do this without using a union bound. To do this, we first set the $\epsilon$-parameter to be the same for every observable, i.e.\ $\forall X, X' \in \Theta:~  \epsilon_{\mathrm{AT}}^X = \epsilon_{\mathrm{AT}}^{X'} =: \epsilon_{\mathrm{AT}}$. Then we note that $\mathcal{S}^{\mathrm{AT}} = \cap_{X \in \Theta} \mathcal{S}^{\mathrm{AT}}_{X}$. It is known that if one takes the intersection of sets each of which only has $\epsilon-$filtered states in the complement, then the intersection also only contains $\epsilon-$filtered states in the complement (\cite[Theorem 5]{George_2020}). Thus, as we established the filtering property for the single observable case, and $\mathcal{S}^{\mathrm{AT}}$ is the intersection of single observable cases, we know that if $\sigma \not \in \mathcal{S}^{\mathrm{AT}}$, then $\sigma$ is $\epsilon_{\mathrm{AT}}-$filtered. This is what we wanted to establish, so this completes the proof.
\end{proof}
Before moving forward, we note that the reason we need the vector $\mathbf{t} \in \mathbb{R}^{\Theta}_{\geq 0}$ is not for security, but rather for the completeness of the protocol. Indeed, if $\mathbf{t} = 0$ then we would filter all states that do not result in statistics $\mathbf{r}$ except with probability $\epsilon_{\mathrm{AT}}$. This has often been the case considered in previous works implicitly and we call this setting the unique acceptance set following terminology from Ref. \cite{George_2020}. However, we note that the probability of obtaining the statistics $\mathbf{r}$ is in general close to zero, so the protocol defined via a unique acceptance set aborts almost all of the time. For this reason a good key length in the unique acceptance setting is in some sense not useful. Thus, we use $\mathbf{t}$ to draw a ``box'' of accepted statistics around some ideal statistics $\mathbf{r}$. This will of course decrease the key rate, but it will increase the completeness, thereby making the protocol practical. Indeed, we can show the following.
\begin{proposition}\label{prop:completeness-bound}
    Let $\mathbf{r}$ be defined via $r_{X} := \Tr{\sigma X}$ where $\sigma$ is the state after the honest implementation of the channel. Let $l_{T},k_{T}$ be the same as in Theorem~\ref{thm:energy_test}, and let $V_1$ be defined as in the proof of Theorem~\ref{thm:energy_test}. Then, assuming $1-\Tr{V_{1}\sigma} < \frac{l_{T}+1}{k_{T}}$, the protocol is $(\nu^{c}_{\mathrm{ET}} + \nu^{c}_{\mathrm{AT}} + \nu^{c}_{\mathrm{EC}})-$complete where $\nu^{c}_{\mathrm{EC}}$ is a parameter of the chosen error correcting code and
    \begin{align*}
        \epsilon^{c}_{\mathrm{ET}} &:= (k_{T} - l_{T}-1)2^{-k_{T}D(P_{l_{T}+1}||Q_{\sigma})} \\
        \epsilon^{c}_{\mathrm{AT}} &:= 2 \sum_{X \in \Theta} e^{-2m_{X} t_{X}^{2}/(4\|X\|^{2}_{\infty})} \ ,
    \end{align*}
    where $m_{X}$ is the number of tests of observable $X$.
\end{proposition}
\begin{proof}
See Appendix \ref{apdx:Completeness}.    
\end{proof}
\noindent We note that if $t_{X} = 0$ for any $X$, then the protocol is always $1$-complete by these bounds, which we do not want.

To summarise, the above theorem tells us that states whose expected values deviate too far from $\mathbf{r}$ in terms of $\mu_X$ and $\mathbf{t}_{X}$, and hence are not part of the acceptance set $\mathcal{S}^{\mathrm{AT}}$, will only be accepted by our testing procedure with very low probability. Thus, at the cost of introducing a small probability of error $\epsilon_{\mathrm{AT}}$, the remaining security analysis focuses on states in $\mathcal{S}^{\mathrm{AT}}$. Additionally, via smart choices of parameter $\mathbf{t}$, the theorem allows us to tune the success probability of the protocol.

\subsection{Finite-size security proof}\label{sec:FinSizeSecPrf}
After having finished all preparations, we now establish the security proof of the present CV-QKD protocol against i.i.d. collective attacks. We state our main result, the security statement against i.i.d. collective attacks, in the following theorem and prove it afterwards.

\begin{theorem}[\textbf{Security statement against i.i.d. collective attacks}]\label{thm:SecurityStatement}
Let $\mathcal{H}_A$ and $\mathcal{H}_B$ be separable Hilbert spaces and let $\epsilon_{\mathrm{ET}}, \epsilon_{\mathrm{AT}}, \bar{\epsilon}, \epsilon_{\mathrm{EC}}, \epsilon_{\mathrm{PA}} > 0$. The objective QKD protocol is $\epsilon_{\mathrm{EC}} + \max\left\{\frac{1}{2}\epsilon_{\mathrm{PA}}+\bar{\epsilon}, \epsilon_{\mathrm{ET}}+\epsilon_{\mathrm{AT}} \right\}$-secure against i.i.d. collective attacks, given that, in case the protocol does not abort, the secure key length is chosen to satisfy
\begin{equation}
\begin{aligned}
    \frac{\ell}{N} \leq \frac{n}{N} &\left[ \min_{\rho \in \mathcal{S}^{\mathrm{E\&A}}} H(X|E')_{\rho} - \delta(\bar{\epsilon}) - \Delta(w) \right]     \\
    & - \delta_{\mathrm{leak}}^{\mathrm{EC}} - \frac{2}{N} \log_2\left( \frac{1}{\epsilon_{\mathrm{PA}}} \right),
\end{aligned}
\end{equation}
where $\delta^{\mathrm{EC}}_{\mathrm{leak}}$ takes the classical error correction cost into account, $\Delta(w)$ is given in Eq.~(\ref{eq:weightCorrection}), $\delta(\epsilon) := 2 \log_2\left( \mathrm{rank}(\rho_X)+3 \right) \sqrt{\frac{\log_2\left(2/\epsilon \right)}{n}}$ and $\mathcal{S}^{\mathrm{E\&A}}$ is defined below.
\end{theorem}
\begin{proof}
According to our assumption, after completing $N$ rounds of the quantum phase in the present QKD protocol, Alice and Bob share the state $\rho_{AB}^{\otimes N} \in \mathcal{D}((\mathcal{H}_A\otimes\mathcal{H}_B)^{\otimes N})$. Alice and Bob choose randomly $k_T$ of those rounds for testing, where they first perform the energy test, followed by the acceptance test. Recall the notion of $\epsilon$-securely filtered states; an input state $\sigma$ is called $\epsilon$-securely filtered if the probability that the corresponding statistical test does not abort on $\sigma$ is less than $\epsilon$. This allows us to define
\begin{align*}
    \mathcal{S}^{\mathrm{ET}} := &\left\{\sigma \in \mathcal{D}_{\leq}(\mathcal{H}_A \otimes \mathcal{H}_B^{n_c} \otimes \mathcal{H}_E): \textrm{ purification of } \rho_{AB}   \right. \\
    & \left. \land \mathrm{Tr}_{E}\left[\sigma\right] \text{ is not $\epsilon_{\mathrm{ET}}$-securely filtered in the ET} \right\}.
\end{align*}
Analogously, as a subset of all states that have not been filtered by the energy test, we define the set of states that have not been filtered by the acceptance test with probability greater than $1-\epsilon_{\mathrm{AT}}$ 
\begin{align*}
    \mathcal{S}^{\mathrm{E\&A}} := &\left\{ \sigma \in \mathcal{S}^{\mathrm{ET}}: \right.\\
    &\left.~ \mathrm{Tr}_{E}\left[\sigma\right] \text{ is not $\epsilon_{\mathrm{AT}}$-securely filtered in the AT} \right\}.
\end{align*}
This set combines the results of Theorem~\ref{thm:energy_test} and Theorem~\ref{Thm:ParameterEstimation}. In what follows, when we refer to `passing the testing' we mean that both tests pass successfully.

Because of the nature of statistical testing, in our security analysis we never know the actual state Bob receives, but only decide to proceed or abort the protocol, based on if the received state lies within a predefined set.
Therefore, we split the security argument into two cases: 
\begin{itemize}
    \item[1.)] the input state $\sigma$ is in set $\mathcal{S}^{\mathrm{E\&A}}$,
    \item[2.)] the input state $\sigma$ is not in set $\mathcal{S}^{\mathrm{E\&A}}$.  
\end{itemize}
Denote by $\Omega$ the event that Alice's and Bob's testing succeeds, i.e., the tests pass. Note that if we write a state conditioned on an event we do not imply that this state was renormalised. 

To ease notation, we define the map $\mathcal{E}^{\mathrm{QKD}} := \mathcal{E}^{\mathrm{key}} \circ \mathcal{E}^{\mathrm{AT}} \circ \mathcal{E}^{\mathrm{ET}}$, representing the action of the QKD protocol, where $\mathcal{E}^{\mathrm{ET}}$ and $\mathcal{E}^{\mathrm{AT}}$ denote the quantum channels representing the energy test and the acceptance test and $\mathcal{E}^{\mathrm{key}}$ is the map denoting the classical postprocessing.

Let $\rho_{ABE} = \sigma^{\otimes N}$ be an arbitrary i.i.d. input state and $\rho_{S_A S_B E'} := \mathcal{E}^{\mathrm{QKD}}\left( \rho_{ABE} \right)$. Here $E'$ denotes Eve's register $E$ including all information she gathered from the classical communication between Alice and Bob. This state can either pass or fail the testing. Note that the protocol is trivially secure if the testing procedure aborts the protocol. For the difference between $\rho_{S_A S_B E'}$ and a uniformly distributed key that is fully decoupled from Eve, we obtain
\begin{align*}
    &\frac{1}{2} \left|\left| \rho_{S_A S_B E'} - \pi_{S_A S_B} \otimes \rho_{E'} \right|\right|_1 \\
    & = (1 - \mathrm{Pr}[\Omega]) \cdot 0 + \frac{1}{2} \left|\left| \rho_{S_A S_B E'|\Omega} - \pi_{S_A S_B} \otimes \rho_{E'|\Omega} \right|\right|_1 \\
    & \leq \frac{1}{2} \left|\left| \rho_{S_A S_B E'|\Omega} - \rho_{S_A S_BE'|\Omega \land S_A = S_B}\right|\right|_1  \\
    & ~~+ \frac{1}{2} \left|\left| \rho_{S_A S_BE'|\Omega \land S_A = S_B}- \pi_{S_A S_B} \otimes \rho_{E'|\Omega} \right|\right|_1\\
    &\leq \epsilon_{\mathrm{EC}} + \frac{1}{2} \left|\left| \rho_{S_A S_BE'|\Omega \land S_A = S_B}- \pi_{S_A S_B} \otimes \rho_{E'|\Omega} \right|\right|_1,
\end{align*}
where, for the second inequality, we inserted the definition of $\epsilon_{\mathrm{EC}}$. The last term can be simplified further, taking into account that the input was assumed to be i.i.d. and therefore the two cases 
\begin{enumerate}
    \item[1.)] the test passes and the input is in set $\mathcal{S}^{\mathrm{E\&A}}$
\end{enumerate}
and
\begin{enumerate}
    \item[2.)]  the test passes and the input is not in $\mathcal{S}^{\mathrm{E\&A}}$
\end{enumerate}
are mutually exclusive. We obtain
\begin{align*}
 &\frac{1}{2} \left|\left| \rho_{S_A S_BE'|\Omega \land S_A=S_B}- \pi_{S_A S_B} \otimes \rho_{E'|\Omega} \right|\right|_1   \\
 &\leq \max\left\{ \mathrm{Pr}\left[A\right] \frac{1}{2} \left|\left| \rho_{S_A E'|\Omega}- \pi_{S_A} \otimes \rho_{E'|\Omega} \right|\right|_1 \right., \\
 &~~~~~\left. \mathrm{Pr}\left[ A^c \right]\frac{1}{2} \left|\left| \rho_{S_A E'|\Omega}- \pi_{S_A} \otimes \rho_{E'|\Omega} \right|\right|_1 \right\},
\end{align*}
where $A:=\{\sigma^{\otimes n}: \sigma \in \mathcal{S}^{\mathrm{E\&A}}\}$ and, following the argument in the proof of \cite[Theorem 3.2.5]{Baudry_2015}, we dropped the register $S_B$ since we condition on $S_A=S_B$, which means that the ideal output and the conditioned output have perfectly correlated classical registers, and hence contain redundant information. The second term in the maximum is upper bounded by 
\begin{align*}
    &\mathrm{Pr}\left[A\right]\frac{1}{2} \left|\left| \rho_{S_AE'|\Omega}- \pi_{S_A} \otimes \rho_{E'|\Omega} \right|\right|_1 \\
    &\leq \mathrm{Pr}\left[~\Omega~|~ \sigma^{\otimes n}\notin \mathcal{S}^{\mathrm{E\&A}}\right]\\
    &\leq  \epsilon_{\mathrm{ET}}+\epsilon_{\mathrm{AT}},
\end{align*}
where the first inequality uses the fact that the distinguishability given that Alice and Bob accept the testing is upper bounded by the probability of passing the test, and, for the second inequality we used Theorem~\ref{thm:energy_test} and Theorem~\ref{Thm:ParameterEstimation} which define the set $\mathcal{S}^{\mathrm{E\&A}}$.

It remains to upper bound the first term in the maximum, which refers to the case where $\sigma \in \mathcal{S}^{\mathrm{E\&A}}$ and describes the fact that Alice's and Bob's shared key is only partially secret. This problem is addressed by performing privacy amplification, which is characterised by the leftover hashing lemma \cite[Lemma 5.6.1]{Renner_2005}. We use the version that applies to infinite dimensional side information (Lemma~\ref{lemma:leftoverHashing}). In Lemma~\ref{lemma:leftoverHashing}, we set $\epsilon_{\mathrm{sec}} = \frac{\epsilon_{\mathrm{PA}}}{2} + 2 \epsilon'$ and $\epsilon' = \frac{\bar{\epsilon}}{2}$. Then, for any input $\sigma^{\otimes n}$ with $\sigma \in \mathcal{S}^{\mathrm{E\&A}}$, the output will satisfy
\begin{align*}
    &\frac{1}{2} \left|\left| \rho_{S_AE'|\Omega} - \pi_{S_A} \otimes \rho_{E'|\Omega} \right|\right|_1\\
    &\leq \frac{1}{2} \left|\left| \rho_{S_A S_B E'|\Omega} - \pi_{S_A S_B} \otimes \rho_{E'|\Omega} \right|\right|_1\\
    &\leq \frac{1}{2} \epsilon_{\mathrm{PA}} + 2 \epsilon'\\
    &= \frac{1}{2} \epsilon_{\mathrm{PA}} + \bar{\epsilon}
\end{align*}
as long as we choose 
\begin{equation}
    \ell \leq \min_{\sigma \in \mathcal{S}^{\mathrm{E\&A}}} H_{\mathrm{min}}^{\bar{\epsilon}}(X|E'C)_{\mathcal{E}^{\mathrm{QKD}}\left(\sigma^{\otimes n}\right)} - 2 \log_2\left(\frac{1}{\epsilon_{\mathrm{PA}}}\right).
\end{equation}
Register $C$ denotes the information reconciliation transcript. Therefore, putting things together, we obtain

\begin{align*}
    &\frac{1}{2} \left|\left| \rho_{S_A S_B E'} - \pi_{S_A S_B} \otimes \rho_{E'} \right|\right|_1 \\
    & \leq \epsilon_{\mathrm{EC}} + \max\left\{\frac{1}{2}\epsilon_{\mathrm{PA}}+\bar{\epsilon}, \epsilon_{\mathrm{ET}}+\epsilon_{\mathrm{AT}} \right\} =: \epsilon.
\end{align*}

Lemma \ref{lemma:removingClassRegister} in Appendix \ref{apdx:TechnicalLemmas} extends a statement in \cite[Lemma 6.4.1]{Renner_2005} to infinite dimensional side information and allows us to remove the classical register $C$ containing the transcript of the information reconciliation procedure from the smooth min-entropy at the cost of $\mathrm{leak}_{\mathrm{EC}}$ bits,
\begin{equation}\label{eq:intermedResOnEll}
 \ell \leq \min_{\sigma \in \mathcal{S}^{\mathrm{E\&A}}} H_{\mathrm{min}}^{\bar{\epsilon}}(X|E')_{\mathcal{E}^{\mathrm{QKD}}\left(\sigma^{\otimes n}\right)} - 2 \log_2\left(\frac{1}{\epsilon_{\mathrm{PA}}}\right) -  \mathrm{leak}_{\mathrm{EC}}.  
\end{equation}

Finally, we use Corollary \ref{cor:AEP}, which is our version of the asymptotic equipartition property \cite[Corollary 3.3.7]{Renner_2005}, to rewrite the smooth min-entropy in terms of the von-Neumann entropy
\begin{equation}
\begin{aligned}
  \ell \leq &n \left[\min_{\sigma \in \mathcal{S}^{\mathrm{E\&A}}} H(X|E')_{\mathcal{E}^{\mathrm{QKD}}\left(\sigma^{\otimes n}\right)} - \delta(\bar{\epsilon})  \right] \\
  &-2 \log_2\left(\frac{1}{\epsilon_{\mathrm{PA}}}\right) -  \mathrm{leak}_{\mathrm{EC}}.  
\end{aligned}
\end{equation}

While this completes our finite-size analysis, we want to optimise over finite-dimensional (in more detail: low-dimensional) states. Our energy test (Theorem \ref{thm:energy_test}) guarantees that any state that is not $\epsilon_{\mathrm{ET}}$-filtered has at most weight $w$ outside the cutoff space (defined by parameter $n_c$ in the energy test), and hence satisfies $\Tr{\rho \Pi^{n_c}} = 1-w$. Using Theorem \ref{thm:dimReduction}, we can relate the values of our objective function on inputs from an infinite dimensional Hilbert space to its values on projections onto a finite-dimensional subspace $\mathcal{H}^{n_c}$ by taking an additional weight-dependent correction term $\Delta(w)$ (see Eq.~(\ref{eq:weightCorrection}))into account. Hence, we arrive at

\begin{equation}
\begin{aligned}
    \ell \leq n &\left[\min_{\sigma \in \mathcal{S}^{\mathrm{E\&A}}} H(X|E')_{\mathcal{E}^{\mathrm{QKD}}\left(\sigma^{\otimes n}\right)} - \delta(\bar{\epsilon}) - \Delta(w) \right] \\& -2 \log_2\left(\frac{1}{\epsilon_{\mathrm{PA}}}\right) -  \mathrm{leak}_{\mathrm{EC}}.
\end{aligned}
\end{equation}

Finally, we divide both sides by $N$, the total number of signals sent, and obtain

\begin{equation}
    \begin{aligned}\label{eq:KeyRateFormula}
    \frac{\ell}{N} \leq \frac{n}{N} &\left[\min_{\sigma \in \mathcal{S}^{\mathrm{E\&A}}} H(X|E')_{\mathcal{E}^{\mathrm{QKD}}\left(\sigma^{\otimes n}\right)} - \delta(\bar{\epsilon}) - \Delta(w) \right] \\ &- \frac{2}{N} \log_2\left(\frac{1}{\epsilon_{\mathrm{PA}}}\right) -  \delta_{\mathrm{leak}}^{\mathrm{EC}},
    \end{aligned}
\end{equation}
where we defined $\delta_{\mathrm{leak}}^{\mathrm{EC}} := \frac{\mathrm{leak}_{\mathrm{EC}}}{N}$ (see Section \ref{sec:ErrorCorrection}).
Hence, the key we obtain is $\epsilon_{\mathrm{sec}} = \max\left\{\frac{1}{2}\epsilon_{\mathrm{PA}}+\bar{\epsilon}, \epsilon_{\mathrm{ET}}+\epsilon_{\mathrm{AT}} \right\}$-secret and $\epsilon_{\mathrm{cor}} = \epsilon_{\mathrm{EC}}$-correct, so $\epsilon := \epsilon_{\mathrm{sec}} + \epsilon_{\mathrm{cor}}$-secure, which finishes the proof.

\end{proof}

\section{Numerical Security Proof Method}\label{sec:NumPrfMethod}
Having derived the secure key rate formula and having transformed it into a finite-dimensional optimisation problem, it remains to calculate lower bounds on the secure key rate numerically. It turns out that the optimisation problem in Eq.(\ref{eq:KeyRateFormula}) is a semidefinite program with convex, nonlinear objective function $f:~\mathcal{D}(\mathcal{H}^{n_c}) \rightarrow \mathbb{R}, ~ \sigma \mapsto H(X|E')_{\sigma}$. Since we are interested in finding a reliable lower bound on the secure key rate, it does not suffice to find an approximate solution to this minimisation problem.  Therefore, we apply the numerical method developed in Refs. \cite{Coles_2016,Winick_2018}, which we are going to summarise briefly in what follows.

\subsection{Idea of the numerical method}\label{sec:NumMethod}
The idea of the numerical method is to split the problem into two steps. In the first step, the nonlinear problem is solved approximately, for example by an iterative first-order algorithm like the Frank-Wolfe algorithm \cite{Frank_Wolfe_1956}. We end up with an approximate solution $\rho_{\text{Step 1}}$ on the minimisation problem. This is, however, not a reliable lower bound on the secure key rate. Therefore, we apply step 2, which helps us to transform this suboptimal solution into a reliable lower bound, using a linearisation and SDP-duality theory. We calculate $\nabla f(\rho_{\text{Step 1}})$, the gradient of our objective function at the approximate minimum from step 1, and use a relaxation theorem to formulate an expanded, linearised semidefinite program. This can be seen as lower bounding the (convex) objective function by a hyperplane, tangent at $\rho_{\text{Step 1}}$. To take numerical imprecisions into account, the feasible set is enlarged by some small $\epsilon_{\mathrm{num}}$. Then the dual of this expanded SDP is solved numerically. Because of results from duality theory in semidefinite programming, every feasible point of this dual SDP is a lower bound on the initial optimisation problem. Consequently, we obtain a reliable lower bound on the optimisation problem in Eq.~(\ref{eq:KeyRateFormula}), and hence a reliable lower bound on the secure key rate.

\subsection{Infinite dimensional, asymptotic optimisation}
In this section, we summarise the details of the formulation of the used numerical method for a DM-CV QKD protocol in the asymptotic limit for infinite dimensional Hilbert spaces, following Refs. \cite{Lin_2019, Upadhyaya_2021}. Even though we treat a more general case, this will be helpful to us to understand the formulation of the optimisation problem in the finite-size regime.

As outlined in the protocol description, in the prepare-and-measure picture, Alice chooses one out of $N_{\mathrm{St}}$ coherent states $\Psi_i \in \{ \alpha_0, ..., \alpha_{N_\mathrm{St}-1} \}$ with probability $p_i$ and sends it to Bob. This can be modelled as Alice preparing the pure state

\begin{equation}
\ket{\Psi}_{AA'} = \sum_{i =0}^{N_{\mathrm{St}-1}} \sqrt{p_i} \ket{i}\otimes\ket{\Psi_i},
\end{equation}

 where Alice keeps register $A$ and sends register $A'$ to Bob via the quantum channel $\mathcal{E}_{A'\rightarrow B}$,
 \begin{equation}
\rho_{AB} = (\mathrm{id}_{A} \otimes \mathcal{E}_{A'\rightarrow B})\left(\ket{\Psi}\!\bra{\Psi} \right),
 \end{equation}
 which is under Eve's control. We denote the joint state of Alice, Bob and Eve by $\rho_{ABE}$.
 As Eve cannot access Alice's lab in the source replacement scheme, P\&M schemes are subject to the constraint $\rho_{A} := \sum_{i,j=0}^{N_{\mathrm{St}}-1}\sqrt{p_i p_j} \langle \Psi_j |\Psi_i \rangle |i\rangle \langle j |_{A} $.

We model the postprocessing steps and the key map conducted by Alice and Bob as quantum channel $\Phi$ that stores the resulting key in the classical register $Z$,
\begin{equation}
  \Phi(\rho_{ABE}) := \sum_{z=0}^{N_{\mathrm{St}}-1} |z\rangle\langle z |_{Z} \otimes \mathrm{Tr}_{AB}\left[ \rho_{ABE} \left( \mathbbm{1}_A\otimes R_B^z \otimes \mathbbm{1}_E\right) \right],
\end{equation}
where $R_B^z$ is the so-called region operator, describing the key map on Bob's side (see Figure \ref{fig:Sketch_Keymap}),
\begin{equation}
    R_B^z := \frac{1}{\pi} \int_{\Delta_r}^{\infty} \int_{\frac{2z-1}{N_{\mathrm{St}}} \pi}^{\frac{2z+1}{N_{\mathrm{St}}} \pi} r |r e^{i \phi}\rangle \langle r e^{i \phi}|~d\phi ~dr.
\end{equation}

\begin{figure}
\subfloat[Standard key map. \label{fig:Sketch_Keymap_old} ]{
   \includegraphics[width=0.46\textwidth]{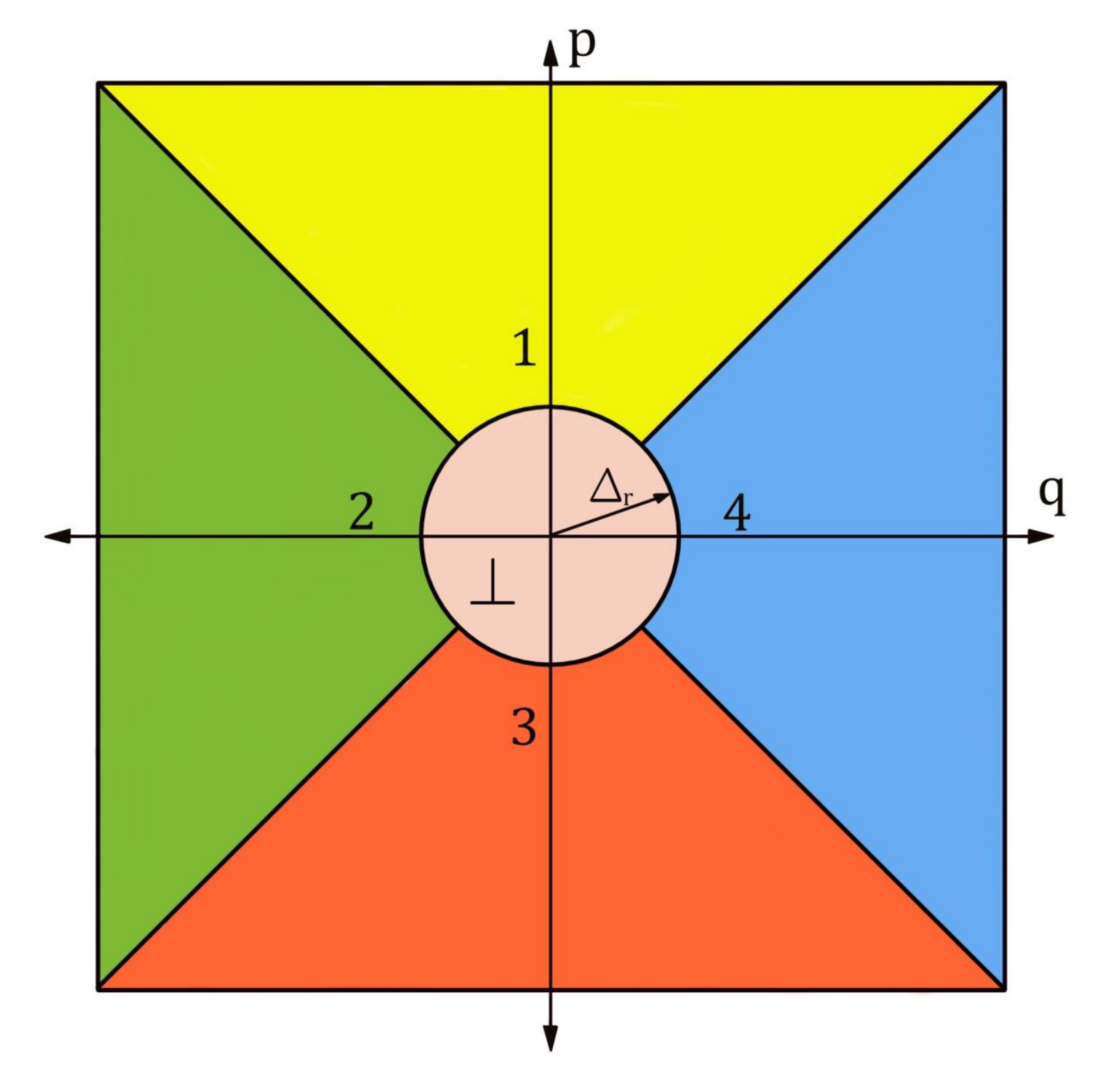}}\\
\subfloat[Modified key map. \label{fig:Sketch_Keymap_new}]{
    \includegraphics[width=0.46\textwidth]{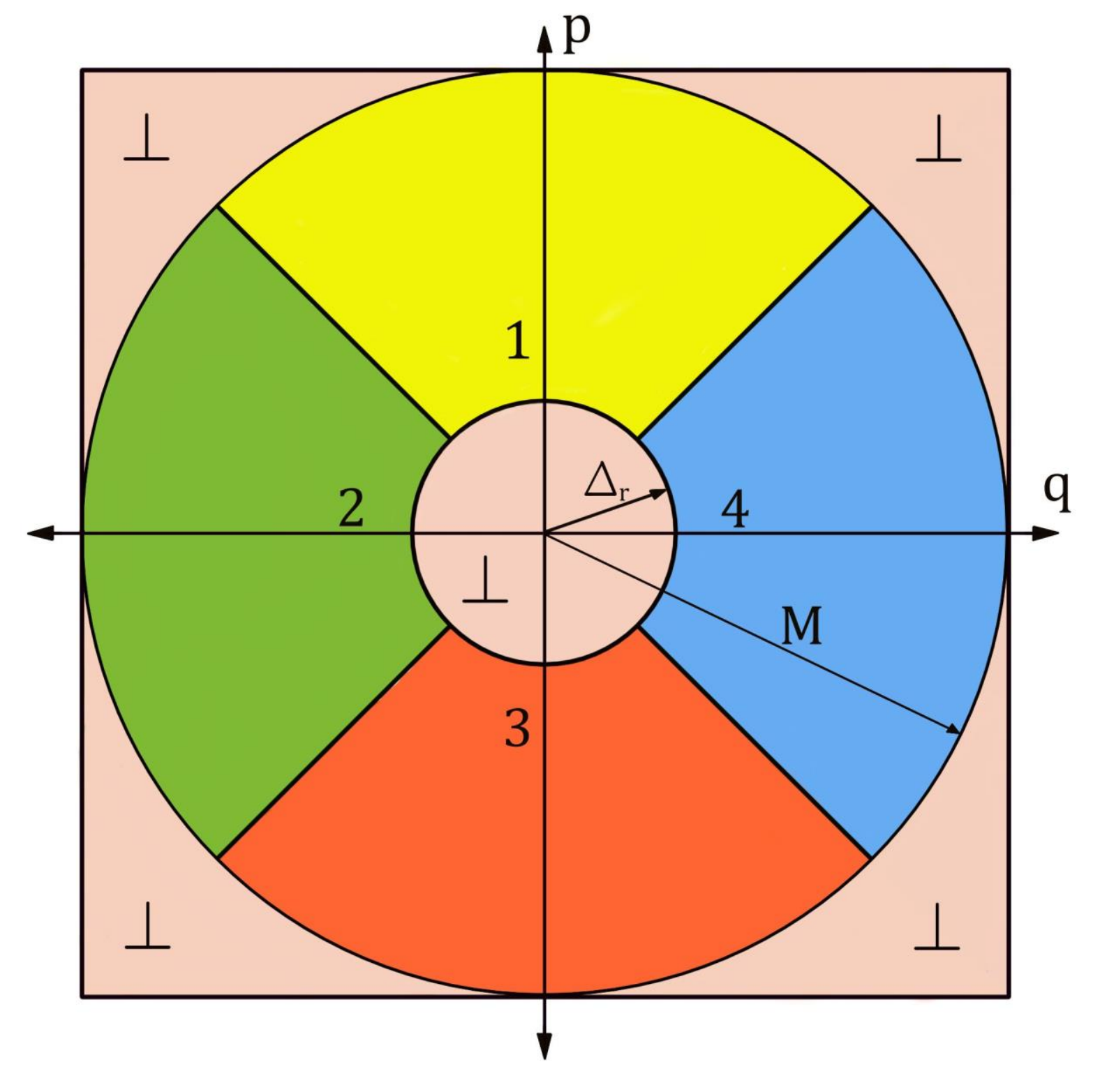}}
\makeatletter\long\def\@ifdim#1#2#3{#2}\makeatother
\caption{ Sketch of the key map in phase space in (a) the standard setting for the ideal protocol, and (b) the modified setting with confined measurement (see Section~\ref{sec:BoundingDetectionRange} and the discussion in Appendix~\ref{sec:APDX:ModifiedKeymap}). The symbol $\perp$ denotes results that are discarded while the shaded areas illustrate which points in phase space are associated with which symbol.\label{fig:Sketch_Keymap}}
\end{figure}

In the asymptotic limit, the secure key rate is given by the Devetak-Winter formula \cite{Devetak_Winter_2006}. Taking realistic error correction into account, this leads to the following expression
\begin{equation}\label{eq:AsymptoticKRformula}
    R^{\infty} = \min_{\rho_{ABE} \in \mathcal{S}^{\infty}} H(Z| E)_{\Phi(\rho_{ABE})} - \delta_{\mathrm{leak}}^{\mathrm{EC}},
\end{equation}
where $\mathcal{S}^{\infty}$ denotes the feasible set of the optimisation to find secure key rates in the asymptotic limit. In what follows, we provide details about this set. 

For ease of notation, we denote the objective function by $f(\rho)$. The set $\mathcal{S}^{\infty}$ is defined by constraints due to Bob's measurements as well as by additional requirements on the quantum state shared between Alice and Bob. As outlined above, we assume Alice's lab is inaccessible to Eve so that her share of the state cannot change during the key-generation process. Next, we take Bob's measurements into account. We generically denote Bob's measurement operators by $\hat{\Gamma}_j$ and the corresponding expected values by $\gamma_j$, where $j \in \{1,..., N_{\mathrm{meas}}\}$ with $N_{\mathrm{meas}}$ being the number of different measurement operators Bob applies.  Additionally, as we optimise over a set of valid density matrices, we require the trace to be equal to one and demand positive semidefiniteness. Then, the generic structure of the optimisation problem reads
\begin{align*}
    \min~ &f(\rho)\\
    \text{subject to }&\\
    & \mathrm{Tr}_{B}\left[\rho\right] = \rho_A\\
    & \Tr{\hat{\Gamma}_j \rho } = \langle \gamma_i \rangle\\
    & \Tr{\rho} = 1\\
    &\rho \geq 0,
\end{align*}
where $j$ runs from $1$ to the number of constraints we introduce. Hence, $\mathcal{S}^{\infty}$ reads
\begin{align*}
    \mathcal{S}^{\infty} := \left\{ \rho \in \mathcal{D}(\mathcal{H}_{AB}):~ \mathrm{Tr}_{B}\left[\rho\right] = \rho_A, \Tr{\hat{\Gamma}_j\rho} = \gamma_j \right\},
\end{align*}
where, to ease the notation, we included the constraint $\Tr{\rho} = 1$ into our set of measurement-induced constraints, by defining $\hat{\Gamma}_0 := \mathbbm{1}$ and the corresponding expected value by $\gamma_0 := 1$. Consequently, we redefine the index set for $j$ as $\{0, ..., N_{\mathrm{meas}}\}$.

As outlined in the protocol description, Bob performs a heterodyne measurement so that he has access to the moments of the received signals. We follow the approach in~\cite{Upadhyaya_2021} to use the photon-number operator $\hat{n}$ and its square $\hat{n}^2$ as Bob's observables and then express our constraints in the displaced number basis since these combinations turned out to give good estimation of the weight when we applied the dimension reduction method. Therefore, $\Gamma_j \in \{\mathbbm{1}, \ket{i}\!\bra{i}\otimes\hat{n}_{\beta_i}, \ket{i}\!\bra{i}\otimes \hat{n}^2_{\beta_i} \}$ and $\gamma_j \in \{1, \langle \hat{n}_{\beta_i} \rangle , \langle \hat{n}^2_{\beta_i} \rangle \}$ for $i \in \{0, ..., N_{\mathrm{St}} - 1 \}$.

As $f$ is a convex function and the feasible set $\mathcal{S}^{\infty}$ is convex, we have a convex optimisation problem, which can be solved using the numerical security proof framework in Refs. \cite{Coles_2016, Winick_2018}.

\subsection{Finite-size optimisation problem}
Note that the objective function of the optimisation in the asymptotic limit (\ref{eq:AsymptoticKRformula}) is the same as for the finite-size problem (\ref{eq:KeyRateFormula}), while the feasible sets differ. Furthermore, there are additional correction terms for the finite-size version of the key rate formula. However, as these terms are constant with respect to the performed optimisation, they do not influence the structure of the SDP.

In the finite-size regime, we do not know the expected values of our observables with certainty. As outlined in the protocol description, we  fix some small $\epsilon_{\mathrm{AT}}>0$ and a testing ratio $r_{\mathrm{test}} \in (0,1)$ such that $k := r_{\mathrm{test}} \cdot N$  and perform testing on $k$ randomly selected rounds. According to Theorem \ref{Thm:ParameterEstimation}, we obtain bounds $\mu_j$ which define our acceptance set. Therefore, our actual optimisation problem reads

\begin{align}
\begin{aligned}\label{eq:HighDimOptProb}
    \alpha := \min ~&f(\rho)\\
    \text{subject to }& \\
    & \mathrm{Tr}_{B}\left[\rho\right] = \rho_A\\
    & \left| \Tr{\hat{\Gamma}_j\rho } - \gamma_j \right| \leq \mu_j\\
    & \Tr{\rho} = 1\\
    & \rho \geq 0    
\end{aligned}
\end{align}
for $j \in \{1, ..., 2 N_{\mathrm{St}}\}$. Note that the constraints $\mathrm{Tr}_{B}\left[\rho\right] = \rho_A$ and $\Tr{\rho} = 1$ are not subject to finite-size effects.

It is shown in Appendix \ref{APDX:FinDimOptProb} that finally, after applying the dimension reduction method, and various steps to bring the SDP to a more favourable form, we obtain the following (primal) optimisation problem

\begin{align}
\begin{aligned}\label{eq:PrimalProblem}
    \beta :=& \min~ f(\bar{\rho})\\
    \text{s.t. }& \\
    & \Tr{P}+\Tr{N} \leq 2 \sqrt{w}\\
    & P \geq \mathrm{Tr}_{B}\left[ \overline{\rho} \right] - \rho_A\\
    & N \geq - \left( \mathrm{Tr}_{B}\left[ \bar{\rho} \right] - \rho_A \right) \\
    &\Tr{\left( \ket{j}\!\bra{j}\otimes\hat{n}_{\beta_j}\right) \bar{\rho} } \geq \mu_j + \langle \hat{n}_{\beta_j} \rangle - w ||\hat{n}_{\beta_j}||_{\infty} \\
    & \Tr{\left( \ket{j}\!\bra{j}\otimes\hat{n}_{\beta_j}\right) \bar{\rho} } \leq \mu_j + \langle \hat{n}_{\beta_j} \rangle  \\
    &\Tr{\left( \ket{j}\!\bra{j}\otimes\hat{n}_{\beta_j}\right) \bar{\rho}} \geq -\mu_j + \langle \hat{n}_{\beta_j} \rangle - w ||\hat{n}_{\beta_j}||_{\infty} \\
    & \Tr{\left( \ket{j}\!\bra{j}\otimes\hat{n}_{\beta_j}\right) \bar{\rho}} \leq -\mu_j + \langle \hat{n}_{\beta_j} \rangle  \\ 
    & \Tr{\left( \ket{j}\!\bra{j}\otimes \hat{n}^2_{\beta_j} \right) \bar{\rho}}\geq \mu_j + \langle \hat{n}_{\beta_j}^2 \rangle - w ||\hat{n}_{\beta_j}^2||_{\infty}\\
    &\Tr{\left( \ket{j}\!\bra{j}\otimes \hat{n}^2_{\beta_j} \right) \bar{\rho}} \leq \mu_j + \langle \hat{n}^2_{\beta_j} \rangle \\
    & \Tr{\left( \ket{j}\!\bra{j}\otimes\hat{n}^2_{\beta_j}\right) \bar{\rho} } \geq -\mu_j + \langle \hat{n}_{\beta_j}^2 \rangle - w ||\hat{n}_{\beta_j}^2||_{\infty} \\
    & \Tr{\left( \ket{j}\!\bra{j}\otimes\hat{n}^2_{\beta_j}\right) \bar{\rho} } \leq -\mu_j + \langle \hat{n}^2_{\beta_j} \rangle  \\ 
    & 1-w \leq \Tr{\overline{\rho}} \leq 1\\
    & \bar{\rho}, P, N \geq 0
\end{aligned}
\end{align}
where $j \in \{0,..., N_{\mathrm{St}}-1\}$ and $a_j$ and $b_j$ denote the $j$-th entry of the vectors $\Vec{a}$ and $\Vec{b}$, respectively.
It remains to solve this SDP numerically to obtain lower bounds on the secure key rate. In the present work, we use the technique introduced in \cite{Winick_2018}, where secure key rates are obtained via the two-step process described in Section \ref{sec:NumMethod}. For the reader's convenience, we derive the corresponding dual problem in Appendix \ref{APDX:FinDimOptProb}.

\subsection{Error correction}\label{sec:ErrorCorrection}

In this subsection, we briefly explain the information-reconciliation leakage term. In the case one is able to carry out the information reconciliation procedure in the Slepian-Wolf limit~\cite{Slepian_Wolf_1973}, the EC leakage term reads
\begin{equation*}
    \delta_{\mathrm{EC}} := H(Y|X) = H(Y) - I(X:Y).
\end{equation*}
Here, $X$ and $Y$ represent Alice's and Bob's key strings. Since we cannot expect to perform error correction in the optimal limit, we assume only a fraction $0 < \beta \leq 1$ of the mutual information between Alice's and Bob's key strings can be used. Hence, $I(X:Y)$ in the formula above is replaced by $\beta I(X:Y)$. Therefore,
\begin{align*}
    \delta_{\mathrm{EC}} \mapsto \delta_{\mathrm{EC}}^{\beta} &:= H(Y) - \beta I(X:Y) \\ &= H(Y) - \beta \left[ H(Y) - H(Y|X) \right]\\
    &= (1-\beta) H(Y) + \beta H(Y|X).
\end{align*}
Finally, the total leakage term is the sum of the correction term we just derived and the verification term. We obtain~\cite{George_2020}
\begin{align}\label{eq:leakageTerm}
\mathrm{leak}_{\mathrm{EC}} \leq n ~\delta_{\mathrm{EC}}^{\beta} + \log_2\left( \frac{2}{\epsilon_{\mathrm{EC}}} \right).
\end{align}
As the present protocol allows postselection, not all signals might be used for signal generation. Hence, not all signals have to undergo the information reconciliation procedure. Therefore, we replace $\mathrm{leak}_{\mathrm{EC}} \mapsto p_{\mathrm{pass}} \mathrm{leak}_{\mathrm{EC}}$, where $p_{\mathrm{pass}}$ is the probability that a round passes the postselection routine.

\subsection{Trusted, nonideal detector approach}
So far, it has been assumed that Bob's detectors are ideal (i.e., $100\%$ detection efficiency and no electronic noise) and we therefore dedicated all noise to Eve. In real-world implementations, detectors are noisy and have detection efficiency smaller than one. The trusted, nonideal detector model introduced in Ref. \cite{Lin_2020} enables us to include realistic detectors in our key rate calculations and allows us to trust those parts of the noise that come from Bob's detection devices. This assumption is reasonable since Bob's detectors are located in his lab, and hence assumed to be inaccessible to Eve.

The idea of the model is to introduce an additional beam splitter in front of every perfect homodyne detector that measure either the $q$ or $p$ quadrature. The transmission is chosen to be equal to the detector efficiencies $\eta_q$ and $\eta_p$. At the second input port of both of those beam splitters, the signal is mixed with a thermal state with mean photon-numbers $\bar{n}_i = \frac{\nu_{\mathrm{el, } i}}{2(1-\eta_i)}$ for $i \in\{q,p\}$. Therefore, the output signals experience electronic noise $\nu_{\mathrm{el,~}q}$ and $\nu_{\mathrm{el,~}p}$, respectively. Finally, two ideal homodyne detectors are used to perform the measurement. For more details regarding the trusted, nonideal detector we refer the reader to Ref. \cite{Lin_2020}. A sketch of the trusted detector scheme can be found in \cite[Figure 2]{Lin_2020}.

\subsection{Bounding the detection range}\label{sec:BoundingDetectionRange}
As outlined in Section \ref{sec:BoundingObservables},  we ensure fast convergence of our acceptance test by constraining the observables to the detection range of the heterodyne detector. In line with the discussion in Ref. \cite{Upadhyaya_2021}, in the ideal (nonrestricted) detector model, operators $\hat{X}$ can be represented as
\begin{equation}\label{eq:OperatorXideal}
    \hat{X}= \int_{\zeta \in \mathbb{C}} f_X(\zeta) \frac{1}{\pi} \ket{\zeta}\!\!\bra{\zeta}~d^2\zeta,
\end{equation}
where $f_{\hat{X}}(\zeta)$ is some scalar-valued function and $\frac{1}{\pi} \ket{\zeta}\!\!\bra{\zeta}$ is the POVM corresponding to an ideal heterodyne measurement. Its noisy counterpart reads
\begin{equation}\label{eq:OperatorXnonIdeal}
    [\hat{X}]' = \int_{\zeta \in \mathbb{C}} f_{\hat{X}}(\zeta) G_{\zeta}~d^2\zeta,
\end{equation}
where $G_{\zeta}$ is the nonideal trusted detector POVM derived in Ref. \cite{Lin_2020}. In order to restrict the measurement results to the interval $\mathcal{M}=[-M,M]^2$, we need to modify function $f_{\hat{X}}$. This involves partitioning the phase space into distinct regions and replacing $f_X(\zeta)$ with $g_X(\zeta)$. Define $\tilde{\mathcal{Q}}_1 := \mathbb{R}_{+}\times\mathbb{R}_{+} \setminus \mathcal{M}$, $\tilde{\mathcal{Q}}_2 := \mathbb{R}_{-}\times\mathbb{R}_{+} \setminus \mathcal{M}$, $\tilde{\mathcal{Q}}_3 := \mathbb{R}_{-}\times\mathbb{R}_{-} \setminus \mathcal{M}$ and  $\tilde{\mathcal{Q}}_4 := \mathbb{R}_{+}\times\mathbb{R}_{-} \setminus \mathcal{M}$. Function $g_{\hat{X}}(\zeta)$ takes into consideration the finite detection range that has been proposed,
\begin{equation}
\begin{aligned}
    &g_{\hat{X}}(\zeta_x, \zeta_y):= \\
    & \begin{cases}
    f_{\hat{X}}(\zeta_x, \zeta_y)  & \text{if } (\zeta_x, \zeta_y)  \in \mathcal{M},\\
    f_{\hat{X}}(\min\{\zeta_x, M\}, \min\{\zeta_x, M\})  & \text{if } (\zeta_x, \zeta_y)  \in \tilde{\mathcal{Q}}_1,\\
    f_{\hat{X}}(\max\{\zeta_x, -M\}, \min\{\zeta_x, M\})  & \text{if } (\zeta_x, \zeta_y)  \in \tilde{\mathcal{Q}}_2,\\
    f_{\hat{X}}(\max\{\zeta_x, -M\}, \max\{\zeta_x, -M\})  & \text{if } (\zeta_x, \zeta_y)  \in \tilde{\mathcal{Q}}_3,\\
    f_{\hat{X}}(\min\{\zeta_x, M\}, \max\{\zeta_x, -M\})  & \text{if } (\zeta_x, \zeta_y)  \in \tilde{\mathcal{Q}}_4,\\
    \end{cases}
\end{aligned}
\end{equation}
where $\zeta_x$ denotes the real part of $\zeta$, while $\zeta_y$ denotes the imaginary part of $\zeta$. In the current protocol, we perform measurements of $\hat{n}$ and $\hat{n}^2$ in the displaced number basis. We derive expressions for the observations of $[\hat{n}]_b'$ and $[\hat{n}^2]_b'$, which are the bounded and noisy equivalents of $\hat{n}$ and $\hat{n}^2$, in Appendix~\ref{APDX:Bounded_Measurement}.

With our observables now being bounded, we can readily observe that we obtain $x_{\hat{n}} = M^2-\frac{1}{2}$ and $x_{\hat{n}^2} = M^4-\frac{1}{2}M^2$ for the constants involved in Theorem \ref{Thm:ParameterEstimation}. 


\section{Results}\label{sec:Results}
\subsection{Quadrature phase-shift keying protocol}
To provide numerical key rates, we restrict our proof for general discrete-modulated CV-QKD protocols to the special case of $N_{\mathrm{St}} = 4$ signal states arranged on a circle in the phase space, a so-called quadrature phase-shift keying protocol. Therefore, in every round, Alice prepares one of the states $\{|\alpha\rangle, | i\alpha\rangle, |-\alpha\rangle, |-i \alpha\rangle  \}$ with equal probability, where $\alpha \in \mathbb{R}$ is arbitrary but fixed. Bob then performs heterodyne detection on the states he receives. While our security proof works for both direct- and reverse reconciliation, we proceed with reverse reconciliation that is known to outperform direct reconciliation for CV-QKD protocols in the long-distance regime. Therefore, Bob performs the key map and assigns symbols to his measurement results, depending on which area of phase space the measurement outcomes lie. This includes the option of performing postselection to increase the key rate. For more details regarding the protocol, we refer the reader to \cite[Protocol 2]{Lin_2019}. Since our description of the numerical method in Section \ref{sec:NumMethod} was general, the expressions there apply to the present special case if we choose $N_{\mathrm{St}} = 4$. 

\subsection{Choice of the weight}\label{sec:choice_of_weight}
In our security proof, the weight $w = \Tr{\rho \Pi^{\perp}}$ plays a twofold role. On the one hand, it appears as a parameter in the energy test, while on the other hand, it determines the size of the correction term $\Delta(w)$ arising from the dimension reduction method. While the asymptotic dimension reduction method gives a bound on the weight via another semidefinite program, in our case $w$ is chosen freely during the energy test. This means that, in principle, one could choose the weight arbitrarily small, resulting in a negligible correction term without corrupting our security statement (possibly resulting in a large $\epsilon_{\mathrm{ET}}$). However, since the energy test only makes a statement in the case when the test passes and aborts otherwise (in which case it is trivially secure), this comes at the cost of a high failure rate of the energy test, hence ultimately a low average key rate. Therefore, the choice of the weight $w$ is a balancing act between aiming for a low correction term and making the energy test pass with high probability. In order to assure that, we required that the energy test passes with high probability in the honest implementation, i.e., when Eve is passive. Therefore, we modelled the quantum channel connecting Alice and Bob as a noisy and lossy Gaussian channel with excess noise $\xi$ and transmittance $\eta$ and calculated the expected weight $w_{\mathrm{exp}}$ outside the cutoff space. Then, one possible choice for the weight is $w \geq w_{\mathrm{exp}}$. We want to highlight that this was a choice motivated by practicality and is not a requirement of the security proof. Alternatively, we may fix $\epsilon_{\mathrm{ET}}$ and just solve the expression for $\epsilon_{\mathrm{ET}}$ obtained from the energy testing theorem (Theorem \ref{thm:energy_test}) for $w$ to obtain $w_{\epsilon}$. In practice, we introduce a minimal weight $w_{\mathrm{min}}$ and choose the weight $w := \max \{w_{\mathrm{exp}}, w_{\epsilon}, w_{\mathrm{min}} \}$ to make sure it is both compatible with the chosen $\epsilon_{\mathrm{ET}}$ and large enough such that the energy test passes with high probability on the honest implementation.

\begin{figure}
\includegraphics[width=0.48\textwidth]{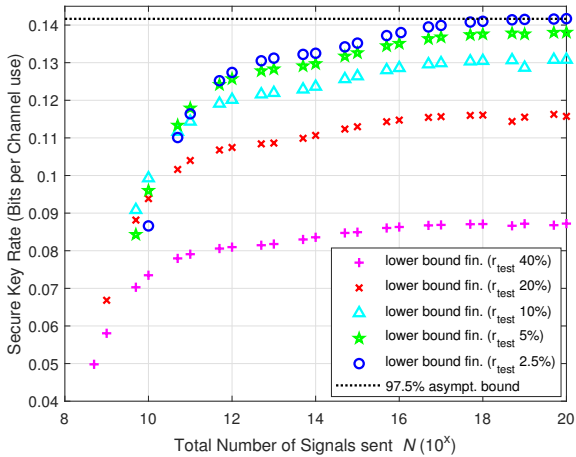}
\caption{Secure key rates over total number of signals sent $N$ for $L= 10$km, $\alpha = 0.85$, $\Delta_r = 0.45$ for ideal, untrusted detectors. \label{fig:KRoverN} }
\end{figure}

\subsection{Details about the implementation}
Before we come to our numerical results, we briefly discuss our choice of parameters and some technical details. 
To demonstrate the performance of the chosen quadrature phase-shift keying protocol under our finite-size security proof, we simulate the expectation values (see Eqs.~(\ref{eq:HighDimOptProb}) and optimisation problems derived thereof) obtained from an experiment by modelling Alice's coherent states passing a noisy and lossy Gaussian channel with excess noise $\xi$ and channel transmittance $\eta$. The excess noise is understood as preparation noise on Alice's side so that it is taken to be fixed at the input of the channel. Hence, Bob experiences the effective noise $\eta \xi$. Note that we measure the noise in the shot noise units. Within the whole work, our transmittance model as a function of the transmission distance $L$ is $\eta = 10^{-0.02 L}$. This corresponds to a transmission of $-0.2$ dB/km that is a common value for optical fibres at the telecom wavelength. 

While the total number of transmitted signals $N$, as well as the testing ratio $\frac{k_T}{N}$ varies, we fix $l_T/k_T$ (see Theorem~\ref{thm:energy_test}) to be $10^{-8}$ and $M=5$. Furthermore, we fix the $\epsilon$ parameters to be 
$\epsilon_{\mathrm{EC}} = \frac{1}{5}\times 10^{-10}$, $\epsilon_{\mathrm{PA}} = \frac{1}{5}\times 10^{-10}$, $\bar{\epsilon} = \frac{7}{10}\times 10^{-10}$, $\epsilon_{\mathrm{AT}} = \frac{7}{10}\times 10^{-10}$ and $\epsilon_{\mathrm{ET}} = \frac{1}{10}\times 10^{-10}$ such that the total security parameter (see Theorem \ref{thm:SecurityStatement}) is $\epsilon = 10^{-10}$. We emphasise that our security proof is independent of the choice of parameters and that those values are chosen for demonstration purposes only.

We applied the numerical framework in \cite{Coles_2016, Winick_2018} to find a lower bound on the minimisation problem in Eq.~(\ref{eq:PrimalProblem}), where the coding was carried out in \textsc{Matlab}\textsuperscript{\textregistered}, version R2020a. The semidefinite programs were modelled using CVX \cite{cvx1, cvx2}, where we used the MOSEK solver (version 9.1.9) \cite{mosek} to solve the semidefinite programs.

\subsection{Simulation Results}

\begin{figure}
\includegraphics[width=0.48\textwidth]{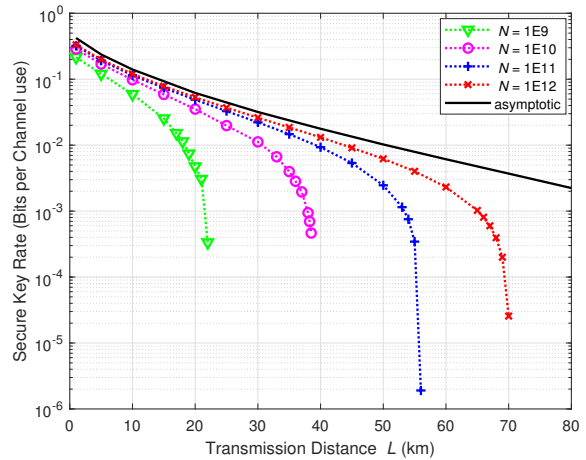}
\caption{Secure key rates over transmission distance $L$ for different total number of signals $N$. We optimised the coherent state amplitude $\alpha$ and the radial postselection parameter $\Delta_r$ and fixed testing ratio $r_{\mathrm{test}} = 10\%$. All curves correspond to ideal, untrusted detectors. \label{fig:KRoverL_diff_N} }
\end{figure}

We present plots of the obtained secure key rates for various parameter choices. If not mentioned otherwise, we fix the preparation noise $\xi = 0.01$ and in all plots, we assume that an error correction code with efficiency $\beta = 0.95$ is used, which is achievable with the latest low-density parity-check codes. We note that it is not entirely clear if constant $\beta$ is also achievable for wide ranges of SNR. However, our security proof method is independent of the particular $\beta$ and  for illustration purposes we fixed it to $0.95$, in accordance with common values used in the literature. If we do not state a particular value for the amplitude $\alpha$ and the postselection parameter $\Delta_r$, the corresponding curves have been obtained after optimising over $\alpha$ and $\Delta_r$ via a coarse grained search. We chose the cutoff space dimension $n_c = 20$, which turned out to be a sound compromise between numerical feasibility (calculation time) and impact on the obtained key rates (see the role of the cutoff number in the security proof in Section~\ref{sec:SecProof}).

In the first two subsections, we present plots in the unique-acceptance scenario (see Section \ref{sec:parameter_estimation}), which is standard in the literature and allows for comparison. We start by discussing our results for untrusted, ideal detectors (so $\eta_d = 1$ and $\nu_{\mathrm{el}} = 0$), which is followed by results for trusted, nonideal detectors ($\eta_d <1$ and $\nu_{\mathrm{el}} > 0$). However, as elaborated on after Theorem \ref{Thm:ParameterEstimation}, in the unique-acceptance scenario, practical protocols will abort with probability close to $1$. Therefore, in the final section, we briefly discuss the nonuniequ-acceptance scenario and present key rates for this practical and realistic case.

\subsubsection{Untrusted, ideal detectors}
In what follows, we present our results for untrusted, ideal detectors. The key rates shown are measured in bits per channel use and the plotted asymptotic key rate curves were generated with the method described in Ref. \cite{Upadhyaya_2021}.

Figure \ref{fig:KRoverN} shows the obtained secure key rates over the total number of signals sent $N$. We fixed the transmission distance to be $10$ km, the coherent state amplitude $\alpha = 0.85$ and the radial postselection parameter $\Delta_r = 0.45$, while we varied the testing ratios (TR). As one can see, we obtain secure key rates for $N \geq 5\times10^8$ for $r_{\mathrm{test}} = 40\%$. Furthermore, our secure key rates approach the asymptotic limit from Ref. \cite{Upadhyaya_2021} for $N\rightarrow \infty$ and low testing ratios. This shows that our analysis is tight in the asymptotic limit. We note that we had to adapt the asymptotic key rate curve in Figure \ref{fig:KRoverN} compared to Ref. \cite{Upadhyaya_2021} because of different weights, and hence different correction terms $\Delta(w)$. The reason behind this is as follows. The weight in the asymptotic regime without testing is determined by solving an additional SDP, and is hence fundamentally different than in our analysis including an energy test (see also the discussion in Section \ref{sec:choice_of_weight}). Our statistical approach allows us to work with smaller weights, and hence smaller correction terms. In order to make the key rate curves comparable, one therefore has to readjust the asymptotic curves in Ref. \cite{Upadhyaya_2021} by the weight correction.

Next, we consider the performance of our secure key rates as a function of the transmission distance for a different number of total rounds $N$ in Figure \ref{fig:KRoverL_diff_N}. We fix the testing ratio to $r_{\mathrm{test}} = 10\%$. Again, we note that for the asymptotic key rates, we do not effectively sacrifice signals for testing. Hence the asymptotic key rates are conceptionally different to the finite-size key rates in the plot and would correspond to finite-size key rates with a testing ratio equal to $0\%$. This explains the tiny difference in key rates between the asymptotic reference curve and the finite-size key rates for low transmission distances.

Our observations from Figure \ref{fig:KRoverN} indicate that it is unlikely positive key rates are obtained for $N$ smaller than $N=5\times 10^8$ at $L=10$ km. Therefore we start our investigation at $N = 10^9$ in Figure \ref{fig:KRoverL_diff_N}, where we have hope to surpass $L=10$ km significantly and go up to $N=10^{12}$, which is the largest $N$ we assume is achievable in experiments with state-of-the-art lasers and heterodyne detectors in a practical amount of time. Note that we optimised over the coherent state amplitude $\alpha$ and the postselection parameter $\Delta_r$ via coarse grained search. We observe positive key rates up to $22$ km for $N=10^{9}$, up to $38.5$ km for $N=10^{10}$, up to $56$ km for $N=10^{11}$ and up to $70$ km for $N=10^{12}$.

\begin{figure}
\includegraphics[width=0.48\textwidth]{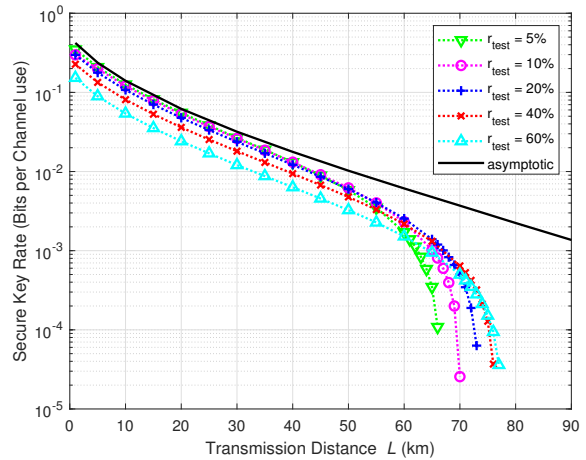}
\caption{Secure key rates over transmission distance $L$ for fixed $N = 10^{12}$, optimised the coherent state amplitude $\alpha$ and the radial postselection parameter $\Delta_r$ and different testing ratios $r_{\mathrm{test}}$. \label{fig:KRoverL_diff_TR} }
\end{figure}

It remains to discuss how much we can improve our results by varying the testing ratio $r_{\mathrm{test}}$. In Figure \ref{fig:KRoverL_diff_TR}, we fix $N=10^{12}$, optimise over $\alpha$ and $\Delta_r$ via a coarse grained search and examine the impact of testing ratios between $5\%$ and $60\%$. As expected, it turns out that for low transmission distances, low testing ratios are advantageous, while the maximal achievable transmission distance can be improved significantly by increasing the fraction of signals used for testing. This is because for high transmission distances the expectation values in our constraints become small, and hence (for the same testing as for lower distances) their uncertainties become relatively large. Higher testing counteracts this effect and increases the secure key rates. Sacrificing $60\%$ of the signals for testing increases the maximal achievable transmission distance from $66$ km (for $5\%$ testing) to $77$ km.

\subsubsection{Trusted, nonideal detectors}
Next, we present our results for the case of trusted, nonideal detectors. For demonstration purposes we choose $\eta_d = 0.72$ and $\nu_{\mathrm{el}} = 0.04$, and emphasise that our analysis is not restricted to this choice. We fix the excess noise again to $\xi = 0.01$. Note that this means that the curves for trusted, nonideal detectors have a higher total noise level compared to the curves for untrusted, ideal detectors in the previous section. Again, we add asymptotic key rate curves, derived following the method presented in Ref. \cite{Upadhyaya_2021}, for comparison. Like in the untrusted, nonideal case, our key rates are tight, i.e. for low testing ratio $r_{\mathrm{test}}$ and a high number of rounds $N$, the obtained finite-size key rates converge to the asymptotic limit. 

We examine the performance of our security proof for different total numbers of rounds, while we fix the testing ratio at $10\%$ and optimise over the coherent state amplitude $\alpha$ and the radial postselection parameter $\Delta_r$ via a coarse grained search. The resulting key rate curves can be seen in Figure \ref{fig:KRoverL_diff_N_trusted}. We see that, as expected, the secure key rates are lower than for the untrusted, ideal detector, but the maximal achievable transmission distances decrease only moderately compared to the untrusted detector with the same excess noise level. We observe positive key rates up to $22$ km (compared to $24$ km for untrusted, ideal detectors), for $N=10^9$ signals, we obtain non-negative key rates up to $39$ km (compared to $41$ km) for $N=10^{10}$, up to $55$ km (compared to $58$ km) for $N=10^{11}$, and up to $67$ km (compared to $71$ km) for $N=10^{12}$f.

\begin{figure}
\includegraphics[width=0.48\textwidth]{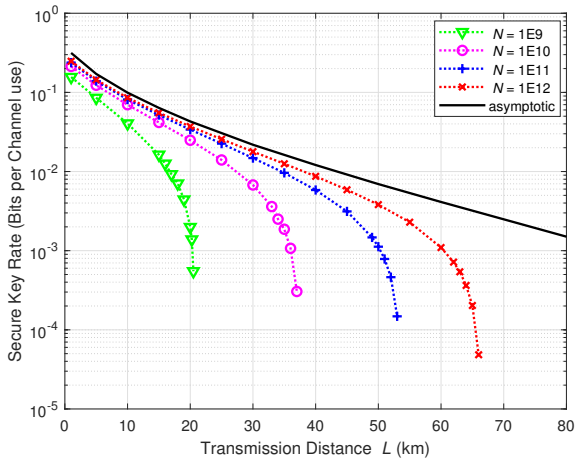}
\caption{Secure key rates over transmission distance $L$ for trusted, nonideal detector  for with $\nu_{\mathrm{el}} = 0.04$ and $\eta_d = 0.72$. We plot key rates for different total number of signals $N$ for optimised the coherent state amplitude $\alpha$ and the radial postselection parameter $\Delta_r$ and fixed testing ratio $r_{\mathrm{test}} = 10\%$. \label{fig:KRoverL_diff_N_trusted} }
\end{figure}

In Figure \ref{fig:KRoverL_diff_TR_trusted}, we plot the obtained secure key rates as a function of the transmission distance $L$ for different testing ratios, while we fix $N=10^{12}$ and optimise over the coherent state amplitude $\alpha$ and the radial postselection parameter $\Delta_r$. As expected the obtained secure key rates are lower than those for the untrusted, ideal detector.

However, for an excess noise level of $\xi = 0.01$, it turns out that the maximal achievable transmission distances do not differ significantly in the trusted detector scenario. For example, when the testing rate is $60\%$ of the signals, the maximal achievable transmission distance for the trusted, nonideal detector is $72$ km while in the untrusted, ideal detector case we obtained $77$ km. For a testing ratio of $5\%$, the maximal achievable transmission distance differs by only $3$ km. The achieved secure key rates in the nonideal detector case are merely lower. Therefore, even for realistic detectors, our method yields practically relevant secure finite-size key rates. We note that this moderate performance difference between key rates using ideal, untrusted detectors and noisy, trusted detectors has already been observed for the asymptotic case in \cite[Section 5.3]{Upadhyaya_Thesis_2021}. The reason behind this is that Bob's noisy observables can be related to his ideal observables by linear combinations. Hence, effectively, the feasible set remains unchanged, while only the objective function changes due to different POVM elements for the noisy, nonideal heterodyne detector. The error correction cost, however, is slightly higher, which explains the observed drop in the secure key rate.

\begin{figure}
\includegraphics[width=0.48\textwidth]{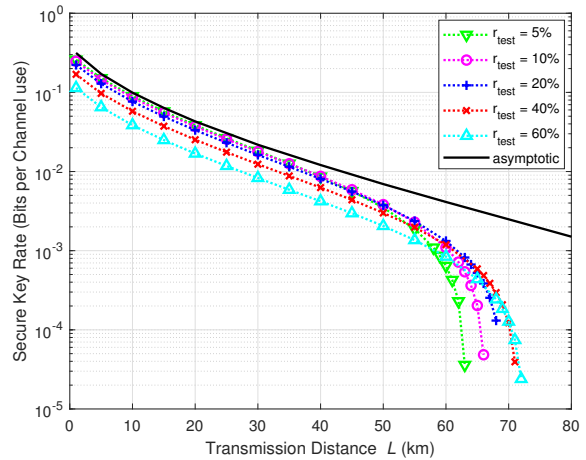}
\caption{Secure key rates over transmission distance $L$ for a trusted, nonideal detector with $\nu_{\mathrm{el}} = 0.04$ and $\eta_d = 0.72$. We fixed $N = 10^{12}$, optimised the coherent state amplitude $\alpha$ and the postselection parameter $\Delta_r$ and examined different testing ratios $r_{\mathrm{test}}$. \label{fig:KRoverL_diff_TR_trusted} }
\end{figure}

\subsubsection{Nonunique acceptance}\label{sec:nonUA}
While it is common in the literature to discuss secure key rates in the unique-acceptance (UA) scenario (where $\mathbf{t}$ in Theorem \ref{Thm:ParameterEstimation} is set to zero), we want to emphasize that the acceptance test of such protocols basically always fails, even in the absence of eavesdroppers. Consequently, although these protocols can achieve high key rates when successful, the expected key rate per key generation round is generally low in practice. Therefore, we turn our attention to the more practical scenario of nonunique acceptance (nonUA), where $\mathbf{t}>0$. Our goal is to investigate the relationship between secure key rate and acceptance probability, which leads to a more useful presentation of secure key rates in practical settings.

Therefore, recall the following results from Section~\ref{sec:parameter_estimation} to gain insights into how the choice of $\mathbf{t}$ influences the secure key rate and the acceptance probability. According to Eq. (\ref{eq:ATset}) the acceptance set grows larger when we choose $\mathbf{t}>0$. Consequently, the optimization performed when solving the key rate finding problem is carried out over a larger set, resulting in lower secure key rates compared to the unique-acceptance scenario. However, Proposition \ref{prop:completeness-bound} provides bounds on the failure probability of the energy test, acceptance test, and the entire QKD protocol (through the union bound). Intuitively, as the sample size increases, we can choose a smaller $\mathbf{t}$. Hence, for illustration purposes, we set $t_X = t_F \mu_X$ for different values of $t_F\geq0$, as this yields 
\begin{equation*}
    \Pr[\mathsf{AT~ Aborts}|\mathsf{Honest}]  \leq 2 |\Theta| \left(\frac{\epsilon_{AT}}{2}\right)^{\frac{t_F^2}{4}},
\end{equation*}
where $\Theta$ denotes the set of observables used in the protocol and $X\in\Theta$.
We want to highlight that this is only a choice and might not be optimal. Further optimizations are left for future work. Furthermore, for the second expression in Proposition \ref{prop:completeness-bound}, we use $D\left( P_{l_T+1}||Q_{\sigma} \right) \geq D(P_{l_T+1} ||Q_{\frac{w}{r}})$ and obtain
\begin{equation*}
\begin{aligned}
    &\Pr[\mathsf{ET~ Aborts}|\mathsf{Honest}]   \\
    &~\leq(k_T - l_T -1) \left( \frac{1-\frac{l_T+1}{k_T}}{1-\frac{w}{r}} \right)^{k_T-l_T-1} \left( \frac{\frac{l_T+1}{k_T}}{\frac{w}{r}}\right)^{l_T+1},
    \end{aligned}
\end{equation*}
where $l_T, k_T, w$ and $r$ are from Theorem \ref{thm:energy_test}. While this bound is sufficient for illustration purposes, we want to note that it is quite loose and we leave tighter bounds for future work. 

To summarize, we have observed that different choices of $\mathbf{t}$ simultaneously impact the acceptance set, and hence the secure key rate, and the acceptance probability. Consequently, in the nonunique acceptance scenario, direct comparisons of the secure key rate for different $\mathbf{t}$ values do not provide meaningful insights, as the expected secure key rate (weighted by the success probability) can vary significantly. Therefore, in this subsection, we introduce a slight modification in how we present our results. Instead of plotting the secure finite-size key rate in bits per channel use, denoted as $\frac{\ell}{N}$, which we obtained from our security proof and have used thus far, we now plot the expected secure key rate per channel use $(1-\nu^c_{\mathrm{QKD}})\times \frac{\ell}{N}$ on the $y$-axis. Thereby, the acceptance probability is calculated assuming that the adversary behaves honestly. We believe that this revised representation of secure key rates better captures the practical relevance, describing the usable and accessible secure key rate in implementations of the investigated protocol. Our intention is to encourage the community to adopt similar reporting methods in future work.

We are now prepared to present and discuss the key rate plots for the nonunique acceptance scenario. Similar to previous sections, we set $\epsilon_{EC} = \frac{1}{5}\times 10^{-10}$ and keep $M=5$ fixed. In what follows, we mainly use the more natural quantity $p_{\mathrm{succ}} := 1-\nu^c_{\mathrm{QKD}}$, which is the `success probability on honest runs' of the analysed protocol. First, we examine the impact of different parameters, specifically $t_F$ (and consequently different acceptance probabilities), on the expected secure key rate. To maintain consistency with Figure \ref{fig:KRoverN}, we set $L=10$km, $\alpha = 0.85$, $\Delta_r = 0.45$, and $r_{\mathrm{test}} = 2.5\%$. We investigate three values of $t_F$, namely ${0.760, 0.832, 1.110}$, which correspond to success probabilities exceeding $50\%$, $75\%$, and $99\%$, respectively. For comparison, we plot the unweighted unique acceptance ($t_F=0$) key rates, along with the asymptotic secure key rate provided in Figure \ref{fig:KRoverN}. Notably, as $N$ grows large, the expected secure key rates for $t_F=1.110$, corresponding to a protocol success probability of $99\%$, closely resemble the nonunique acceptance key rates and the asymptotic secure key rate. This observation underscores the tightness of our key rates even in the nonunique acceptance case.

\begin{figure}
\includegraphics[width=0.48\textwidth]{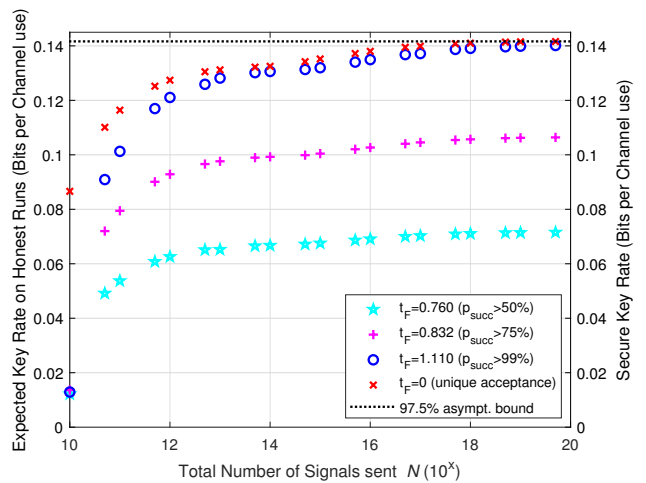}
\caption{Comparison of different expected nonunique acceptance key rates (left $y$-axis) over total number of signals sent for untrusted, ideal detectors and $L=10$km, $\alpha = 0.85$, $\Delta_r = 0.45$ and $r_{\mathrm{test}} = 2.5\%$. As explained in the main text, we plot the expected key rate $(1-\nu^c_{\mathrm{QKD}})\times \frac{\ell}{N}$ for the nonunique acceptance curves ($t_F \in {0.760, 0.832, 1.110}$). For comparison, we also plot the unique acceptance key rates, $t_F = 0$ and asymptotic key rate, (both right $y$-axis) known from Figure~\ref{fig:KRoverN}. \label{fig:nonUA_KR_over_N} }
\end{figure}

Next, we analyze the impact of the nonunique acceptance scenario on the achievable transmission distance. We set $N=10^{12}$ and $r_{\mathrm{test}} = 10\%$ and optimize over $\alpha$ as well as the postselection parameter $\Delta_r$. We consider four values of $t_F$, specifically $t_F \in {0.760, 0.832, 1.110, 1.270}$, which correspond to success probabilities exceeding $50\%$, $75\%$, $99\%$, and $99.9\%$ respectively. Additionally, we plot the unique acceptance key rates ($t_F=0$) and the asymptotic secure key rates from Figure \ref{fig:KRoverL_diff_N} for comparison.

We observe that the expected secure key rates for short to medium transmission distances are close to the unique acceptance key rate, particularly for $t_F = 1.110$ and $t_F=1.270$. While the expected secure key rates for low to medium transmission distances are close to the unique acceptance key rate, in particular for $t_F = 1.110$ and $t_F=1.270$, the achievable transmission distances drop slightly to $61$km for $t_F=1.270$, $62$km for $t_F = 1.110$, $63$km for $t_F = 0.832$ and $64.5$km for $t_F=0.714$, from $70$km in the unique acceptance case.

This demonstrates that, at the expense of lower expected secure key rates, it is possible to increase the maximum achievable transmission distance towards those of the unique acceptance key rate. We expect that a tighter bound on $\Pr[\mathsf{AT~ Aborts}|\mathsf{Honest}]$ would close this small remaining gap, allowing for smaller values of $t_F$ with equal success probabilities, as our current bound overestimates the protocol failure probability. This, in turn, would result in higher key rates and increased achievable transmission distances.

\begin{figure}
\includegraphics[width=0.48\textwidth]{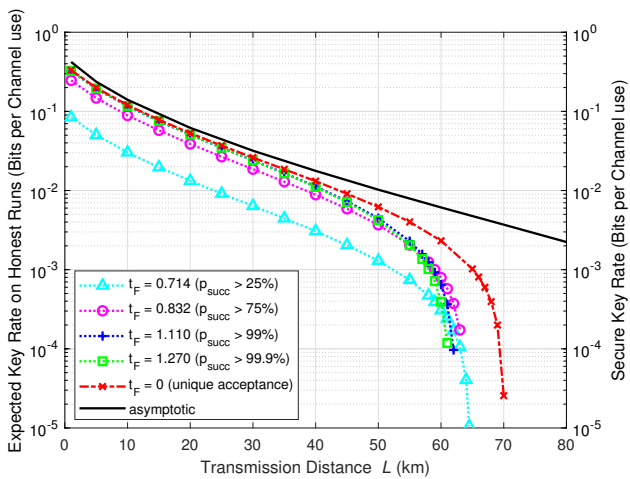}
\caption{Secure nonunique acceptance key rates over transmission distance $L$ for untrusted, ideal detectors. We fixed $N = 10^{12}$ and the testing ratio $r_{\mathrm{test}} = 10\%$ and optimised the coherent state amplitude $\alpha$ as well as the postselection parameter $\Delta_r$. As explained in the main text, for nonunique acceptance curves ($t_F \in {0.760, 0.832, 1.110, 1.270}$) we plot the expected secure key rate $(1-\nu^c_{\mathrm{QKD}})\times \frac{\ell}{N}$, while we report secure key rates for unique acceptance curves ($t_F=0$ and asymptotic). Thus, dotted curves refer to the left, while the dash-dot and the solid curves refer to the right $y$-axis. \label{fig:nonUA_KR_over_L} }
\end{figure}


\section{Conclusion}\label{sec:conclusion}
In our work, we established a composable security proof against i.i.d. collective attacks in the finite-size regime. We tackled the problem of infinite dimensions by introducing a new energy test (Theorem \ref{thm:energy_test}) to bound the weight outside a finite-dimensional subspace and applying the dimension reduction method \cite{Upadhyaya_2021} to take the influence of the weight correction term into account. Furthermore, we argued that in the finite-size regime acceptance testing is the suitable statistical treatment, rather than parameter estimation, known from asymptotic security analyses. We rigorously extended the epsilon security proof method of Ref. \cite{Renner_2005} to handle infinite dimensional side information and finally extend the numerical security proof framework in Refs. \cite{Coles_2016, Winick_2018} to obtain tight lower bounds on the finite-size key rates for a general DM CV-QKD protocol. Furthermore, our security analysis is capable of taking detector imperfections and limitations into account and offers the opportunity to trust Bob's detection devices. 

For illustration, we apply our security proof method to a four-state phase-shift keying protocol and calculate the achievable secure key rates in various scenarios. However, we emphasise that our approach is not limited to four signal states or phase-shift keying modulation but applies to general discrete modulation patterns. We show that under experimentally viable conditions one can obtain positive finite-size key rates up to at least $73$~km transmission distance for moderate to low noise. Through a comprehensive and detailed analysis of the success probability in an honest implementation, we are able to provide a clear and thorough examination of DM CV-QKD protocols. This enables us to report expected secure key rates, rather than solely focusing on achievable secure key rates in cases where the protocol does not abort. Additionally, it allows us to discuss the three crucial aspects of DM CV-QKD protocols, namely security, key rate, and success probability, together in a coherent manner.

Let us take this opportunity to discuss an alternative composable finite-size security proof for DM CV-QKD protocols against i.i.d. collective attacks given in Ref. \cite{Lupo_2022}. The authors of that work used a proof method based on the extremality of Gaussian states and developed an interesting way to leverage the finite detection range of realistic detectors to bound the dimension of the problem. While our work also considers the finite detection range of realistic detectors, we want to highlight that the weight, and hence the bound for the cutoff space comes from the energy test and does not directly rely on the detection limit. This gives us additional flexibility and allows us to achieve small weights and a smaller impact of the detection limit on the secure key rate. However, despite this shared aspect, the security argument is very different, making a direct comparison of the obtained key rates is not straightforward. It was already shown in Ref. \cite{Lin_2019} that the asymptotic key rates obtained using the framework of Refs. \cite{Coles_2016, Winick_2018} yield significantly better lower bounds than those in Ref. \cite{Ghorai_2019} which is another numerical approach employing Gaussian extremality. Lupo and Ouyang \cite{Lupo_2022} compared their QPSK key rates with the analytical key rates given in Ref. \cite{Denys_2021}, which are known not be tight for four signal states (and known to be lower than the key rates by Ref. \cite{Lin_2019}). As our key rates converge for large block sizes against the asymptotic key rates given by \cite{Lin_2019}, one can nevertheless conclude that our method achieves clearly higher secure key rates than the recently published finite-size security analysis in Ref. \cite{Lupo_2022}. Additionally, our work also takes the success probability of the examined protocol into account, allowing to report practically relevant expected secure key rates. However, a direct comparison of both methods to achieve bounded operators, and hence finite dimensional problems, using the same security proof framework and similar assumptions on the detectors and taking the success probabilities of the different statistical testing procedure into account would be interesting in the future. 

While we prove security against i.i.d. collective attacks, which are assumed to be optimal up to de Finetti correction terms that are massive in the small block length limit, a rigorous security proof against general attacks remains an open question. One issue is that known energy tests on almost i.i.d. states do not bound the weight outside a cutoff space in a way that is useful to apply our numerical method. Furthermore, we require a chain rule for smooth min-entropies to remove an infinite dimensional register, which is not straightforward. This is even a technical issue that applies to the work of Renner and Cirac \cite{Renner_Cirac_2009}. However, assuming a photon-number cutoff, our method is able to handle coherent attacks as well, applying methods developed in Ref. \cite{George_2020}. Therefore, a rigorous general attack security analysis for general DM CV-QKD protocols needs to solve multiple open problems; hence, a generalisation to coherent attacks is left for future work.

 \begin{acknowledgements}
 The authors thank Anthony Leverrier for discussions about the efficiency of error correction and Ignatius William Primaatmaja for spotting an error in an earlier version of our manuscript. This work was performed at the Institute for Quantum Computing, at the University of Waterloo, which is supported by Innovation, Science, and Economic Development Canada. This research has been supported by the NSERC under the Discovery Grants Program (Grant No. 341495).
 \end{acknowledgements}
 

\appendix
\onecolumngrid
\section{Proof of the energy testing theorem\label{apdx:proof_ET}}
In this appendix, we prove our energy testing theorem (Theorem \ref{thm:energy_test})  both for ideal detectors and trusted nonideal detectors. We begin with the proof for ideal detectors.
\begin{proof}
We start by proving an operator inequality related to heterodyne measurements, similarly to Lemma III.2 in Ref. \cite{Renner_Cirac_2009} for homodyne detection. We define the operators
\begin{align}
    W_1 &:= \Pi^{\frac{\hat{q}^2+\hat{p}^2-1}{2} \geq n_c},\\
    V_1 & := \frac{1}{\pi} \int_{|\alpha|^2 \geq \beta_{\mathrm{test}}^2} |\alpha\rangle\langle\alpha| ~d\mu_\alpha,
\end{align}
where $W_1$ is the projector onto the span of the eigenvectors of the operator $\frac{\hat{q}^2+\hat{p}^2-\mathbbm{1}}{2}$ corresponding to (generalized) eigenvalues greater or equal to $n_c$, and $V_1$ describes our test measurement, where the heterodyne detection gives outcomes with amplitudes greater than or equal to $\beta_{\mathrm{test}}$. Defining $W_0 := \mathbbm{1} - W_1$ and $V_0 := \mathbbm{1}-V_1$, it can be easily seen that $\{V_0, V_1\}$ and $\{W_0, W_1\}$ form POVMs.    

Recall that the photon-number operator is defined as $\hat{n} = \frac{1}{2}( \hat{q}^2 + \hat{p}^2 -\mathbbm{1})$. One observes $W_1 := \sum_{n \geq n_c} |n\rangle\langle n|$ and, using $\langle \gamma e^{i\theta} | n\rangle = \frac{e^{-\frac{\gamma^2}{2}} \gamma^n e^{-i \theta n}}{\sqrt{n!}}$ (see, for example \cite[p. 37]{Barnett_Radmore}), it can be seen that $V_1 = \sum_{n \in \mathbb{N}} \frac{\Gamma(n+1,\beta_{\mathrm{test}})}{\Gamma(n+1,0)} |n\rangle\langle n|$. Therefore, comparing the coefficients of $V_1$ and $W_1$ and recalling that, for fixed first argument, the incomplete gamma function is monotonically decreasing in its second argument, we conclude that $\langle n|W_1 | n\rangle \leq 1 \leq \frac{\Gamma(n_c+1,0)}{\Gamma(n_c+1,\beta_{\mathrm{test}})} \langle n |V_1 |n \rangle ~ \forall n \in \mathbb{N}$. Hence, we find that
\begin{equation}\label{eq:operatorRelation_VW}
    W_1 \leq \frac{\Gamma(n_c+1,0)}{\Gamma(n_c+1,\beta_{\mathrm{test}})} V_1.
\end{equation}
To ease notation, we define $r^{\mathrm{ideal}}(n_c,\beta_{\mathrm{test}}) := \frac{\Gamma(n_c+1,0)}{\Gamma(n_c+1,\beta_{\mathrm{test}})}$. The operator $W_0$ is the projector onto the cutoff space $\mathcal{H}^{n_c}$ and $W_1$ projects onto the orthogonal complement of the cutoff space. Therefore, $w = \mathbb{E}\left[W_1\right] = \mathrm{Tr}\left[\rho W_1 \right] \leq r^{\mathrm{ideal}}(n_c, \beta_{\mathrm{test}}) \mathrm{Tr}\left[\rho V_1 \right] = r^{\mathrm{ideal}}(n_c,\beta_{\mathrm{test}}) \mathbb{E}\left[ V_1\right]$. Hence, $\frac{w}{r^{\mathrm{ideal}}(n_c,\beta_{\mathrm{test}})} \leq \mathbb{E}\left[V_1\right]$. To ease notation, we use the short notation $r:= r^{\mathrm{ideal}}$. As it will turn out in the end, we actually do not need to distinguish between two different $r$ for ideal and nonideal detectors.\\

For our analysis, we consider an arbitrary density matrix $\rho$, whose weight outside a cutoff space of dimension $n_c$ can be either larger or smaller than some chosen real number $w \in [0,1]$,
\begin{enumerate}
    \item[1)] $\rho$ is such that $\Tr{\rho W_1} < w$;
    \item[2)] $\rho$ is such that $\Tr{\rho W_1} \geq w$.
\end{enumerate}
In the first case, the energy test accepts on a state which lies indeed with the acceptance set of the energy test. In that case, we can proceed with our security analysis. In the second case, the energy test accepts on a state that does not lie within the acceptance set of the energy test. We now need to make sure that this happens only with small probability $\epsilon_{\mathrm{ET}}$. 

Note that for fixed $\rho$ Born's rule induces a probability distribution in the probability space over outcomes; hence, the i.i.d. testing of it induces a probability distribution over the sequences. The fundamental error, denoted $\epsilon_{\mathrm{ET,~fund}}$, in the i.i.d. setting for our test strategy is the maximum probability of obtaining a sequence that passes the test even though the expected weight for the prototype $\rho$ is greater than or equal to $w$. We denote this probability for a fixed prototype as $\mathrm{Pr}\left[ \left|\left\{ Y_i:~ Y_i \leq \beta_T\right\}\right| \leq l_T ~|~ \rho \right]$. The maximum probability is then obtained by maximising this probability over all such prototypical $\rho$. Therefore, we derive the upper bound

\begin{align*}
    \epsilon_{\mathrm{ET,~fund}} &:= \max_{\rho \in \mathcal{D}(\mathcal{H})} \mathrm{Pr}\left[ \left|\left\{ Y_i:~ Y_i \leq \beta_T\right\}\right| \leq l_T ~|~ \Tr{\rho W_1} \geq w \right]\\
    &= \max_{\rho \in \mathcal{D}(\mathcal{H}):~ \Tr{\rho W_1}\geq w} \mathrm{Pr}\left[ \left|\left\{ Y_i:~ Y_i \leq \beta_T\right\}\right| \leq l_T ~|~ \rho \right]\\
    &\leq \max_{\rho \in \mathcal{D}(\mathcal{H}):~ \Tr{\rho V_1}\geq \frac{w}{r}} \mathrm{Pr}\left[ \left|\left\{ Y_i:~ Y_i \leq \beta_T\right\}\right| \leq l_T ~|~ \rho \right].
\end{align*}
While the first line defines $\epsilon_{\mathrm{ET,~fund}}$, for the second line we  recall that according to our testing strategy, we only have to deal with $\rho$ with expected weight larger than or equal to $w$, which allows us to rewrite the first line by including this condition into the set we maximise over. Density matrices $\rho$ with expected weight smaller than $w$ are not relevant in this part of our analysis. 

Finally, for the inequality in the last step, recall from the first part of the proof, that $W_1 \leq r V_1$; hence, 
\begin{equation*}
    \{\rho \in \mathcal{D}(\mathcal{H}):~ \Tr{\rho W_1}\geq w\} \subseteq \{\rho \in \mathcal{D}(\mathcal{H}):~ r \Tr{\rho V_1} \geq w\}.
\end{equation*}
Now let $\vec{f}_{k_T} \in \{0,1\}^{k_T}$ be a vector containing   `$0$' if $V_0$ was realised and `$1$' if $V_1$ was realised, i.e., for each of the test rounds we write `$0$' if the measurement result of the heterodyne measurement was within a circle of radius $\beta_T$ in the phase-space and `$1$' otherwise and define $f_{k_T}$ be the type induced by $\vec{f}_{k_T}$. Furthermore, define $\Tilde{Q}_{\frac{w}{r}} := \left\{ \begin{pmatrix}
    1-y \\ y
\end{pmatrix}:~ y \in \left[ \frac{w}{r}, 1 \right] \right\}$ and $P_{j} := \begin{pmatrix}
    1-\frac{j}{k_T}\\
    \frac{j}{k_T}
\end{pmatrix}$.
Then, the set we are maximising over reads $\left\{\rho \in \mathcal{D}(\mathcal{H}):~ \begin{pmatrix}
    \Tr{\rho V_0}\\
    \Tr{\rho V_1}
\end{pmatrix} \in \Tilde{Q}_{\frac{w}{r}}\right\}$. Recalling that $V_0 = \mathbbm{1}-V_1$, we introduce $Q_{\rho} := \begin{pmatrix}
    1-\Tr{\rho V_1}\\
    \Tr{\rho V_1}
\end{pmatrix}$.

We observe that
\begin{align*}
    &\mathrm{Pr}\left[ \left|\left\{ Y_i:~ Y_i^2 \leq \beta_T^2\right\}\right| = j ~|~ \rho \right]= \mathrm{Pr}\left[ f_{k_T} = P_j ~|~ \rho\right],
\end{align*}
which is given by the product of the size of the corresponding type class and its probability
\begin{align}\label{eq:types_1}
    \mathrm{Pr}\left[ f_{k_T} = P_j ~|~ \rho\right] =  |T(P_j)| ~ Q_{\rho}^{k_T},
\end{align}
where $|T(P_j)|$ denotes the size of type class $P_j$ and by $Q_{\rho}^{k_T}$ we denote the product distribution $Q_{\rho}^{k_T}:=\Pi_{j=0}^{k_T} Q_{\rho}$.

Next, we use two theorems from Ref. \cite{Cover_Thomas}. The first one, Theorem 11.1.2 of\cite{Cover_Thomas}, tells us that, for $n$ i.i.d. random variables $X_1, ..., X_n$ drawn according to $Q(x)$, the probability of a certain $n$-sequence $\vec{x}$ only depends on its type $P_{\vec{x}}$, $Q^{n}(\vec{x}) = 2^{-n \left(H(P_{\vec{x}}) + D(P_{\vec{x}}||Q) \right)}$. The second one, Theorem 11.1.3 of \cite{Cover_Thomas}, gives  an upper bound for the size of a type class of type $P \in \mathcal{P}_n$ (so a type with denominator $n$), $|T(P)|\leq 2^{n H(P)}$.  Applying both to Eq. (\ref{eq:types_1}) yields
\begin{align}\label{eq:type-prob-bound}
    \mathrm{Pr}\left[ f_{k_T} = P_j ~|~ \rho\right] \leq  2^{-k_T D(P_j || Q_{\rho})},
\end{align}
where $D$ is the Kullback-Leibler divergence.
Collecting what we found so far, we arrive at
\begin{align*}
    \epsilon_{\mathrm{ET,~fund}} \leq \max_{\rho \in \mathcal{D}(\mathcal{H}):~ \Tr{\rho V_1} \geq \frac{w}{r}} \sum_{j=0}^{l_T} 2^{-k_T D(P_j || Q_{\rho})}.
\end{align*}
We assume that $\frac{l_T}{k_T} < \frac{w}{r}$; hence, $Q_\frac{w}{r}:= \begin{pmatrix}
    1-\frac{w}{r} \\ \frac{w}{r}
\end{pmatrix}$ will always be the closest to each of the $P_j$ among all $y \in \left[\frac{w}{r}, 1\right]$. Furthermore, choosing $j=l_T$ minimises the relative entropy between $P_j$ and $Q_{\frac{w}{r}}$,
\begin{align*}
\forall j \leq l_T~ \forall y \in \left[\frac{w}{r}, 1\right]: ~D(P_j || Q_{\rho}) \geq D(P_j || Q_{\frac{w}{r}}) \geq D(P_{l_T} || Q_{\frac{w}{r}}).
\end{align*}
Therefore, we conclude that
\begin{align*}
    \epsilon_{\mathrm{ET,~fund}} &\leq \max_{\rho \in \mathcal{D}(\mathcal{H}):~ \Tr{\rho V_1} \geq \frac{w}{r}} \sum_{j=0}^{l_T} 2^{-k_T D(P_j || Q_{\rho})}\\
    & \leq \sum_{j=0}^{l_T} 2^{-k_T D\left(P_{l_T} || Q_{\frac{w}{r}}\right)}\\
    & = (l_T+1) \cdot 2^{-k_T D\left(P_{l_T} || Q_{\frac{w}{r}}\right)} =: \epsilon_{\mathrm{ET}}.
\end{align*}
This completes the proof.
\end{proof}

It remains to prove the energy testing theorem for trusted, nonideal detectors. The second part of the proof follows the arguments of the proof for ideal detectors. However, the measurement operator for trusted, nonideal detectors differs from the measurement operator $V_1$ for the ideal detector. Therefore it remains to show that the measurement operator for the trusted, nonideal case dominates $W_1$ as well (possibly with another constant $r(n_c, \beta_{\mathrm{test}})$.
\begin{proof}
According to \cite{Lin_2020} the POVM elements for the trusted, nonideal heterodyne measurement with efficiency $\eta_d$ and electronic noise $\nu_{\mathrm{el}}$ are given by
\begin{equation}
    G_y = \frac{1}{\eta_d \pi} \hat{D}\left( \frac{y}{\sqrt{\eta_d}} \right) \hat{\rho}_{\mathrm{th}}\left( \overline{n}_d \right) \hat{D}^{\dagger}\left( \frac{y}{\sqrt{\eta_d}} \right),
\end{equation}
where $\overline{n}_d := \frac{1-\eta_d + \nu_{\mathrm{el}}}{\eta_d}$. Therefore, the modified measurement operator is $\tilde{V}_1 := \int_{y^2 \geq \beta_{\mathrm{test}}^2} G_y ~d\mu_{y}$. We use \cite[Eq.~(6.13) and (6.14)]{Glauber_1967} to express $G_y$ in the number basis. For simplification, we define $C_{n,m} := \frac{1}{\pi \eta_d{\frac{m-n}{2}+1}} \sqrt{\frac{n!}{m!}} \frac{\overline{n}_d^n}{(1+\overline{n}_d)^{m+1}}$, $a := \frac{1}{\eta_d(1+\overline{n}_d)}$ and $b:= \eta_d \overline{n}_d(1+\overline{n}_d)$,  and obtain for $n\leq m$
\begin{equation}\label{eq:Gy_numberbasis}
    \langle n |G_y | m \rangle = C_{n,m} e^{-a |y|^2} (y^*)^{m-n} L_n^{(m-n)}\left( - \frac{|y|^2}{b} \right),
\end{equation}
where 
\begin{equation}\label{eq:laguerre}
  L_k^{\alpha}(x) = \sum_{j=0}^{k} (-1)^j \binom{k+\alpha}{k-j} \frac{x^j}{j!}   
\end{equation}
is the generalised Laguerre polynomial of degree $k$ and with parameter $\alpha$ \cite{Oldham_2008}. The following calculation is a special case of the derivation in \cite[Appendix F]{Kanitschar_2021a} and \cite[Appendix 5.1]{Kanitschar_Thesis_2021}.
\begin{align*}
    \tilde{V}_1  = \int_{y^2 \geq \beta_{\mathrm{test}}^2} G_y ~d\mu_{y} &= \sum_{m,n} C_{n,m} |n\rangle\langle n| \int_{y^2 \geq \beta_{\mathrm{test}}^2} y^{m-n+1} e^{-a y^2} L_n^{(m-n)}\left( -\frac{y^2}{b} \right) ~dy~\int_{\theta = 0}^{2 \pi} e^{-i \theta}~d\theta \\
    &=  \sum_{m,n} C_{n,m} |n\rangle\langle m| \int_{y^2 \geq \beta_{\mathrm{test}}^2} y^{m-n+1} e^{-a y^2} L_n^{(m-n)}\left( -\frac{y^2}{b} \right) ~dy~ 2\pi \delta_{n,m}\\
    &= 2\pi \sum_{n} C_{n,n} |n\rangle\langle n| \int_{y^2 \geq \beta_{\mathrm{test}}^2} y e^{-a y^2} L_n\left( -\frac{y^2}{b} \right) ~dy\\
    &= \pi \sum_{n} C_{n,n} |n\rangle\langle n| \int_{z \geq \beta_{\mathrm{test}}} e^{-a z} L_n\left( -\frac{z}{b} \right) ~dz\\
    &= \pi \sum_{n} C_{n,n} |n\rangle\langle n| \sum_{j=0}^{n} \binom{n}{n-j}\frac{1}{a^{j+1} b^j} \frac{\Gamma(j+1, a \beta_{\mathrm{test}})}{ \Gamma(j+1)} \\
\end{align*}
Note that we substituted $y^2 \mapsto z$ for the fifth equality and that we used the definition of the Laguerre polynomials to obtain the last line. Inserting $C_{n,n}$ and simplifying the obtained expression yields
\begin{equation*}
    \tilde{V}_1 = \sum_{n} \left( \frac{\overline{n}_d}{1+ \overline{n}_d}\right)^n  \sum_{j=0}^{n} \binom{n}{j} \left(\frac{1}{\overline{n}_d} \right)^j \frac{\Gamma(j+1, a \beta_{\mathrm{test}})}{ \Gamma(j+1)} |n\rangle\langle n|.
\end{equation*}
We define and simplify
\begin{align*} 
    U &:= \sum_{n=0}^{n_c-1} \left( \frac{\overline{n}_d}{1+ \overline{n}_d}\right)^n  \sum_{j=0}^{n} \binom{n}{j} \left(\frac{1}{\overline{n}_d} \right)^j \frac{\Gamma(j+1, a\beta_{\mathrm{test}})}{ \Gamma(j+1)} |n\rangle\langle n| + \frac{\Gamma(n_c+1, \beta_{\mathrm{test}})}{ \Gamma(n_c+1)}  \sum_{n=n_c}^{\infty}\left( \frac{\overline{n}_d}{1+ \overline{n}_d}\right)^n \sum_{j=0}^{n} \binom{n}{j} \left(\frac{1}{\overline{n}_d} \right)^j  |n\rangle\langle n| \\
&= \sum_{n=0}^{n_c-1} \left( \frac{\overline{n}_d}{1+ \overline{n}_d}\right)^n  \sum_{j=0}^{n} \binom{n}{j} \left(\frac{1}{\overline{n}_d} \right)^j \frac{\Gamma(j+1, a\beta_{\mathrm{test}})}{ \Gamma(j+1)} |n\rangle\langle n| + \frac{\Gamma(n_c+1, \beta_{\mathrm{test}})}{ \Gamma(n_c+1)} \sum_{n=n_c}^{\infty} \left( \frac{\overline{n}_d}{1+ \overline{n}_d}\right)^n  \left(\frac{1}{\overline{n}_d} +1 \right)^n |n\rangle\langle n|\\
&= \sum_{n=0}^{n_c-1} \left( \frac{\overline{n}_d}{1+ \overline{n}_d}\right)^n  \sum_{j=0}^{n} \binom{n}{j} \left(\frac{1}{\overline{n}_d} \right)^j \frac{\Gamma(j+1, a\beta_{\mathrm{test}})}{ \Gamma(j+1)} |n\rangle\langle n| + \frac{\Gamma(n_c+1, \beta_{\mathrm{test}})}{ \Gamma(n_c+1)} \sum_{n=n_c}^{\infty} |n\rangle\langle n|
\end{align*}    

Note that the quotient $\frac{\Gamma(j+1, a\beta_{\mathrm{test}})}{ \Gamma(j+1)}$ is monotonically increasing in $j$, therefore $\forall j \geq n_c:~ \frac{\Gamma(n_c+1, a\beta_{\mathrm{test}})}{ \Gamma(n_c+1)} \leq \frac{\Gamma(j+1,a \beta_{\mathrm{test}})}{ \Gamma(j+1)}$. Hence, $U \leq \tilde{V}_1$.
Based on the structure of $W_1$, we observe $W_1 \leq \frac{\Gamma(n_c+1)}{ \Gamma(n_c+1, \beta_{\mathrm{test}})} U$. Defining $r^{\mathrm{nonideal}}(\beta_{\mathrm{test}}) := \frac{\Gamma(n_c+1)}{\Gamma(n_c+1, a\beta_{\mathrm{test}})}$ and combining our operator relations, we obtain
\begin{equation}
    W \leq r^{\mathrm{nonideal}}(\beta_{\mathrm{test}}) U \leq r^{\mathrm{nonideal}}(\beta_{\mathrm{test}}) \tilde{V}_1.
\end{equation}
The rest of the proof is identical to the ideal case.
\end{proof}

\section{Bounded measurements}\label{APDX:Bounded_Measurement}
Here, we discuss how bounded measurements affect observations as well as the energy test and the key map. First, we discuss modifications of the observables.
\subsection{Observables}
In this section, we derive $[\hat{n}]_r'$ and $[\hat{n}^2]_r'$, the noisy and restricted observables used in the present protocol. We start by writing $\hat{n}$ and $\hat{n}^2$ in anti-normal ordering and replacing the ladder operators $\hat{a} \leftrightarrow \zeta$ and $\hat{a}^{\dagger} \leftrightarrow \zeta^*$ and obtain
\begin{align}
    f_{\hat{n}} &= |\zeta|^2-1,\\
    f_{\hat{n}^2} &= |\zeta|^4 - 3 |\zeta|^2 + 1.
\end{align}

We will now adopt the approach described in \cite[Appendix  D]{Upadhyaya_2021}, but replace $f_{\hat{n}}$ by $g_{\hat{n}}$ and $f_{\hat{n}^2}$ by $g_{\hat{n}^2}$. While $f_{\hat{n}}(\zeta_x, \zeta_y)$ coincides with $g_{\hat{n}}(\zeta_x, \zeta_y)$ for $(\zeta_x, \zeta_y) \in [-M,M]^2$, we have $g_{\hat{n}} = M^2+\zeta_y^2-1$ for $|\zeta_x| \geq M$ and $\zeta_y \in [-M,M]$. Furthermore, for $\zeta_x, \zeta_y \geq M$, we obtain $g_{\hat{n}}(\zeta_x, \zeta_y) = 2M^2-1$. Similar results can be derived for the other regions, and the same principle applies to $g_{\hat{n}^2}$. Following the method in \cite[Appendix D]{Upadhyaya_2021} for this modified setup yields (applying an asymptotic expansion [as $M$ is large compared to the other appearing quantities] and keeping leading correction terms) the following expression for the Q-function (up to prefactor $\frac{1}{\pi}$) of the restricted, noisy, trusted operators
\begin{equation}
    \begin{aligned}
    \bra{\alpha}[\hat{n}_{\beta}]_r'\ket{\alpha} &= \eta_d |\gamma|^2 \left(1-\frac{1}{\sqrt{\pi} \tilde{M}} e^{-\tilde{M}^2} \right)+\eta_d c^2 \left(1 - \frac{2\tilde{M}+2\tilde{M}^2\sqrt{\pi}+2\sqrt{\pi}}{\pi \tilde{M}} e^{-\tilde{M}^2}\right) \\&+ 2  \eta_d c \frac{\tilde{M}^2\sqrt{\pi}\eta_d+1}{\sqrt{\pi} \tilde{M}}e^{-\tilde{M}^2} - 1 \left( 1 - \frac{2}{\sqrt{\pi}\tilde{M}} e^{-\tilde{M}^2}\right) + \mathcal{O}\left( e^{-2 \tilde{M}^2} \right),
    \end{aligned}
\end{equation}
\begin{equation}
    \begin{aligned}
    \bra{\alpha}[\hat{n}^2_{\beta}]_r'\ket{\alpha} &= \eta_d^2\left[ 1-\frac{1}{\sqrt{\pi} \tilde{M}} e^{-\tilde{M}^2}\right] |\gamma|^4 \\
    &+\eta_d^2\left[4c^2 \left( 1-\frac{8\tilde{M}^2+8\tilde{m}+1}{4\sqrt{\pi} \tilde{M}}e^{-\tilde{M}^2}\right)-\frac{3 c^2}{\eta_d} \left( 1-\frac{9c^2+4\tilde{M}^2}{6\sqrt{\pi}\tilde{M}c^2}\right) e^{-\tilde{M}^2}\right]  |\gamma|^2\\
    &+2\eta_d^2 c^2\left(1-\frac{13\sqrt{\pi}-3\tilde{M}-4\tilde{M}^2\sqrt{\pi}c^2}{8 \pi \tilde{M}} \right)e^{-\tilde{M}^2} - 3 \eta_d c^2 \left(1-\frac{3-c^2}{2 \sqrt{\pi} \tilde{M}} \right)e^{-\tilde{M}^2} \\&+1 \left(1-\frac{\eta_d^2 c^4 \tilde{M}^4 + 3 \eta_d c^2 \tilde{M}^2-3+2\eta_d^2}{\sqrt{\pi}\tilde{M} \eta_d^2} \right)e^{-\tilde{M}^2} + \mathcal{O}\left(e^{-2\tilde{M}^2}\right),
    \end{aligned}
\end{equation}
where $\tilde{M}:=\frac{M}{\sqrt{\eta_d}c}$, $c^2:=1+\bar{n} = 1+\frac{1-\eta_d+\nu_{el}}{\eta_d}=\frac{1+\nu_{el}}{\eta_d}$ and $\gamma :=\alpha-\frac{\beta}{\sqrt{\eta_d}}$. 
We observe that when $M$ is chosen to be sufficiently large, the neglected terms become extremely small, often (depending on the particular choice of $M$) even below the level of machine precision. It is important to highlight that the numerical method we employ to obtain accurate lower bounds on the secure key rate \cite{Coles_2016, Winick_2018} accounts for small violations of constraints and finite-precision errors in the representation of operators, which may have a magnitude of $\epsilon'$. Therefore, as long as we ensure that the neglected terms remain below this threshold, the resulting lower bounds remain reliable. For more details about handling numerical imprecisions in the used security proof framework, we refer to \cite[Section 3.3]{Winick_2018}. This shows that the effect of restricting our measurement to only a finite detection range has negligible impact on our implementation. Furthermore, notice that for $\tilde{M} \rightarrow \infty:$ 
\begin{align*}
    &\bra{\alpha}[\hat{n}_{\beta}]_r'\ket{\alpha} \stackrel{\tilde{M} \rightarrow \infty}{\rightarrow}\bra{\alpha}[\hat{n}_{\beta}]'\ket{\alpha}\\
    &\bra{\alpha}[\hat{n}^2_{\beta}]_r'\ket{\alpha} \stackrel{\tilde{M} \rightarrow \infty}{\rightarrow} \bra{\alpha}[\hat{n}^2_{\beta}]'\ket{\alpha},
\end{align*}
i.e., as expected, we recover the results for the unbounded (noisy, nonideal) measurement from \cite{Upadhyaya_2021}.

By the uniqueness of the Q-function, we obtain then
\begin{equation}
    \begin{aligned}
[\hat{n}_{\beta}]_r' &\simeq A(\tilde{M}) \hat{n}_{\frac{\beta}{\sqrt{\eta_d}}} + B(\tilde{M}) \mathbbm{1},
    \end{aligned}
\end{equation}

\begin{equation}
    \begin{aligned}
[\hat{n}^2_{\beta}]_r' &\simeq C(\tilde{M}) \hat{n}^2_{\frac{\beta}{\sqrt{\eta_d}}} + D(\tilde{M}) \hat{n}_{\frac{\beta}{\sqrt{\eta_d}}}+E(\tilde{M}) \mathbbm{1},
    \end{aligned}
\end{equation}
where
\begin{align*}
    A(\tilde{M}) :=& \eta_d\left( 1-\frac{1}{\sqrt{\pi} \tilde{M}} e^{-\tilde{M}^2}\right)\\
    B(\tilde{M}) :=& \eta_d c^2 \left(1 - \frac{2\tilde{M}+2\tilde{M}^2\sqrt{\pi}+2\sqrt{\pi}}{\pi \tilde{M}} e^{-\tilde{M}^2}\right) \\&+ 2  \eta_d c \frac{\tilde{M}^2\sqrt{\pi}\eta_d+1}{\sqrt{\pi} \tilde{M}}e^{-\tilde{M}^2} - 1 \left( 1 - \frac{2}{\sqrt{\pi}\tilde{M}} e^{-\tilde{M}^2}\right)\\
    C(\tilde{M}) :=&\eta_d^2\left( 1-\frac{1}{\sqrt{\pi} \tilde{M}} e^{-\tilde{M}^2}\right)\\
    D(\tilde{M}):=&\eta_d^2\left[4c^2 \left( 1-\frac{8\tilde{M}^2+8\tilde{m}+1}{4\sqrt{\pi} \tilde{M}}e^{-\tilde{M}^2}\right)-\frac{3 c^2}{\eta_d} \left( 1-\frac{9c^2+4\tilde{M}^2}{6\sqrt{\pi}\tilde{M}c^2}\right) e^{-\tilde{M}^2}\right. \\ &\left. ~-\left(1-\frac{1}{\sqrt{\pi}\tilde{M}} e^{-\tilde{M}^2}\right)\right] \\
    E(\tilde{M}):=& +2\eta_d^2 c^2\left(1-\frac{13\sqrt{\pi}-3\tilde{M}-4\tilde{M}^2\sqrt{\pi}c^2}{8 \pi \tilde{M}} \right)e^{-\tilde{M}^2} - 3 \eta_d c^2 \left(1-\frac{3-c^2}{2 \sqrt{\pi} \tilde{M}} \right)e^{-\tilde{M}^2} \\&+1 \left(1-\frac{\eta_d^2 c^4 \tilde{M}^4 + 3 \eta_d c^2 \tilde{M}^2-3+2\eta_d^2}{\sqrt{\pi}\tilde{M} \eta_d^2} \right)e^{-\tilde{M}^2} + \mathcal{O}\left(e^{-2\tilde{M}^2}\right)
\end{align*}
and, again, the restricted operators converge to the unrestricted operators given in \cite{Upadhyaya_2021} for $\tilde{M} \rightarrow \infty$.

\subsection{Energy test}
After having clarified our observables, we can proceed with the energy test. Therefore, let us review the purpose of the energy test. When performing the energy test, we take some fraction of all rounds and check if $q^2+p^2$ is smaller or larger than some arbitrary but fixed value $\beta_{\mathrm{test}}^2$. As long as we choose $\beta_{\mathrm{test}} \leq M$, this binary measurement is not affected by the finite detection range (note that we do not need to know the exact value but only need to know if it is smaller than our testing parameter) as can be seen from the definition of the measurement operator $V_1$,

\begin{equation}
    V_1 := \frac{1}{\pi} \int_{|\alpha|^2 \geq \beta_{\mathrm{test}}^2} |\alpha\rangle\langle\alpha| ~d\mu_\alpha,
\end{equation}
here for the ideal heterodyne measurement POVM, but the same applies if we replace $\frac{1}{\pi} \ket{\alpha}\!\!\bra{\alpha}$ by $G_{\alpha}$. Comparing to Eq. (\ref{eq:OperatorXideal}) in the ideal case or to Eq. (\ref{eq:OperatorXnonIdeal}) for the nonideal detector, we see that $f_{V_1} = 1$. As a result, the integral remains the same even for the bounded operator.
Summing up, the energy test remains completely unaffected by this modification, providing we select a value for $M$ that is not smaller than $\beta_{\mathrm{test}}$.

\subsection{Modified key map}\label{sec:APDX:ModifiedKeymap}
As we only want to use unambiguous measurement results, we restrict our key regions to the area between the postselection circle in the middle of the phase space and the detection-range bound at $M$. For $z=0$ and $z=2$, we obtain
\begin{equation}
    R_B^z := \frac{1}{\pi}  \int_{\frac{2z-1}{N_{\mathrm{St}}} \pi}^{\frac{2z+1}{N_{\mathrm{St}}} \pi} \int_{\Delta_r}^{\frac{r}{cos(\theta)}} r |r e^{i \phi}\rangle \langle r e^{i \phi}|~d\phi ~dr,
\end{equation}
and for $z=1$ and $z=3$, we obtain
\begin{equation}
    R_B^z := \frac{1}{\pi}  \int_{\frac{2z-1}{N_{\mathrm{St}}} \pi}^{\frac{2z+1}{N_{\mathrm{St}}} \pi} \int_{\Delta_r}^{\frac{r}{\sin(\theta)}} r |r e^{i \phi}\rangle \langle r e^{i \phi}|~d\phi ~dr.
\end{equation}
We note that this integral cannot be computed analytically anymore which increases the computation time extensively. One possible solution is to slightly modify the key map by discarding not only results lying outside $\mathcal{M}=[-M,M]^2$ but outside a circle with radius $M$. Then the region operators read
\begin{equation}
    R_B^z := \frac{1}{\pi} \int_{\Delta_r}^{M} \int_{\frac{2z-1}{N_{\mathrm{St}}} \pi}^{\frac{2z+1}{N_{\mathrm{St}}} \pi}  r |r e^{i \phi}\rangle \langle r e^{i \phi}|~d\phi ~dr,
\end{equation}
which can be calculated analytically. Although we increase the region corresponding to $\perp$, as the removed areas are close to the corners of $[-M,M]^2$ we do not expect a significant impact on the key rate, while speeding up the calculation considerably. Thus, we modify the key map accordingly for our simulations.

\section{Technical lemmas}\label{apdx:TechnicalLemmas}
In this section, we present technical lemmas we use in the security proof to generalise existing finite-dimensional statements to their infinite dimensional counterparts.

\begin{proposition}[\textbf{Relation between $\epsilon$-balls}] \label{prop:epsilonBalls}
For $\rho \in \mathcal{D}_{\leq}(\mathcal{H})$ we have $\mathcal{B}_{\mathrm{PD}}^{\epsilon}(\rho) \subseteq \mathcal{B}_{\mathrm{TD}}^{2\epsilon}(\rho) \subseteq \mathcal{B}_{\mathrm{PD}}^{\sqrt{2\epsilon}}(\rho)$. 
\end{proposition}
\begin{proof}
   Consider $\rho \in \mathcal{D}_{\leq}(\mathcal{H})$ and $\sigma \in \mathcal{B}_{\mathrm{PD}}^{\epsilon}(\rho)$. 
   
   For the first inclusion, by one of the Fuchs-van de Graaf inequalities (Eq.~(\ref{eq:FvdG})), we have $\Delta(\rho,\sigma) \leq \mathcal{P}(\rho,\sigma) \leq \epsilon$, hence if $\sigma \in \mathcal{B}_{\mathrm{PD}}^{\epsilon}(\rho)$, we have $2 \Delta(\rho, \sigma) \leq 2\epsilon$. Thus, due to the definition of the trace distance ball (without a factor $\frac{1}{2}$), every $\sigma \in \mathcal{B}_{\mathrm{PD}}^{\epsilon}(\rho)$ is contained in $\mathcal{B}_{\mathrm{TD}}^{2\epsilon}(\rho)$. 
   
   For the second inclusion, assume $\sigma \in \mathcal{B}_{\mathrm{TD}}^{2\epsilon}(\rho)$. Then, by the other Fuchs-van de Graaf inequality in Eq. (\ref{eq:FvdG}), we have $\mathcal{P}(\rho, \sigma) \leq \sqrt{2 \Delta(\rho, \sigma)} \leq \sqrt{2 \epsilon}$. Hence, if $\sigma \in \mathcal{B}_{\mathrm{TD}}^{2\epsilon}(\rho)$ it is as well contained in $\mathcal{B}_{\mathrm{PD}}^{\sqrt{2\epsilon}}$.
\end{proof}

\begin{lemma}[\textbf{Data-processing inequality for the trace distance under CPTNI maps}]\label{lemma:DPI}
Let $\mathcal{H}$ be a separable Hilbert space and let $\rho, \sigma$ be compact, self-adjoint trace-class-1 operators over the separable Hilbert space $\mathcal{H}$ and let $\mathcal{E}$ be a completely positive trace non-increasing (CPTNI) map.\\
Then,
\begin{align*}
    \left|\left| \mathcal{E}(\rho) - \mathcal{E}(\sigma)  \right|\right|_1 \leq  \left|\left| \rho - \sigma \right|\right|_1.
\end{align*}
\end{lemma}
\begin{proof}
Consider $\rho, \sigma$ compact, self-adjoint and trace-class operators, as in the statement. Then, trivially, $\rho-\sigma$ is self-adjoint as well. Furthermore, compact operators form a vector space, so $\rho-\sigma$ is compact, too.

Now we may apply the spectral theorem for compact, self-adjoint operators on $\rho-\sigma$ and find an orthonormal basis diagonalising $\rho-\sigma$. Let $P$ be the positive part and $Q$ the negative part of the diagonal form of $\rho-\sigma$, $\rho-\sigma = U(P+Q)U^\dagger$, where $P \perp Q$. Note that we found $P, Q$ diagonal, $P \perp Q$ with $||\rho - \sigma||_1 = ||P+Q||_1$.

Since $\mathcal{E}$ is a CPTNI map, we can find a Kraus representation $\mathcal{E}(\tau) = \sum_i K_i \tau K_i^{\dagger}$ where $\sum_i K_i^{\dagger} K_i = \mathbbm{1}$. Inserting $\tau = UDU^{\dagger}$, where $D$ is the diagonal form and $U$ the corresponding transformation, we obtain
\begin{align*}
    \mathcal{E}(\tau) = \mathcal{E}(U \tau U^{\dagger}) = \sum_{i} K_i U D U^{\dagger} K_i^{\dagger} =  \sum_{i} K_i U D \left(K_i  U\right)^{\dagger} = \sum_{i} \tilde{K}_i D \tilde{K}_i^{\dagger}.
\end{align*}
Note that we defined $\tilde{K}_i := K_i U$ and observe 
\begin{align*}
    \sum_i \tilde{K}_i^{\dagger} \tilde{K}_i \tilde{K}_i^{\dagger} = \sum_{i} \left(K_i U\right)^{\dagger} K_i U = U^{\dagger} \left(\sum_i K_i^{\dagger} K_i\right) U = U^{\dagger}U = \mathbbm{1}.
    \end{align*}
Define the new channel $\tilde{\mathcal{E}}(\tau) = \sum_i \tilde{K}_i \tau \tilde{K}_i^{\dagger}$.
Finally, we conclude
\begin{align*}
    || \mathcal{E}(\rho) - \mathcal{E}(\sigma)||_1 =& || \mathcal{E}(\rho-\sigma)||_1 = || \tilde{\mathcal{E}}(P+Q)||_1 = || \tilde{\mathcal{E}}(P) + \tilde{\mathcal{E}}(Q)||_1 \\
    \leq& ||\tilde{\mathcal{E}}(P)||_1 +  ||\tilde{\mathcal{E}}(Q)||_1 = \Tr{\tilde{\mathcal{E}}(P)} + \Tr{\tilde{\mathcal{E}}(Q)} \\
    \leq&  \Tr{P} + \Tr{Q} = \Tr{P+Q} = ||P+Q||_1 \\
    =& ||\rho - \sigma||_1,
\end{align*}
which proves the claim.

\end{proof}

\begin{lemma}[\textbf{Leftover hashing lemma against infinite dimensional side information}]\label{lemma:leftoverHashing}
Let $\rho_{XE} \in \mathcal{D}_{\leq}(\ell^{\infty}_{X} \otimes \mathcal{H}_E)$, where $X$ is finite. Let $K$ and $X$ be finite sets with $|K|=2^{\ell} \leq |X|$ and let $\{\mathcal{F}, \mathcal{P}_{\mathcal{F}}\}$ be a family of two-universal $\{X, K\}$-hash functions. Let $\epsilon'>0$ and $\epsilon_{\mathrm{PA}} := 2(\epsilon_{\mathrm{sec}} - 2\epsilon')$, where $\epsilon_{\mathrm{sec}} \geq 2 \epsilon' + \frac{1}{2} \sqrt{2^{\ell - H_{\mathrm{min(PD)}}^{\epsilon'}(X|E)_{\rho}}}$ in case of purified distance smoothing and in case of trace distance smoothing $\epsilon_{\mathrm{sec}} \geq 2 \epsilon' + \frac{1}{2} \sqrt{2^{\ell - H_{\mathrm{min(TD)}}^{2\epsilon'}(X|E)_{\rho}}}$ .\\
Then,
\begin{equation*}
  \frac{1}{2} || \rho_{F(x)EF} - \pi_k \otimes \rho_{EF}||_1 \leq 2\epsilon' + \frac{1}{2} \sqrt{2^{\ell - H_{\mathrm{min(PD)}}^{\epsilon'}(X|E)_{\rho}}} \leq \epsilon_{\mathrm{sec}}.   
 \end{equation*}

This implies that for the purified distance smoothing ball, if
\begin{equation*}
    \ell \leq H_{\mathrm{min(PD)}}^{\epsilon'}(X|E)_{\rho} - 2 \log_2\left(\frac{1}{\epsilon_{\mathrm{PA}}}\right),
\end{equation*}
or, for the trace distance smoothing ball, if 
\begin{equation*}
    \ell \leq H_{\mathrm{min(TD)}}^{2\epsilon'}(X|E)_{\rho} - 2 \log_2\left(\frac{1}{\epsilon_{\mathrm{PA}}}\right),
\end{equation*}
the obtained key is $\epsilon_{\mathrm{sec}}$-secure. 
\end{lemma}
\begin{proof}
 We start the proof with \cite[Proposition 21]{Berta_2016} for the case $|K| = 2^{\ell}$ since we are interested in bit-strings. Then, Proposition 21 states that for $X,K$, two sets of finite cardinality with $|K| = 2^{\ell} \leq |X|$, $\{\mathcal{F}, \mathcal{P}_{\mathcal{F}}\}$, a family of two-universal $\{X,K\}$-hash functions, $\rho_{XE} = (\rho_E^x)_{x \in X} \in \mathcal{D}_{\leq}(\ell_X^{\infty} \otimes \mathcal{M}_E)$ and $\epsilon' > 0$
 \begin{equation*}
     \mathbb{E}_{\mathcal{F}}|| (T_f \otimes \mathrm{id}_E)(\rho_{XE}) - \pi_K \otimes \rho_E ||_1 \leq \sqrt{ 2^{\ell - H_{\mathrm{min}}^{\epsilon'}(X|E)_{\rho}}}+ 4 \epsilon',
 \end{equation*}
 holds. Here $\mathbb{E}_{\mathcal{F}}$ denotes the expectation with respect to $\mathcal{P}_{\mathcal{F}}$, $T_{f}$ is the map applying the hash function and $\pi_K = \frac{1}{|K|} \sum_{s \in K} |s\rangle\langle s|$. Note that $K$ denotes the alphabet the hash function map into and that Ref. \cite{Berta_2016} uses the purified distance in the smooth min-entropy definition. 

 First, we rewrite the left-hand side
 \begin{align*}
   \mathbb{E}_{\mathcal{F}}|| (T_f \otimes \mathrm{id}_E)(\rho_{XE}) - \pi_k \otimes \rho_E ||_1 &= \sum_{f} p(f) || (T_f \otimes \mathrm{id}_E)(\rho_{XE}) - \pi_K \otimes \rho_E ||_1  \\
   &= \left|\left| \sum_{f} p(f) \left[ (T_f \otimes \mathrm{id}_E)(\rho_{XE}) - \pi_K \otimes \rho_E\right] \otimes |f\rangle\langle f|  \right|\right|_1 \\
   &= || \rho_{F(X)EF} - \pi_K \otimes \rho_{EF}||_1.
 \end{align*}
 
 We replace the left-hand side of the original statement with what we just derived and divide by two to obtain a statement in trace distance and obtain
 \begin{equation*}
  \frac{1}{2} || \rho_{F(X)EF} - \pi_k \otimes \rho_{EF}||_1 \leq 2\epsilon' + \frac{1}{2} \sqrt{2^{\ell - H_{\mathrm{min(PD)}}^{\epsilon'}(X|E)_{\rho}}} \leq \epsilon_{\mathrm{sec}}.   
 \end{equation*}
Let $\epsilon_{\mathrm{PA}} := 2(\epsilon_{\mathrm{sec}} - 2\epsilon') > 0.$ Then, we derive
\begin{align*}
2^{\ell - H_\mathrm{min(PD)}^{\epsilon'}(X|E)_{\rho}} \leq  \epsilon_{\mathrm{PA}}^2 = 4 (\epsilon_{\mathrm{sec}}-2\epsilon')^2 \Rightarrow \ell \leq H_{\mathrm{min(PD)}}^{\epsilon'}(X|E)_{\rho} - 2\log_2\left( \frac{1}{\epsilon_{\mathrm{PA}}} \right),
\end{align*}
where $F \in \mathcal{F}$.
This gives us the statement in purified distance smoothing. By Proposition \ref{prop:epsilonBalls}, we yield the proposed statement in trace distance smoothing.
\end{proof}

\begin{lemma}[\textbf{Chain rule for smooth min-entropies}]\label{lemma:chainRule}
Let $\mathcal{H}_A, \mathcal{H}_B, \mathcal{H}_C$ be separable Hilbert spaces with $|\mathcal{H}_B| = n$.\\
Then for smoothing in trace distance,
\begin{align*}
    H_{\mathrm{min(TD)}}^{\epsilon}(AB|C)_{\rho} - \log_2(n) \leq H_{\mathrm{min(TD)}}^{\epsilon}(A|BC)_{\rho},
\end{align*}
as well as for smoothing in purified distance
\begin{align*}
    H_{\mathrm{min(PD)}}^{\epsilon}(AB|C)_{\rho} - \log_2(n) \leq H_{\mathrm{min(PD)}}^{\epsilon}(A|BC)_{\rho},
\end{align*}
\end{lemma}
\begin{proof}
The proof in purified distance smoothing can be found in \cite[Lemma 4.5.6]{Furrer_2012} and it is straightforward to show that the proof given there works for trace distance smoothing as well.
\end{proof}

\begin{lemma}[\textbf{Strong subadditivity of smooth min-entropy}]\label{lemma:strongSubadditivity}
Let $\mathcal{H}_A, \mathcal{H}_{B}$ and $\mathcal{H}_C$ be separable Hilbert spaces and $\rho \in \mathcal{D}_{\leq}(\mathcal{H}_A \otimes \mathcal{H}_B \otimes \mathcal{H}_C)$.\\
Then, for either smoothing ball
\begin{equation*}
    H_{\mathrm{min(TD)}}^{\epsilon}(A|BC)_{\rho} \leq H_{\mathrm{min(TD)}}^{\epsilon}(A|B)_{\rho},
\end{equation*}
\begin{equation*}
    H_{\mathrm{min(PD)}}^{\epsilon}(A|BC)_{\rho} \leq H_{\mathrm{min(PD)}}^{\epsilon}(A|B)_{\rho}.
\end{equation*}
\end{lemma}
\begin{proof}
The proof for the trace distance follows from \cite[Lemma 3.2.7]{Renner_2005} which states the strong subadditivity for finite-dimensional Hilbert spaces since this proof only relies on \cite[Lemma 3.1.7]{Renner_2005} (its proof is identical for separable Hilbert spaces) and the fact that the trace distance is monotonic under CPTNI maps (which we have established in Lemma~\ref{lemma:DPI}). Therefore, it remains to prove the statement in purified distance smoothing.

Consider the map $\mathcal{E}(\omega_B) := \omega_B \otimes \mathbbm{1}_C$. By the data-processing inequality \cite[Proposition 4.5.1]{Furrer_2012} for $\mathcal{E}: \mathcal{M}_{\overline{C}}\rightarrow \mathcal{M}_{\overline{B}} $ and $\omega \in \mathcal{D}_{\leq}(\mathcal{M}_{AB})$, where $\mathcal{M}$ stands for a von Neumann algebra and $\mathcal{E}^*$ denotes the dual map of $\mathcal{E}$, we obtain
\begin{equation}
    H_{\mathrm{min(PD)}}^{\epsilon}(A|\overline{B})_{\omega} \leq H_{ \mathrm{min(PD)}}^{\epsilon}(A|\overline{C})_{\mathrm{id}_A \otimes \mathcal{E}^{*}(\omega)}.
\end{equation}
Letting $\mathcal{M}_{\overline{B}} := \mathcal{B}(\mathcal{H}_{B}) \otimes \mathcal{B}(\mathcal{H}_C)$ and $\mathcal{M}_{\overline{C}} = \mathcal{B}(\mathcal{H}_B)$ and $\omega = \rho$, we obtain
\begin{equation}
    H_{\mathrm{min(PD)}}^{\epsilon}(A|BC)_{\rho} \leq H_{\mathrm{min(PD)}}^{\epsilon}(A|B)_{\mathrm{id}_{AB} \otimes \mathrm{Tr}_C[\rho]} = H_{\mathrm{min(PD)}}^{\epsilon}(A|B)_{\rho_{AB}}.
\end{equation}
This completes the proof in the purified distance. 
\end{proof}

\begin{lemma}[\textbf{Conditioning on classical register}]\label{lemma:CondClassRegister}
Let $\mathcal{H}_A$ and $\mathcal{H}_B$ be separable Hilbert spaces and $Z$ a classical register. Consider $\rho_{ABZ} \in \mathcal{D}_{\leq}(\mathcal{H}_A \otimes \mathcal{H}_B \otimes \ell_Z^{\infty})$.\\
Then we have
\begin{equation}\label{eq:CondClassReg1}
    H_{\mathrm{min(TD)}}^{\epsilon}(AB|Z)_{\rho} \geq \inf_{z \in (\lambda_z)_z} H_{\mathrm{min(TD)}}^{\epsilon}(A|B)_{\rho_{AB}^z}
\end{equation}
in trace distance and
\begin{equation}\label{eq:CondClassReg2}
    H_{\mathrm{min(PD)}}^{\epsilon}(AB|Z)_{\rho} \geq \inf_{z \in (\lambda_z)_z} H_{\mathrm{min(PD)}}^{\frac{\epsilon^2}{2}}(A|B)_{\rho_{AB}^z}
\end{equation}
in purified distance.
\end{lemma}

\begin{proof}
Since $Z$ is a classical register, $z$'s are mutually orthogonal. By the definition of the min-entropy (see Section~\ref{sec:SmoothMinEntropy}) we have $\forall z$
\begin{align*}
    &\lambda \Tr{\rho^z_{AB}} \sum_z \mathbbm{1}_A \otimes |z\rangle\langle z| - \sum_{z} \rho_{AB}^z \otimes |z\rangle \langle z| \geq 0\\
    \Leftrightarrow & \lambda \Tr{ \rho_{AB}^z} \cdot \mathbbm{1}_A  - \rho_{AB}^{z} \geq 0.
\end{align*}
Therefore, again recalling the definition of the min-entropy, we obtain
\begin{equation}\label{eq:condClassRegProof}
    H_{\mathrm{min(TD)}}(A|BZ) = \inf_{z} H_{\mathrm{min(TD)}}(A|B)_{\rho_{AB}^z}.
\end{equation}
Using the definition of smoothed min-entropies, we know that for every $\delta >0$ and for every $z\in Z$ there exists $\tilde{\rho}_{AB}^{z} \in \mathcal{B}^{\epsilon}_{TD}(\rho_{AB}^z)$ such that
\begin{align*}
    H_{\mathrm{min(TD)}}(\tilde{\rho}_{AB}^{z}||\rho_B^z) = \inf_{z} H_{\mathrm{min(TD)}}(\rho_{AB}^{z}||\rho_{B}^{z}) - \delta,
\end{align*}
for example, if we let $\tilde{\rho}_{AB}^{z}$ be the optimiser for the smooth min-entropy. Then, defining $\tilde{\rho}_{ABZ} := \sum_z \tilde{\rho}_{AB}^{z}$, we obtain from Eq.~(\ref{eq:condClassRegProof})
\begin{align*}
    H_{\mathrm{min(TD)}}(\tilde{\rho}_{ABZ}||\rho_{BZ}) = \inf_z H_{\mathrm{min(TD)}}(\tilde{\rho}_{AB}^{z} ||\rho_B^z) \geq H_{\mathrm{min(TD)}}^{\epsilon}(\rho_{AB}^{z} ||\rho_{B}^{z})-\delta.
\end{align*}
It remains to show that $\tilde{\rho}_{ABZ}$ is in the smoothing ball of $\rho_{ABZ}$. We use the trace distance ball, where $\tilde{\rho}_{ABZ}$ is guaranteed to be a subnormalized state. Therefore, following \cite{Renner_2005}, we first prove Eq. (\ref{eq:CondClassReg1})
\begin{equation*}
    \left|\left| \tilde{\rho}_{ABZ} - \rho_{ABZ} \right|\right|_1 = \inf_z \left|\left| \tilde{\rho}_{AB}^{z} - \rho_{AB}^{z} \right|\right|_1 \leq \sum_z \Tr{\rho_{AB}^{z}} \epsilon \leq \epsilon,
\end{equation*}
which concludes the proof in trace distance smoothing. Using Proposition \ref{prop:epsilonBalls}, we obtain Eq. (\ref{eq:CondClassReg2}).

\end{proof}

\begin{lemma}[\textbf{Removing a classical communication register}]\label{lemma:removingClassRegister}
Let $\mathcal{H}_C, \mathcal{H}_{E'}$ and $\mathcal{H}_X$ be separable Hilbert spaces and $\dim(\mathcal{H}_C) < \infty$ as well as $\dim(\mathcal{H}_{X}) < \infty$, where $X$ is the raw key and $C$ the transcript of the communication between Alice and Bob. Let $\rho \in \mathcal{D}(\mathcal{H}_X \otimes \mathcal{H}_{E'} \otimes \mathcal{H}_C)$ and let $\rho_{XE'} \in \mathcal{D}(\mathcal{H}_X \otimes \mathcal{H}_{E'})$ be the state after tracing out the register $C$.\\
Then for smoothing in trace distance,
\begin{equation*}
    H_{\mathrm{min(TD)}}^{\epsilon}(X|E'C)_{\rho} \geq H_{\mathrm{min(TD)}}^{\epsilon}(X|E')_{\rho} - \mathrm{leak}_{\mathrm{EC}}.
\end{equation*}
\end{lemma}
\begin{proof}
This proof follows closely the proof of \cite[Lemma 2]{Scarani_2008}. We define $Y$ to be the other party's local information used during information reconciliation and start with the left-hand side of the statement,
\begin{align*}
    H_{\mathrm{min(TD)}}^{\epsilon}(X|E'C)_{\rho} &\geq H_{\mathrm{min(TD)}}^{\epsilon}(XC|E')_{\rho} - \log_2(|C|)\\
    & \geq H_{\mathrm{min(TD)}}^{\epsilon}(X|E')_{\rho} + H_{\mathrm{min(TD)}}(C|XE')_{\rho} - \log_2(|C|) \\
    & \geq H_{\mathrm{min(TD)}}^{\epsilon}(X|E')_{\rho} +H_{\mathrm{min(TD)}}(C|XYE')_{\rho} - \log_2(|C|) \\
    & \geq H_{\mathrm{min(TD)}}^{\epsilon}(X|E')_{\rho} +H_{\mathrm{min(TD)}}(C|XYE')_{\rho} - \log_2(|C|).
\end{align*}
The first inequality follows from the chain rule for smooth-min entropies (Lemma \ref{lemma:chainRule}) and the second inequality is an extension of \cite[Lemma 3.2.10]{Renner_2005} for an infinite dimensional register $C \rightarrow E'$. We remark that proving this extension requires extending the min-entropy part of \cite[Lemma 3.1.8]{Renner_2005} which we have done in Lemma \ref{lemma:CondClassRegister} and \cite[Lemma 3.1.1]{Renner_2005} where the proof for the infinite dimensional case is identical to the proof given there. The third line is obtained by the strong subadditivity property of the smooth min-entropy (Lemma \ref{lemma:strongSubadditivity}) and the last inequality comes from the fact that $E' \leftrightarrow (X,E') \leftrightarrow C$ forms a Markov-chain since $C$ is computed by Alice and Bob as a function of $XY$. Finally, since $\log_2(|C|)$ stands for the number of all possible information-reconciliation transcripts, we may replace it with the actual leakage $\mathrm{leak}_{\mathrm{EC}}$ giving the number of bits needed to implement the used information-reconciliation scheme.
\end{proof}

\section{Generalisation of the asymptotic equipartition property}\label{apdx:GeneralisationAEP}

In this appendix, we generalise the asymptotic equipartition property \cite[Corollary 3.3.7]{Renner_2005} to infinite dimensions. The proof there requires an ordering on the eigenvalues as well as the Birkhoff-von Neumann theorem, so it needs some care to generalise the AEP statement to infinite dimensions. We note that the fully quantum asymptotic equipartition property was extended to infinite dimensions in Refs. \cite{Furrer_2012,Furrer_2011, Khartri_2019, Fawzi_2022}. However, as noted in Ref. \cite{George_2020} this version is harder to apply numerically. The basic idea of our proof relies on the fact that the infinite dimensional min-entropy can be converged via projections \cite{Furrer_2011}. Before we come to the actual proof, it requires some preparations.

We start by extending the definition of the max-relative entropy to infinite dimensions.

\begin{definition}[\textbf{infinite dimensional max-relative entropy}]\label{def:maxRelEntr}
Let $\mathcal{H}_A$ be a Hilbert space, and let $P,Q \in \mathrm{Pos}(\mathcal{H}_A)$. Then the max-relative entropy is defined by 
\begin{equation*}\label{eq:maxRelEntropy}
    D_{\mathrm{max}}(P||Q) = \inf\{\lambda:~ P \leq 2^\lambda Q\}.
\end{equation*}
\end{definition}

Next, we prove that $D_{\mathrm{max}}$ is a Rényi divergence just as in finite dimensions.
\begin{proposition}
For the max-relative entropy, as defined in Eq.~(\ref{eq:maxRelEntropy}) the following statement holds
\begin{itemize}
    \item[(1)] Normalisation: $D_{\mathrm{max}}(aP||bQ) = D_{\mathrm{max}}(P||Q) + \log_2(a) - \log_2(b)$
    \item[(2)] Dominance: For $P,Q,Q' \in\mathrm{Pos}(\mathcal{H}_A)$ and $Q\leq Q'$, we have $D_{\mathrm{max}}(P||Q) \geq D_{\mathrm{max}}(P||Q').$
\end{itemize}
\end{proposition}
\begin{proof}
We prove the two points separately.
\begin{itemize}
    \item[(1)] Let $\lambda^* := D_{\mathrm{max}}(P||Q)$ and $\lambda := D_{\mathrm{max}}(aP||bQ)$. We show two directions.
    \begin{enumerate}
        \item[$\geq$] Using the definition of the max-relative entropy yields $aP \leq 2^{\lambda} bQ$, which implies that $P \leq 2^{\lambda} \frac{b}{a}Q$. According to the definition, $\lambda^*$ is the infimum of all $\mu$ such that $P \leq 2^{\mu}Q$, hence $2^{\lambda^*} \leq 2^{\lambda} \frac{b}{a}$. Taking the logarithm and rearranging yields $\lambda \geq \lambda^* +\log_2(a) - \log_2(b)$, which concludes the first direction.
    \item[$\leq$] Using the max-relative entropy yields $P\leq 2^{\lambda^*}Q$. This is equivalent to $aP \leq 2^{\lambda^*} \frac{a}{b} bQ$. According to definition, $\lambda$ is the infimum of all $\mu$ such that $aP \leq 2^{\mu} bQ$; hence $2^{\lambda} \leq 2^{\lambda^*} \frac{a}{b}$. Taking the logarithm and rearranging yields $\lambda \leq \lambda^* + \log_2(a) - \log_2(b)$.
    \end{enumerate}

    \item[(2)] Again, for $\lambda^* := D_{\mathrm{max}}(P||Q)$, we have $P \leq 2^{\lambda^*} Q$. Since $Q \leq Q'$, we have $P \leq 2^{\lambda^*} Q \leq 2^{\lambda^*} Q'$. So, $\lambda^*$ is feasible for $D_{\mathrm{max}}(P||Q')$. Hence, it is an upper bound. This proves the claim.
\end{itemize}
\end{proof}

Defining $H_{\mathrm{min}}(\rho_{AB}||\sigma_B) = -D_{\mathrm{max}}(\rho_{AB} || \mathbbm{1}_A \otimes \sigma_B)$ gives us the following corollary.

\begin{corollary}\label{cor:NormAndDom}
    Let $\rho \in \mathrm{Pos}(\mathcal{H}_A \otimes \mathcal{H}_B)$ and $\sigma, \sigma' \in\mathrm{Pos}(\mathcal{H}_A)$ such that $\sigma \leq \sigma'$.\\
    Then the following statements hold:
    \begin{itemize}
        \item[(1)] Normalisation: $H_{\mathrm{min}}(a \rho ||b\sigma) = H_{\mathrm{min}}(\rho||\sigma) - \log_2(a) + \log_2(b)$

        \item[(2)] Dominance: $H_{\mathrm{min}}(\rho||\sigma) \leq H_{\mathrm{min}}(\rho|| \sigma')$.
    \end{itemize}
\end{corollary}

\begin{definition}
    Let $\mathcal{H}_A$ and $\mathcal{H}_B$ be separable Hilbert spaces, and let $\rho\in\mathrm{Pos}(\mathcal{H}_A \otimes \mathcal{H}_B)$ as well as $\sigma \in\mathrm{Pos}(\mathcal{H}_B)$. \\

    For $\epsilon \in (0, \sqrt{\Tr{\rho}})$ the smooth min-entropy is given by
    \begin{align*}
        H^{\epsilon}_{\mathrm{min(TD)}}(\rho||\sigma) := \sup_{\tilde{\rho} \in \mathcal{B}_{\mathrm{TD}}^{\epsilon}(\rho)} H_{\mathrm{min(TD)}}(\tilde{\rho} ||\sigma).
    \end{align*}
\end{definition}
Note that this coincides with the definition given in the main text (Section \ref{sec:Notation}). Next, we want to generalise \cite[Lemma 2]{Furrer_2011}. Therefore, we introduce sequences of projectors $\{\Pi^k\}_{k \in \mathbb{N}}$ onto finite-dimensional subspaces $U \subseteq \mathcal{H}$ of the relevant Hilbert space $\mathcal{H}$, that converge to the identity $\mathbbm{1}_{\mathcal{H}}$ with respect to $||\cdot||_1$. Then we define a sequence of non-normalised projected states as $\hat{\rho}^k := \Pi^k \hat{\rho} \Pi^k$. For a more detailed description, we refer the reader to \cite[Section II]{Furrer_2011}. We note that the following could be trivially further generalised to a continuity claim for the smoothed max-relative entropy.

\begin{lemma}\label{lemma:AEPhelpLemma}
    Let $\rho_{B} \in \mathcal{D}(\mathcal{H}_A\otimes \mathcal{H}_B)$ and let $\{ \hat{\rho}_{AB}^{k}\}_{k=1}^{\infty}$ a sequence of normalised projected states converging to $\rho_{AB}$ in the $||\cdot||_1$-norm. Let $\sigma_{B} \in\mathcal{D}(\mathcal{H}_B)$ and $\{ \hat{\sigma}_{B}^{k}\}_{k=1}^{\infty}$ be a sequence of normalised projected states that converge to $\sigma_B$.\\
    For any fixed $t\in(0,1)$ there exists $k_0 \in \mathbb{N}$ such that $\forall k\geq k_0$ we have
    \begin{align*}
        H_{\mathrm{min(TD)}}^{\epsilon}(\rho_{B}||\sigma_B) \geq H_{\mathrm{min(TD)}}^{t\epsilon}(\hat{\rho}_{AB}^k || \hat{\sigma}_B^k) + \log_2 \left( \mathrm{Tr}\left[\Pi_B^k \sigma \Pi_B^k\right] \right).
    \end{align*}
\end{lemma}
\begin{proof}
For fixed $\sigma$ the statement can be established by showing $\forall k\geq k_0:~\mathcal{B}_{\mathrm{TD}}^{t \epsilon}\left( \hat{\rho}_{AB}^k\right) \subseteq \mathcal{B}_{\mathrm{TD}}^{\epsilon}(\rho_{AB})$, where the proof is then identical to the proof of \cite[Lemma 2]{Furrer_2011}. Therefore, we take this result as established, so $\exists k_0$ such that \begin{equation}\label{eq:AEPprf1}
  \forall k\geq k_0:~ H_{\mathrm{min(TD)}}^{\epsilon}(\rho_{AB}||\sigma) \geq H_{\mathrm{min(TD)}}^{\epsilon}(\hat{\rho}_{AB}^k ||\sigma).  
\end{equation}
We are using this result and Corollary \ref{cor:NormAndDom} to prove the general case. We deduce
\begin{align*}
    H_{\mathrm{min(TD)}}^{t\epsilon}(\hat{\rho}_{AB}^{k} || \hat{\sigma}_B^k) = H_{\mathrm{min(TD)}}^{t\epsilon}\left(\hat{\rho}_{AB}^{k} \left|\left| \frac{\sigma_B^k}{\Tr{\Pi_B^k \sigma \Pi_B^k}} \right.\right.\right) = H_{\mathrm{min(TD)}}^{t\epsilon}\left(\hat{\rho}_{AB}^{k} || \sigma_B^k \right) + \log_2\left( \frac{1}{\Tr{\Pi_B^k \sigma \Pi_B^k}}\right),
\end{align*}
where we applied the normalisation property in Corollary \ref{cor:NormAndDom} for the second equality. Then, using the dominance property in Corollary \ref{cor:NormAndDom} and noting that $\sigma \geq \Pi_B^k \sigma \Pi_B^k = \sigma^k$, we obtain
\begin{align*}
    H_{\mathrm{min(TD)}}^{t\epsilon}\left(\hat{\rho}_{AB}^{k} || \sigma_B^k \right) \leq H_{\mathrm{min(TD)}}^{t\epsilon}\left(\hat{\rho}_{AB}^{k} || \sigma_B\right).
\end{align*}
Putting things together, we showed
\begin{align*}
 H_{\mathrm{min(TD)}}^{t\epsilon}(\hat{\rho}_{AB}^{k} || \hat{\sigma}_B^k) \leq    H_{\mathrm{min(TD)}}^{t\epsilon}\left(\hat{\rho}_{AB}^{k} || \sigma_B\right) - \log_2\left( \Tr{\Pi_B^k \sigma \Pi_B^k}\right).
\end{align*}
Thanks to Eq.~(\ref{eq:AEPprf1}) we know already that there exists such a $k_0$ to bound $H_{\mathrm{min(TD)}}^{\epsilon}(\hat{\rho}_{AB}^k ||\sigma)$. This completes the proof.
\end{proof}

In the next Lemma, we extend Renner's AEP \cite[Theorem 3.3.6]{Renner_2005} to infinite dimensional side information. Note that we cannot generalise register $A$ to infinite dimensions, as the correction term is a function of the dimension of this register. However, this generalisation is not required for QKD anyways.

\begin{lemma}
Let $\mathcal{H}_A$ and $\mathcal{H}_B$ be separable Hilbert spaces, where $\mathcal{H}_A$ is finite-dimensional, $\mathrm{dim}\left(\mathcal{H}_A\right) < \infty$. Let $\rho_{AB} \in\mathcal{D}(\mathcal{H}_A \otimes \mathcal{H}_B)$ and $n \in \mathbb{N}$.\\
Then, for any $\epsilon \in (0,1)$
\begin{equation*}
    \frac{1}{n} H_{\mathrm{min(TD)}}^{\epsilon}(\rho_{AB}^{\otimes n}||\sigma_B^{\otimes n}) \geq H(AB)_{\rho} - H(B)_{\rho} - D(\rho_B||\sigma_B) - \delta,
\end{equation*}
where $\delta = 2 \log_2\left( \mathrm{rank}(\rho_A) + \Tr{\rho_{B}^2 \left( \mathbbm{1}_A \otimes \sigma_B^{-1}) +2 \right)}  \right) \sqrt{\frac{\log_2\left(\frac{1}{\epsilon}\right)}{n} +1}$. In terms of purified distance, we replace $\epsilon \mapsto \sqrt{\epsilon}$.    
\end{lemma}
\begin{proof}
We follow the proof of \cite[Proposition 8]{Furrer_2011}. Let $\left(\Pi_A^k, \Pi_B^k \right)$ be sequences of projectors such that $\forall k' \geq k:~\Pi_A^k \leq \Pi_A^{k'}$ that converges to the identity in the weak operator topology and similarly for the projectors in $B$. Then, the $n$-fold projectors $\left( \left( \Pi_A^k\right)^{\otimes n}, \left( \Pi_B^k\right)^{\otimes n} \right)$ satisfy these conditions as well. 

Fix $t \in (0,1)$. Then, by Lemma \ref{lemma:AEPhelpLemma} there $\exists k_0 \in \mathbb{N}$ such that $\forall k \geq k_0$
\begin{equation*}
    H_{\mathrm{min(TD)}}^{\epsilon}\left( \rho_{AB}^{\otimes n} \right|\!\left| \sigma_{B}^{\otimes n}\right) \geq H_{\mathrm{min(TD)}}^{t \epsilon}\left( \left(\hat{\rho}_{AB}\right)^{\otimes n} \right|\!\left| \left(\hat{\sigma}_{B}\right)^{\otimes n}\right) - n \log_2\left( \Tr{\Pi^k \sigma \Pi^k} \right)
\end{equation*}
holds. We used that the trace is multiplicative over tensor products. Next, since we are working on projections, we can apply \cite[Theorem 3.3.6]{Renner_2005} and obtain
\begin{equation}
    \frac{1}{n} H_{\mathrm{min(TD)}}^{\epsilon}\left( \rho_{AB}^{\otimes n} \right|\!\left| \sigma_{B}^{\otimes n}\right)  \geq H\left(\hat{\rho}_{AB}^k \right) - H\left(\hat{\rho}_{B}^k \right) - D\left(\hat{\rho}_{B}^k \right|\!\left|\hat{\sigma}_{B}^k \right) - \delta(t \epsilon) - \log_2\left( \Tr{\Pi^k \sigma \Pi^k} \right).
\end{equation}
When we take the limit of $k\rightarrow \infty$ the left-hand side doesn't change, while the right-hand side, by our assumptions on the projections, recovers the true states. The $\log_2$-term drops, as $\log_2\left( \Tr{\sigma}\right) = \log_2(1) = 0$. Hence, we obtain
\begin{equation*}
  \frac{1}{n} H_{\mathrm{min(TD)}}^{\epsilon}\left( \rho_{AB}^{\otimes n} \right|\!\left| \sigma_{B}^{\otimes n}\right)  \geq H\left(\rho_{AB} \right) - H\left(\rho_{B} \right) - D\left(\rho_{B} \right|\!\left|\sigma_{B} \right) - \delta(t \epsilon). 
\end{equation*}
Finally, taking the limit $t\rightarrow 1$ completes the proof.
\end{proof}

We obtain the final result of this section, the generalised Asymptotic Equipartition Property, as a corollary.

\begin{corollary}[\textbf{Asymptotic Equipartition Property}]\label{cor:AEP}
Let $\mathcal{H}_X$ and $\mathcal{H}_E$ be separable Hilbert spaces, where $\mathcal{H}_X$ is finite-dimensional. Let $\rho_{XE}$ be a classical-quantum state.\\
Then, for smoothing in terms of trace distance
\begin{equation*}
    \frac{1}{n} H_{\mathrm{min(TD)}}^{\epsilon}(X|E)_{\rho_{XE}^{\otimes n}} \geq H(X|E) - \delta(\epsilon),
\end{equation*}
where $\delta(\epsilon) := 2 \log_2(\mathrm{rank}(\rho_X)+3) \sqrt{\frac{\log_2(2/\epsilon)}{n}}$. For smoothing in terms of purified distance every $\epsilon$ needs to be replaced by $\sqrt{\epsilon}$.
\end{corollary}

\begin{proof}
The proof is now identical to that of Corollary 3.3.7 of Ref. \cite{Renner_2005}, where we omit the simplifications at the end of the proof. The purified distance bound can be obtained by Proposition \ref{prop:epsilonBalls}.
\end{proof}

\section{Derivation of the finite-dimensional optimisation problem} \label{APDX:FinDimOptProb}
In this section, we are motivating and deriving the primal and dual SDP we have to solve in order to obtain a lower bound on the secure key rate. Our starting point is the infinite dimensional optimisation problem, given in Eq.~(\ref{eq:HighDimOptProb}), which we obtain based on Bob's observations. By introducing slack variables, the inequality constraints can be turned into equality constraints.

\begin{align*}
   \begin{aligned}
    \min~&f(\rho)\\
    \text{subject to }& \\
    & \mathrm{Tr}_{B}\left[\rho\right] = \rho_A\\
     & \left| \Tr{\hat{\Gamma}_j\rho } - \gamma_j \right| \leq \mu_j\\
    & \Tr{\rho} = 1\\
    & \rho \geq 0    
\end{aligned} 
\hspace{5mm }\Leftrightarrow  \hspace{5mm}
\begin{aligned}
    \min~& f(\rho)\\
    \text{subject to }& \\
    & \mathrm{Tr}_{B}\left[\rho\right] = \rho_A\\
     & \Tr{\hat{\Gamma}_j\rho } \leq \mu_j +  \gamma_j \\
     & -\Tr{\hat{\Gamma}_j\rho } \leq 
     \mu_j -  \gamma_j \\
    & \Tr{\rho} = 1\\
    & \rho \geq 0
\end{aligned}
\end{align*}
Next, we apply the dimension reduction method \cite{Upadhyaya_2021} and obtain the expanded finite-dimensional optimisation
\begin{align*}
\begin{aligned}
    \min~& f(\overline{\rho})\\
    \text{subject to }& \\
    & \frac{1}{2} \left|\left|\mathrm{Tr}_{B}\left[\bar{\rho}\right] - \rho_A \right|\right|_1 \leq \sqrt{w}\\
     & \mu_j + \gamma_j - w \left|\left| \hat{\Gamma}_j \right|\right|_{\infty} \leq \Tr{\hat{\Gamma}_j \overline{\rho}}  \leq \mu_j +  \gamma_j\\
    & 1-w \leq \Tr{\overline{\rho}} \leq 1\\
    & \overline{\rho} \geq 0
\end{aligned}
\end{align*}
where we replaced the infinite dimensional $\rho$ by the finite-dimensional $\overline{\rho}$ and used the improved bound $\sqrt{w}$ for the trace-norm constraint from \cite[page 59]{Upadhyaya_Thesis_2021}. Furthermore, we can rewrite the trace-norm constraint (see, for example,~\cite{Watrous_2018}). We obtain
\begin{align}\label{eq:primalOptAPDX}
\begin{aligned}
    \min~& f(\overline{\rho})\\
    \text{subject to }& \\
    & \Tr{P}+\Tr{N} \leq 2 \sqrt{w}\\
    & P \geq \mathrm{Tr}_B\left[ \overline{\rho} \right] - \rho_A\\
    & N \geq -\left( \mathrm{Tr}_B\left[ \overline{\rho} \right] - \rho_A \right)\\
    & \Tr{\hat{\Gamma}_j \overline{\rho}}  \leq \mu_j +  \gamma_j\\
    & \Tr{-\hat{\Gamma}_j \overline{\rho}}  \leq \mu_j -  \gamma_j + w \left|\left| \hat{\Gamma}_j \right|\right|_{\infty}\\
    & 1-w \leq \Tr{\overline{\rho}} \leq 1\\
    & \overline{\rho}, N, P \geq 0
\end{aligned}
\end{align}

The numerical method in \cite{Winick_2018} lower bounds the minimum of the objective function as follows. Let $\overline{\rho}^*$ minimise $f$ over the feasible set $\mathcal{S}$. Then, we have
\begin{align}
f(\overline{\rho}^*) &\geq f(\overline{\rho}) + \Tr{(\overline{\rho}^* - \overline{\rho}) \nabla f(\overline{\rho}} \geq f(\overline{\rho}) + \min_{\sigma \in \mathcal{S}} \Tr{(\sigma-\overline{\rho}) \nabla f(\overline{\rho}} \\
&= f(\overline{\rho}) - \Tr{\overline{\rho} \nabla f(\overline{\rho}}) - \min_{\sigma \in \mathcal{S}} \Tr{\sigma \nabla f(\overline{\rho})}.
\end{align}
Therefore, in what follows, we consider this linearised problem. The feasible set is given by the constraints in Eq.~(\ref{eq:primalOptAPDX}). Furthermore, for ease of notation, we denote all measurement operators by the label $\hat{\Lambda}$ and call the right-hand sides of the constraints related to measurements and the trace-condition $\lambda_j$ to obtain a more abstract form of our optimisation problem. Then, the problem reads
\begin{align}
\begin{aligned}
    \min~& \langle \nabla f(\overline{\rho}), \sigma \rangle\\
    \text{subject to }& \\
    &\lambda_k - \Tr{\hat{\Lambda}_k \sigma} \geq 0\\
    &2 \sqrt{w} - \Tr{P}-\Tr{N} \geq 0 \\
    &P - \mathrm{Tr}_B\left[ \sigma \right]+ \rho_A \geq 0 \\
    &N + \mathrm{Tr}_B\left[ \overline{\rho} \right] - \rho_A \geq 0 \\
    &\overline{\rho}, N, P \geq 0
\end{aligned}
\end{align}

The standard form of a semi-definite program is 
\begin{itemize}
    \item[(P)] Primal problem: 
\begin{align*}
    \alpha := \inf~&\langle X, H_1\rangle_{\mathcal{H}_1}\\
    \text{subject to }&\\
    & \mathcal{N}(X) - H_2 \in \mathcal{K}_2\\
    & X \in \mathcal{K}_1
\end{align*}
\item[(D)] Dual problem:
\begin{align*}
    \beta := \sup~&\langle Y, H_2\rangle_{\mathcal{H}_2}\\
    \text{subject to }&\\
    &  H_1 - \mathcal{N}^{*}(X) \in \mathcal{K}_1^*\\
    & Y \in \mathcal{K}_2^*
\end{align*}
\end{itemize}

Note that $\mathcal{K}_1$ denotes the cone 
\begin{equation*}
    \mathcal{K}_1 := \left\{\begin{pmatrix}
    x_1 \\ x_2 \\ x_3
\end{pmatrix}: x_1, x_2, x_3 \in \mathcal{H}_1 \land x_1, x_2, x_3 \geq 0 \right\},
\end{equation*}
where $\mathcal{K}_1^*$ denotes the dual cone of $\mathcal{K}_1$ and $\langle \cdot, \cdot \rangle_{\mathcal{H}_1}$ and $\langle \cdot, \cdot \rangle_{\mathcal{H}_2}$ are the inner products on the Hilbert spaces $\mathcal{H}_1$ and $\mathcal{H}_2$, where the optimisation problems are set. In our case, we have $\mathcal{K}_1^* = \mathcal{K}_1$ and the first inner product is the Hilbert-Schmidt inner product over the Hilbert space of bounded linear operators and the second inner product is the inner product induced by the component-wise inner products of Hilbert spaces of the constituents of $Y$. $H_1, H_2$ and $\mathcal{N}$ are known, while $X$ is the primal optimisation variable, while $Y$ is the optimisation variable in the dual problem.

For the present problem, we identify 
\begin{align*}
 X &= \left(\sigma \oplus P \oplus N  \right)  \\
 H_1 &= \left( \nabla f(\overline{\rho}) \oplus 0 \oplus 0 \right)\\
 H_2 &= -\left(-\lambda_1 \oplus \hdots \oplus \lambda_{6 N_{\mathrm{St}}} \oplus 2 \sqrt{w} \oplus \rho_A \oplus -\rho_A  \right)\\
 Y &= \left( y_1 \oplus \hdots \oplus y_{6 N_{\mathrm{St}}} \oplus s \oplus \tau \oplus \Theta \right),
\end{align*}
where we interpret scalars as $1\times1$ matrices, as well as the linear map
\begin{align*}
    \mathcal{N}(X) &= \mathcal{N}(X_1 \oplus X_2 \oplus X_3) \\&
    = \left( -\Tr{X_1 \hat{\Lambda}_1} \oplus \hdots \oplus -\Tr{X_1 \hat{\Lambda}_{6 N_{\mathrm{St}}}} \oplus - \Tr{X_2} \oplus - \Tr{X_3} \oplus X_2 - \mathrm{Tr}_B\left[ X_1\right] \oplus X_3 - \mathrm{Tr}_B\left[ X_1\right] \right).
\end{align*}

It remains to find the dual (adjoint) of $\mathcal{N}$, defined by $\langle Y, \mathcal{N}(X) \rangle_{\mathcal{H}_2} = \langle \mathcal{N}^*(Y), X \rangle_{\mathcal{H}_1}$. One can show that
\begin{align*}
    \mathcal{N}^*\left( y_1 \oplus \hdots \oplus y_{6 N_{\mathrm{St}}} \oplus s \oplus \tau \oplus \Theta \right) = \left( \left( -  \sum_{j=1}^{6N_{\mathrm{St}}} y_j \hat{\Lambda}_j - \tau \otimes I_B - \Theta \otimes I_B \right)\oplus \left(-s \cdot I + \tau \right) \oplus \left(- s \cdot I + \Theta \right) \right).
\end{align*}

Therefore, the dual problem reads
\begin{align*}
    - \max ~&\vec{y}\cdot \vec{\lambda} + 2\sqrt{w} s + \Tr{\rho_A \tau} - \Tr{\rho_A \Theta}\\
    \text{subject to }&\\
    &\nabla f(\overline{\rho}) + \sum_{j=1}^{6 N_{\mathrm{St}}} y_j \hat{\Lambda}_j + \tau \otimes I_B - \Theta \otimes I_B \geq 0\\
    &s \cdot I - \tau \geq 0\\
    & s \cdot I + \Theta \geq 0\\
    & \vec{y} \geq 0, ~s \geq 0,~\tau, \Theta \geq 0.
\end{align*}

Finally, we apply the relaxation in \cite{Winick_2018} to take numerical imprecisions into account. This adds $\epsilon_{\mathrm{num}}$ to the vector $\vec{v}$ as well as to $2 \sqrt{w}$. Therefore, as claimed, we finally obtain the dual.

\section{Completeness}\label{apdx:Completeness}

\begin{proof}[Proof of Proposition \ref{prop:completeness-bound}]
    First, we show why the completeness may be decomposed into multiple epsilon terms. This has also been explained in other work \cite{Arnon_Friedman_2020}. By definition of completeness of a QKD protocol (Definition \ref{def:completeness}),
    \begin{align*}
         \Pr[\mathsf{Abort}|\mathsf{Honest}] 
        =& \Pr[\mathsf{ET\,Abort} \cup \mathsf{AT\,Abort} \cup \mathsf{EC\,Abort}|\mathsf{Honest}] \\
        \leq & \Pr[\mathsf{ET\,Abort}|\mathsf{Honest}] + \Pr[\mathsf{AT\,Abort}|\mathsf{Honest}] + \Pr[\mathsf{EC\,Abort}|\mathsf{Honest}] \ ,
    \end{align*}
    where we have used that energy test, acceptance test, and error correction are the steps in the protocol which might abort and then applied the union bound. We take $\mathsf{Honest}$ to mean the input at the step conditioned on the previous inputs passing on the honest input. We take the honest input to be the state $\sigma^{\otimes n} := (\mathcal{E}_{\mathrm{Honest}}(\rho))^{\otimes n}$, where $\rho$ is the state the devices effectively prepare and $\mathcal{E}_{\mathrm{Honest}}$ is the assumed memoryless noisy channel when there is no eavesdropper. We can then take each of these conditional probabilities and define a notion of $\epsilon-$completeness for these subprotocols in the same manner as for the whole protocol (Definition \ref{def:completeness}). The completeness of error correction is a choice of error correcting code, so we leave this as an input parameter of the protocol, $\epsilon^{c}_{\mathrm{EC}}$. Thus we are only interested in bounding the other two probabilities. 
    
    For the energy test, it's very similar to what was done in Appendix \ref{apdx:proof_ET}, so we follow the notation from that section. First,
    \begin{align*}
        \Pr[\mathsf{ET\,Abort}|\mathsf{Honest}] = \Pr[|\{Y_{i}:Y_{i} \leq \beta_{T}\}| > l_{T} | \sigma] = \sum_{j = l_{T} + 1}^{k_{T}} \Pr[f_{k_{T}} = P_{j}|\sigma] \leq \sum_{j = l_{T} + 1}^{k_{T}} 2^{-k_{T} D(P_{j}||Q_{\sigma})} \ ,
    \end{align*}
    where we first used the definition of when the protocol aborts, then decomposed it into the types, and finally applied \eqref{eq:type-prob-bound} using $Q_{\sigma} := \begin{pmatrix} 1 - \Tr{\sigma V_{1}} \\ \Tr{\sigma V_{1}} \end{pmatrix}$. Finally, this sum may be tedious to calculate, so we make an assumption to simplify the calculation. We assume that $1 - \Tr{V_{1}\sigma} < \frac{l_{T}+1}{k_{T}}$. This means that every $j > l_{T}+1$ term can only lead to a larger relative divergence term than the $l_{T}+1$ term. It follows that we have
    \begin{align*}
        \Pr[\mathsf{ET\,Abort}|\mathsf{Honest}] \leq \sum_{j = l_{T} + 1}^{k_{T}} 2^{-k_{T} D(P_{j}||Q_{\sigma})} \leq (k_{T}-l_{T}-1) 2^{-k_{T}D(P_{l_{T}+1}||Q_{\sigma})} =: \epsilon^{c}_{\mathrm{ET}} \ .
    \end{align*}
    This completes the explanation for the energy test.

    For the acceptance test,
    \begin{align*}
        \Pr[\mathsf{AT\,Abort}|\mathsf{Honest}] = \Pr[\exists X \in \Theta: |v_{X} - r_{X}| > t_{X} | \sigma] \leq \sum_{X \in \Theta} \Pr[|v_{X} - r_{X}| > t_{X} | \sigma] \leq 2 \sum_{X \in \Theta} e^{-2m_{X} \frac{t_{X}^{2}}{4\|X\|^{2}_{\infty}}} =: \epsilon^{c}_{\mathrm{AT}} \ ,
    \end{align*}
    where we used the definition of the accepted observations ($\mathcal{O}$ in Theorem \ref{Thm:ParameterEstimation}), the union bound, and then Hoeffding's inequality. Combining these terms completes the proof.
\end{proof}

\twocolumngrid
\bibliography{Bibliography}
	
\end{document}